\documentclass[amsmath, amsthm, amssymb, floatfix, notitlepage, nofootinbib, twocolumn,superscriptaddress, 10pt]{revtex4-2}

\usepackage[utf8]{inputenc}
\usepackage{amsthm}
\usepackage[inline]{enumitem}
\makeatletter
\newtheorem*{rep@theorem}{\rep@title}
\newcommand{\newreptheorem}[2]{%
\newenvironment{rep#1}[1]{%
 \def\rep@title{#2 \ref*{##1}, repeated}%
 \begin{rep@theorem}}%
 {\end{rep@theorem}}}
\makeatother

\newtheorem{lemma}{Lemma}

\newtheorem{definition}{Definition}

\newtheorem{proposition}{Proposition}
\newtheorem{assumption}{Assumption}
\newreptheorem{assumption}{Assumption}

\usepackage{amsmath}
\usepackage{amssymb}
\usepackage{array}
\usepackage{braket}
\usepackage{color}
\usepackage{graphicx}
\usepackage{tikz}
\usepackage{lipsum}
\usepackage{colortbl}
\usepackage[normalem]{ulem}
\usetikzlibrary{quantikz}
\usepackage{hyperref}
\hypersetup{colorlinks,allcolors=black,linktoc=all}

\newcommand{\norm}[1]{\Vert{#1}\Vert}%
\newcommand{\op}{\textnormal{op}}
\newcommand{\mc}[1]{\mathcal{#1}}
\newcommand{\phdg}{\phantom{\dagger}}
\newcommand{\bignorm}[1]{\left\Vert{#1}\right\Vert}%
\newcommand{\abs}[1]{ | {#1} |}
\newcommand{\bigabs}[1]{ \left| {#1} \right|}

\newcolumntype{L}[1]{>{\raggedright\let\newline\\\arraybackslash}p{#1}}

\DeclareMathOperator{\tr}{Tr}
\DeclareMathOperator{\sign}{sign}

\begin{document}
\title{Quantum advantage and stability to errors in analogue quantum simulators}
\author{Rahul Trivedi}
\thanks{Both authors contributed equally.}
\affiliation{Max-Planck-Institut für Quantenoptik, Hans-Kopfermann-Str.~1, 85748 Garching, Germany.}
\affiliation{Munich Center for Quantum Science and Technology (MCQST), Schellingstr. 4, D-80799 Munich, Germany. }
\affiliation{Electrical and Computer Engineering, University of Washington, Seattle, WA -  98195, USA}
\author{Adrian Franco Rubio}
\thanks{Both authors contributed equally.}

\affiliation{Max-Planck-Institut für Quantenoptik, Hans-Kopfermann-Str.~1, 85748 Garching, Germany.}
\affiliation{Munich Center for Quantum Science and Technology (MCQST), Schellingstr. 4, D-80799 Munich, Germany. }
\author{J.~Ignacio Cirac}
\affiliation{Max-Planck-Institut für Quantenoptik, Hans-Kopfermann-Str.~1, 85748 Garching, Germany.}
\affiliation{Munich Center for Quantum Science and Technology (MCQST), Schellingstr. 4, D-80799 Munich, Germany. }

\date{\today}
\begin{abstract}
Several quantum hardware platforms, while being unable to perform fully fault-tolerant quantum computation, can still be operated as analogue quantum simulators for addressing many-body problems. However, due to the presence of errors, it is not clear to what extent those devices can provide us with an advantage with respect to classical computers. In this work we consider the use of noisy analogue quantum simulators for computing physically relevant properties of many-body systems both in  equilibrium and undergoing dynamics. We first formulate a system-size independent notion of stability against extensive errors, which we prove for Gaussian fermion models, as well as for a restricted class of spin systems. Remarkably, for the Gaussian fermion models, our analysis shows the stability of critical models which have long-range correlations. Furthermore, we analyze how this stability may lead to a quantum advantage, for the problem of computing the thermodynamic limit of many-body models, in the presence of a constant error rate and without any explicit error correction.
\end{abstract}
\maketitle

Quantum information processing systems hold the promise of solving a number of problems in physics and computer science faster than their classical counterparts \cite{jozsa2001quantum, montanaro2016quantum}. However, most quantum algorithms with theoretical performance guarantees require a fault-tolerant quantum computer \cite{aharonov1996limitations, aharonov1997fault, aharonov2006fault}. While in principle possible, implementing a fault tolerant quantum computer is a technological challenge that could still take a long time to solve. This has motivated several investigations trying to identify both quantum algorithms, as well as physically relevant computational problems, that can be addressed by quantum hardware in the near term and without any explicit error correction.

Analog quantum simulators, wherein a target Hamiltonian is mimicked by an experimentally controllable system, have shown some promise in solving problems arising in many-body physics in the near term \cite{altman2021quantum, daley2022practical, cirac2012goals}. A typical analogue quantum simulator, while not necessarily being able to perform an arbitrary computation, would instead aim to approximately implement a relevant spatially-local Hamiltonian, $H$. In several many-body problems, $H$ can additionally be taken to be translationally invariant. The quantum simulator can then be used to prepare a physically relevant quantum state $\rho_H$ associated with the Hamiltonian $H$, such as its ground state, Gibbs state or a state produced under dynamics. In a typical quantum simulation experiment, we would then measure the expectation value of an intensive observable $O$, which is either often a local observable at a single site on the lattice or a correlation function, i.e.~a product of local observables at a few sites. Examples of such Hamiltonians and observables can be found in a variety of problems in physics --- for e.g.~in study of correlated electronic systems \cite{zhou2021high, stormer1999fractional, broholm2020quantum, qin2022hubbard}, quantum spin systems \cite{spalek2007tj, mi2022time, chen2022error, pal2010many, vosk2015theory} as well as lattice-gauge theories \cite{rothe2012lattice, kogut1983lattice}. Furthermore, there have been several proposals to implement quantum simulation of these models in different hardware platforms, such as cold atoms in optical lattices, trapped ion systems or superconducting qubits \cite{esslinger2010fermi, houck2012chip, porras2004effective, blatt2012quantum, garcia2004implementation, zohar2015quantum, zohar2012simulating, zohar2013simulating, banuls2020simulating}.

Practically, quantum simulators offer several distinct advantages in solving many-body problems as opposed to general purpose quantum computers. First, quantum simulators aim to solve only a much smaller and specialized set of problems, and thus have much milder hardware requirements than a universal quantum computer. Furthermore, quantum simulators are more naturally suited to many-body problems, since they avoid a Hamiltonian-to-circuit mapping for e.g.~by trotterizing the evolution into a quantum circuit, which typically incurs in a rapid proliferation of errors \cite{cirac2012goals,stilck2021limitations, de2022limitations, gonzalez2022error, daley2022practical}. Furthermore, since the observables of interest are typically local intensive observables, we expect them to be somewhat more robust to errors even if the global quantum state of the simulator is very sensitive. These expectations make quantum simulators very promising in providing some advantage with respect to classical computers when addressing typical quantum many-body physics problems.

However, developing rigorous criteria to outline the quantum advantage of a quantum simulator runs into several theoretical issues. \emph{First}, quantum simulators do not implement any error correction and typically simulate many-body physics in the presence of noise. While several previous works have theoretically outlined the computational power of quantum simulators by developing the notion of a universal quantum simulator \cite{kraus2007universal,cubitt2018universal, zhou2021strongly} and rigorously established the possibility of quantum advantage without noise, the presence of experimentally realistic noise has to be carefully accounted for in understanding their utility in many-body problems. \emph{Second}, since quantum simulators are usually devoted to analyzing intensive observables, and in many-body physics we are typically interested in the thermodynamic limit of such observables, we need to revisit the usual notion of quantum advantage. In particular, instead of characterizing the quantum and classical effort required to compute the many-body observable as a function of the system size, which is not meaningful in the thermodynamic limit, we can characterize the effort required to compute the many-body intensive observable within a user-specific precision of the thermodynamic limit \cite{aharonov2022hamiltonian, watson2022computational}. 
In this paper, we address both of these issues --- we provide evidence that many physically relevant critical and non-critical many-body models are stable to errors in the quantum simulators. Importantly, even without error correction, we can use a quantum simulator to determine the thermodynamic limits of intensive observables in these problems to a hardware-limited precision. Furthermore, we also propose a notion of quantum advantage, in the presence of errors, for such problems, where the figure of merit is the computational time to obtain an intensive quantity in the thermodynamic limit to a hardware-limited precision. By providing explicit lower bounds on certifiable classical algorithms for the many-body problems that we consider, we provide evidence that quantum simulators can possibly provide superpolynomial to exponential quantum advantage over rigorous classical algorithms even without error correction.

\section{Stable quantum simulation tasks}\label{sec:stability}
\subsection{General setup}\label{sec:stability_setup}
To keep our analysis general, we will consider quantum simulators for solving both closed system (i.e.~implementing a Hamiltonian) as well as open system (i.e.~implementing a Lindbladian) many-body problems. Suppose that the quantum simulator was trying to configure a spatially local Lindbladian $\mathcal{L}$ given by
\begin{align}\label{eq:target_lind_main_text}
\mathcal{L} = \sum_\alpha \mathcal{L}_\alpha,
\end{align}
where $\mathcal{L}_\alpha$ is a Lindbladian acting on spins within a local region $\Lambda_\alpha$. For translationally invariant problems, the superoperator $\mathcal{L}_\alpha$ would additionally be independent of $\Lambda_\alpha$, and for closed system problems, we can assume $\mathcal{L}_\alpha= -i[h_\alpha, \cdot]$ for some operator $h_\alpha$ supported on $\Lambda_\alpha$. The quantum simulator would, in general, suffer from coherent errors in the configured Lindbladian as well as incoherent errors arising due to its interaction with an external environment. To account for these errors, we model the `implemented' Hamiltonian on the quantum simulator by
\begin{subequations}\label{eq:noisy_lind_main_text}
\begin{align}
\mathcal{L}'(t) = \sum_{\Lambda} \bigg( \mathcal{L}_\alpha'  -i [h_{\text{SE}, \alpha}(t), \cdot], \bigg).
\end{align}
Here $\mathcal{L}_\alpha' - \mathcal{L}_\alpha$ is an error in the Lindbladian implemented on spins in $\Lambda_\alpha$ --- this can arise either from configuration errors in the Hamiltonian and the jump operators corresponding to $\mathcal{L}_\alpha$, or from incoherent errors that can be well approximated as Markovian. Furthermore, we also consider the possibility of non-Markovian incoherent errors --- these are captured by $h_{\text{SE}, \alpha}(t)$, which accounts for the interaction of the spins in the region $\Lambda_\alpha$ with an external environment (in the interaction picture with respect to the environment). For concreteness, we will model the decohering environment as a Gaussian environment and assume that
\begin{align}
h_{\text{SE}, \alpha}(t) = \sum_{j = 1}^{n_L} A_{j, \alpha}(t) Q_{j, \alpha}^\dagger +\text{h.c.},
\end{align}
\end{subequations}
where $Q_{1, \alpha}, Q_{2, \alpha} \dots Q_{n_L, \alpha}$ are the jump operators, each supported on $\Lambda_\alpha$, through which the spins in $\Lambda_\alpha$ interact with a decohering environment, and $A_{1, \alpha}(t), A_{2, \alpha}(t) \dots A_{n_L, \alpha}(t)$ are annihilation operators for the environment. We assume that $[A_{j, \alpha}(t), A_{j', \alpha'}^\dagger(t')] = \delta_{\alpha, \alpha'}\delta_{j, j'} K_{j, \alpha}(t - t')$ for bosonic environments or $\{A_{j, \alpha}(t), A_{j', \alpha'}^\dagger(t')\} = \delta_{\alpha, \alpha'}\delta_{j, j'} K_{j, \alpha}(t - t')$ for fermionic environments and we choose the normalization of $Q_{j, \alpha}$ such that $\int_\mathbb{R} \abs{K_{j, \alpha}(\tau)} d\tau = 1$. The function $K_{j, \alpha}(\tau)$ can be understood as the memory kernel corresponding to the non-Markovian system-environment interaction. In particular and not unexpectedly, choosing $K_{j, \alpha}(\tau) = \delta(\tau)$ would yield a Markovian master equation for an environment initially in the vacuum state.

We introduce a parameter $\delta$ such that $\norm{\mathcal{L}_\alpha' - \mathcal{L}_\alpha}_\diamond \leq \delta$ and $\norm{Q_{j, \alpha}} \leq \sqrt{\delta}$ --- the parameter $\delta$ can be considered to be the ``hardware error rate'' in the quantum simulator. A well designed experimental setup can, in principle, achieve $\delta \ll 1$ --- however, since there are an extensive number of errors in the simulator, we generically expect the state of the simulator to be at a distance of $\delta \times n$ from the target state. In the worst case, this would imply that the results of the quantum simulator can only be trusted when $\delta < o(1/n)$, and this would limit their applicability to small-scale problems. Importantly, for applications of quantum simulators to problems in many-body physics, this would imply that noisy quantum simulators, in the worst case, cannot be used to faithfully capture thermodynamic limits (i.e.~$n\to \infty$).

\begin{figure*}[t]
    \centering
    \includegraphics[scale=0.6]{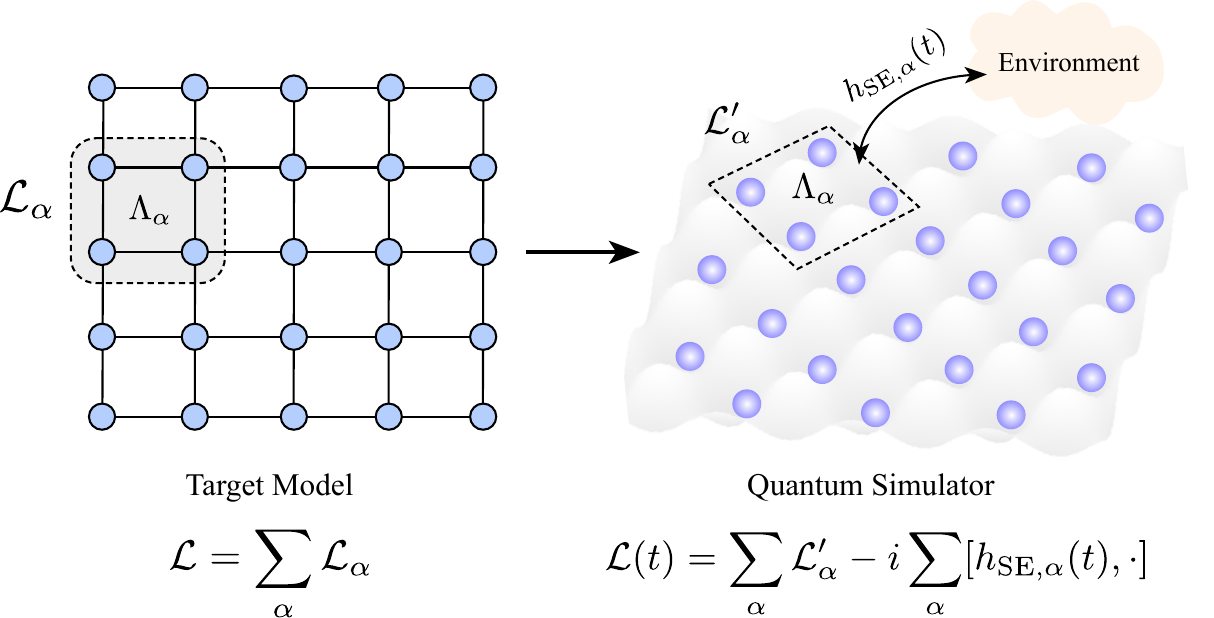}
    \caption{Schematic depiction of our error model for analogue quantum simulator. A target Lindbladian $\mathcal{L}$, expressed as sum of Lindbladian terms modelling interactions between groups of spins, when implemented on an analogue quantum simulator would have an hardware error per qubit --- this error can either be due to an incorrect configuration of the Lindbladian (i.e.~implementing $\mathcal{L}_\alpha'$ instead of $\mathcal{L}_\alpha$) or due to interaction with an external decohering environment.}
\end{figure*}

An alternative viewpoint would be to ask if there are certain interesting many-body problems for which a good estimate for the thermodynamic limit can be produced with a hardware with constant errors. This motivates us to look for `quantum simulation tasks' which are stable to these extensive errors, as made precise in the following definition.
\begin{definition}[Stable quantum simulation task]\label{def:stability} The quantum simulation task on $n$ spins of measuring an observable $O$ in a state $\rho_{\mathcal{L}}$ associated with a target Lindbladian $\mathcal{L}$ is said to be stable if the corresponding state $\rho'_\mathcal{L}$ prepared by the noisy simulator satisfies
\[
\abs{\textnormal{Tr}(O \rho) - \textnormal{Tr}(O \rho')} \leq f(\delta),
\]
for some $f$, independent of $n$, such that $f(\delta) \to 0$ as $\delta \to 0$.
\end{definition}
\noindent If a quantum simulation task is stable as per this definition, we can hope to be able to estimate the thermodynamic limit of the observable on a quantum simulator to a precision limited only by the hardware error rate $\delta$, and independent of the size of the problem. In particular, these problems would not require the hardware error to be scaled down with system size even in the absence of error correction, and it is reasonable to consider them to be problems that analogue quantum simulators can conceivably solve in the near term.

In the remainder of this section, we systematically study several important problems arising in many-body physics, and show that commonly considered intensive observables are expected to be stable to errors. We first study geometrically local Gaussian fermion models with Gaussian errors, and show that intensive observables (either local observables, or translationally invariant sums of local observables) are stable both for the problem of constant-time dynamics and equilibrium without any restrictive assumptions on the model --- our results hold not only for gapped models, but also for gapless models. Then, we study the same question for (non-Gaussian) many-body spin systems --- here, we use well known locality properties of these models to show that local observables are stable for constant-time dynamics and equilibrium properties, but with more restrictive assumptions on the system (e.g.~the Hamiltonian being stably gapped for ground states, or exhibiting exponential clustering of correlations in the Gibbs state). Our results are summarized in Table \ref{tab:stability}, and lend strong evidence for several many-body problems being amenable to noisy quantum simulation.
%\begin{table*}
%\centering
% \begin{tabular}{@{} l*{5}{>c<} @{}}
% \toprule
% Subject & \multicolumn{5}{c@{}}{$N_{\textnormal{trials}}$}\\
% \cmidrule(l){2-6}
% & 20 & 60 & 100 & 140 & 180 \\
% \midrule
% \multirow{2}{*}{AD} & 50.5\pm 3.2 & 50.3\pm4.5 & 52.6\pm4.7 & 51.7\pm5.3 & 51.5\pm 8.4\\
% & 74.8\pm4.9 & \\[1ex]
% \multirow{2}{*}{AS} & 50.4\pm3.6 & & & \\ 
% & 74.9\pm6.8 & & & & \\
% \multirow{2}{*}{NR} & 50.3\pm2.7 & & & \\ 
% & 75.2\pm6.5 & & & & \\
% \multirow{2}{*}{RA} & 49.9\pm2.7& & &\\ 
% & 74.8\pm6.6& & & &\\
% \midrule
% \multirow{2}{*}{Average} & 50.3\pm3.0& & &\\
%  & 74.9\pm6.2& & & &\\
% \bottomrule
% \end{tabular}
%\end{table*}
% \begin{table*}
% \begin{tabularx}{\textwidth}{ | X | c | }
%   \hline
%   \lipsum[1] & top\\
%   \hline
%   \noindent\parbox[c]{\hsize}{\lipsum[1]} & center\\
%   \hline
%   \noindent\parbox[b]{\hsize}{\lipsum[1]} & bottom\\
%   \hline
% \end{tabularx}
% \end{table*}
\begin{table*}
    \begin{tabular}{p{2.5cm} L{2.5cm} L{4cm} L{4cm}  L{4cm}}
    \hline 
       \textbf{Problem} & 
       \textbf{Error} &
       \textbf{Assumption} &
       \textbf{Observable} &
       \textbf{Stability, $f(\delta)$}\\ \hline \\
       
        Dynamics &
        General errors & 
        GF: None. \newline \ \newline \ \newline
        SS: None. &
        GF: $k$-local observables. \newline \ \newline \ 
        SS: Local observables. &
        GF:  $O(\delta t)$. \newline \ \newline \ \newline
        SS: $O(\delta t^{d + 1})$. \\ \ \\ \hline \\
            
        Ground state  &
        Coherent Hamiltonian errors &
        GF: Assumption 1. \newline \ \newline \ \newline
        SS: Stable gap. &
        GF: Translationally invariant $k-$local observables \newline \ \newline SS: Local observables. &
        GF: $O(\delta^\beta)$, where $\beta$ is a model dependent constant. \newline \ \newline
        SS: $O(\delta)$. \\ \ \\ \hline  \\

        Gibbs state & 
        Coherent Hamiltonian errors &
        GF: No assumption. \newline \ \newline \ \newline \ \newline
        SS: Stably exponentially clustered correlations. &
        GF: Translationally invariant $k-$local observables. \newline \ \newline
        SS: Local observables. &
        GF: $O(\sqrt{\delta})$. \newline  \ \newline \ \newline \ \newline
        SS: $\tilde{O}(\exp(-\Omega(\log^{1/d}{\delta^{-1}})))$.  \\ \ \\ \hline \\

        Fixed points &
        GF: Coherent and Incoherent Markovian errors. \newline \ \newline
        SS: General errors&
        GF: Assumption 2. \newline \ \newline \ \newline \ \newline
        SS: Rapid Mixing. &
        GF: Translationally invariant $k-$local observables \newline \ \newline \ \newline SS: Local observables \ \newline. &
        GF: $O(\delta^\beta)$, where $\beta$ is a model dependent constant. \newline \ \newline \ \newline \ \newline 
        SS: $O(\delta)$. \\ \\
    \hline
    \end{tabular}
    \caption{Summary of the stability results for dynamical and equilibrium many-body problems, together with the required assumptions on the many-body model and errors. Both results for Gaussian fermions and spin systems are summarized --- `GF' indicates Gaussian fermions and `SS' indicates spin systems. Note that observables for the Gaussian fermionic problems are all quadratic. Also, $\tilde{O}$ in the first column suppresses $\log(\delta^{-1})$ factors.}
    \label{tab:stability}
\end{table*}

% The stability of quantum simulation tasks for many-body physics problems, while hard to prove rigorously, can be intuitively understood. For simplicity, consider the specific setting in which both the target Hamiltonian $H$ and $V = H' - H$ are translationally invariant --- in this case, if the thermodynamic limit of an observable (e.g.~a local observable) $O$ in a non-equilibrium or equilibrium state associated with the translationally invariant Hamiltonian $H(s) = H + sV$ exists and is a continuous function of $s$, then it can be seen that it is indeed stable in the sense of definition \ref{def:stability}. However, this argument falls short of a full proof since in an actual experiment, even for target Hamiltonians that are translationally invariant, the presence of errors can make it translationally varying thus making it hard to define its thermodynamic limit. In particular, disorder can induce Anderson or many-body localization which can even result in local observables being unstable. In the following sections, however, we show that for several many-body problems, physically interesting observables and order parameters are stable even to translationally varying errors in the Hamiltonian.

\subsection{Gaussian fermion models}\label{sec:stability_gaussian_fermion}
We will consider fermions arranged on a $d-$dimensional lattice with $L$ sites in each direction $\mathbb{Z}_L^d$, and at each site we have $D$ fermionic modes --- we denote by $c_x^\alpha$ for $x \in \mathbb{Z}_L^d, \alpha \in \{1, 2 \dots 2D\}$ the Majorana operators associated with each site $x$. We consider a general open quantum simulation problem with geometrically local interactions with interaction range $R$. This is specified by a quadratic Hamiltonian $H$, and $n_L$ linear jump operators $L_{j, x}$ for every site $x \in \mathbb{Z}_L^d$,
\begin{subequations}\label{eq:quadH_jump}
\begin{align}
    &H = \sum_{\substack{x, y \in \mathbb{Z}_L^d \\ d(x, y) \leq R}} 
 \sum_{\alpha, \beta = 1}^{2D}h^{\alpha, \beta}_{x, y} c_x^\alpha c_y^\beta, \\
    &L_{j, x} = \sum_{\substack{y \in \mathbb{Z}_L^d \\ d(x, y) \leq R}} \sum_{\alpha = 1}^{2D} l_{j; x, y}^{\alpha} c_y^\alpha, \forall \ j\in \{1, 2 \dots n_L\}.
\end{align}
\end{subequations}
Without loss of generality, we can assume that $|h^{\alpha, \beta}_{x, y} | \leq 1$, $\abs{l_{j; x, y}^{\alpha}} \leq 1$ and $n_L \leq 2D(2R + 1)^d$.

For the results in this subsection, we restrict ourselves to Gaussian errors (coherent or incoherent) when this model is implemented on a quantum simulator. Due to coherent hardware errors, the quantum simulator instead implements a perturbed free-fermion Hamiltonian $H'$,
\begin{subequations}\label{eq:quadH_jump_pert}
\begin{equation}
        H' = \sum_{\substack{x, y \in \mathbb{Z}_L^d \\ d(x, y) \leq R}} 
        \sum_{\alpha, \beta = 1}^{2D} h'^{\alpha, \beta}_{x, y} c_x^\alpha c_y^\beta,
    \label{eq:pertquadH}
\end{equation}
such that $\abs{h^{\alpha, \beta}_{x, y} - h'^{\alpha, \beta}_{x, y}}\leq \delta$. Furthermore, due to errors in the configuration of the jump operators, or due to Markovian incoherent errors, the quantum simulator implements perturbed jump operators $L_{j,x}'$,
\begin{equation}
        L_{j, x}' = \sum_{\substack{y \in \mathbb{Z}_L^d \\ d(x, y) \leq R}} 
        \sum_{\alpha = 1}^{2D} l'^{\alpha}_{j; x, y}  c_y^\alpha,\  \forall \ j\in \{1, 2 \dots n_L\}.
    \label{eq:pertquadL}
\end{equation}
\end{subequations}
where again $\abs{l_{j; x, y}^\alpha - l'^{\alpha}_{j; x, y}} \leq \delta$. Furthermore, we also consider Gaussian incoherent interactions with a decohering environment which, following the general setup described previously, is captured by a Gaussian system-environment $H_\text{SE}(t) = \sum_{x\in \mathbb{Z}_L^d}h_{x, \text{SE}}(t)$ with
\begin{subequations}
    \begin{align}\label{eq:system_environment_Hamiltonian}
        h_{x, \text{SE}}(t) =  \sum_{j = 1}^{n_L} A_{j, x}(t)Q_{j, x}^\dagger + \text{h.c.}.
    \end{align}
Here
\begin{align}
    Q_{j, x} = \sum_{\substack{y \in \mathbb{Z}_L^d \\ d(x, y) \leq R}} q_{j, x; y}^\alpha c_{y}^\alpha,
\end{align}
\end{subequations}
with $\abs{q_{j, x; y}^\alpha} \leq \sqrt{\delta}$ and $A_x(t)$ is an annihilation operator in the fermionic environment coupling to sites in the neighbourhood of $x$. These annihilation operators satisfy $\{A_x(t), A_{x'}^\dagger(s)\} = \delta_{x, x'}K_x(t - s)$, where $K_x(\tau)$ is the memory kernel describing the system-environment interaction and is assumed to satisfy $\int_{\mathbb{R}} \abs{K_x(\tau)}d\tau \leq 1$.

%%%%%%%%%%%%%%%%%%%%%%%%%%%%%%%%%%%%%%%%%%%%%%%%%%%%%%%%%%%%%%%%%%%%%%%
\emph{Finite-time dynamics}. We first consider the problem of evolving the quantum simulator for time $t$ and measure the expectation value of Gaussian observables $O_0$ which are either $k-$local, i.e.~they act on a set $\mathcal{S} \subseteq \mathbb{Z}_L^d$ of $k$ sites
\begin{align}\label{eq:gaussian_few_sites_obs}
O_0 = \sum_{x, y \in \mathcal{S}} \sum_{\alpha, \beta = 1}^{2D} o_{x, y}^{\alpha, \beta} c_x^\alpha c_y^\beta,
\end{align}
or weighted averages of $k-$local Gaussian observables i.e.~are of the form $\sum_{i = 1}^M w_i O_i$, where $O_i$ is of the form Eq.~\ref{eq:gaussian_few_sites_obs}, $\sum_{i = 1}^M \abs{w_i} = 1$ and $M$ can possibly grow with $n = DL^d$. We consider an arbitrary Gaussian initial state, and let the target state $\rho$ be the state obtained by evolving it with the target Lindbladian specified by Eq.~\ref{eq:quadH_jump} for time $t$. We show in Appendix \ref{app:t_evol_ff} that
\begin{proposition}
The quantum simulation task of measuring $k-$local Gaussian observables, or their weighted sums, after constant-time dynamics under a spatially local Gaussian Hamiltonian is stable to coherent and incoherent errors with $f(\delta) = O(t \delta)$.
\label{prop:t_evol_ff}
\end{proposition}
\noindent We point out that the dependence of the error between the observable in perturbed and unperturbed models on $t$ is independent of the dimensionality of the lattice $d$ --- this result is thus stronger than what would be expected simply from locality, wherein the error would be expected to grow as $t\times (\text{Number of sites in the light cone at time }t) \propto t^{d + 1}$ --- we revisit this in section IV.

%%%%%%%%%%%%%%%%%%%%%%%%%%%%%%%%%%%%%%%%%%%%%%%%%%%%%%%%%%%
\emph{Equilibrium}. We next study the stability properties of intensive observables in equilibrium. For simplicity, we first consider the closed-system setting and study the stability of the ground state and Gibbs state. Suppose that the quantum simulator implements a target geometrically local Hamiltonian (Eq.~\ref{eq:quadH_jump}a), but due to the presence of coherent errors instead configures a perturbed Hamiltonian $H'$ (Eq.~\ref{eq:quadH_jump_pert}a). We emphasize that we only consider the simpler coherent Hamiltonian errors while studying the stability of ground state and Gibbs state, since these states can only be meaningfully defined for a closed system. Later on in this section, we will study the more natural equilibrium problem of the `Lindbladian fixed point' where incoherent errors can also be accounted for. Furthermore, we make the following additional physically reasonable assumption on the density of modes of the target Hamiltonian $H$.
\begin{assumption}\label{assumption:cont_eigenenergies}
The number of eigenfrequencies $n_\eta$ of $H$, which are eigenvalues of the matrix $h_{x, y}^{\alpha, \beta}$ defining the target Hamiltonian, lying in the interval $[-\eta, \eta]$ for sufficiently small $\eta$ satisfy the upper bound
\begin{align}\label{eq:density_of_modes_free_fermion}
n_\eta \leq n f_h(\eta) + \kappa(\eta, n),
\end{align}
where $n = DL^d$ is the number of fermionic modes, $f_h(\eta) \leq O(\eta^\alpha)$ for some $\alpha > 0$  and $\kappa(\eta, n)$ is $o(n)$ for any fixed $\eta$.
\label{assum:eigenfreqs}
\end{assumption}
% \begin{figure}
% \centering
% \includegraphics[scale=0.55]{figure_assump_1.pdf}
% \caption{Schematic depiction of a 1D translationally invariant Gaussian fermionic problem with nearest neighbour interactions and periodic boundary conditions, whose eigenmodes are described by a dispersion relation $\omega(k)$. The eigenmodes correspond to points on the $k$-axis spaced by $2\pi / n$, so the number of modes with energies in $[-\eta, \eta]$ can be approximated by the length of the corresponding interval (i.e.~inverse image with respect to the function $\omega(k)$) divided by $2\pi / n$.}
% \end{figure}
\noindent Alternatively stated, this assumption is a continuity\footnote{More precisely, it demands that $\lim_{n \to \infty} n_\eta / n$ is a Hölder continuous function of $\eta$.} condition on the thermodynamic limit of the fraction of eigenmodes with energies in the interval $[-\eta, \eta]$ and it ensures that eigenvalues do not accumulate too fast near zero. It is expected to be true for most physically relevant models --- in particular, it is weaker than the existence of a gap and thus contains gapped models, which are well known to exist for many experimentally relevant many-body problems. For a gapped Gaussian fermionic model, $f_h(\eta) = 0$ since, if there are fermionic eigenmodes near $0$, then adding a fermion into these modes would provide an excited state with only $O(n^{-1})$ energy higher than the ground state energy. Furthermore, we also expect this assumption to be generically true for translationally invariant local Hamiltonians, where the eigenfrequencies can be described by a smooth dispersion relation $\omega(k)$ as a function of the momentum $k$ associated with that mode. In this case, $f_h(\eta) \leq O(\eta^\alpha)$, with $\alpha$ being determined by the derivatives of $\omega(k)$ in the vicinity of $\omega = 0$.
%For instance, consider the 1D translationally invariant case with periodic boundary conditions --- if $\omega(k_0) = 0$, then for small $\delta k$, $\omega(k_0 + \delta k ) \approx \omega^{(m)}(k_0) \delta k^m / m!$, where $m > 0$ is the order of the first non-zero derivative of $\omega(k)$ at $k = k_0$. We can then see that the length of the interval along the $k$ axis corresponding to $\omega(k) \in [-\eta, \eta]$ is $O(\eta^{1/m})$ --- since the eigen-modes for this system are spaced by $2\pi / n$ along the $k$ axis, this yields $n_\eta \leq n\times O(\eta^{1/m})$ and thus satisfies assumption \ref{assum:eigenfreqs}.

 The observables we consider while analyzing ground states are translationally invariant Gaussian observables generated $k-$locally, i.e., if $O_0$ is a $k-$local observable of the form of Eq.~\ref{eq:gaussian_few_sites_obs}, then we consider observables of the form
\[
O = \frac{1}{n}\sum_{x \in \mathbb{Z}_L^d} \tau_x(O_0),
\]
where $\tau_x(O_0)$ is the observable $O_0$ translated by $x$. Considering the target state to be the ground state of the Gaussian fermion Hamiltonian, we then obtain the following proposition (proved in appendix \ref{app:gs_ff})
\begin{proposition}
The quantum simulation task of measuring translationally invariant Gaussian observables generated $k-$locally, in the ground state of a spatially local Gaussian fermion model whose density of modes satisfies Eq.~\ref{eq:density_of_modes_free_fermion} is stable to coherent Hamiltonian errors with $f(\delta) = O(\sqrt{\delta}) + f_h(O(\delta^{1/4})) \leq O(\delta^\beta)$ for some model-dependent constant $\beta$.
\label{prop:gs_ff}
\end{proposition}
\noindent We note that the stability result above holds with only a mild continuity assumption on the density of modes of the model. Importantly, it holds for models which are not gapped, i.e.~the energy separation between the ground state and the first excited state vanishes as $n\to\infty$. As an example, consider the ground state of the 1D Su--Schrieffer--Heeger (SSH) model on $n$ fermions with periodic boundary condition:
\begin{equation}
    H_{\text{SSH}}[J]\equiv\sum_{i=1}^{n}{t_ia_i^\dagger a^{\phantom\dagger}_{i+1}}+\text{H.c},~~~~~t_i=\begin{cases}
    1 & i \text{ odd,}\\
    J & i \text{ even.}
    \end{cases}
    \label{eq:SSH}
\end{equation}
where $a_{n + 1} \equiv a_n$. This model displays a (topological) phase transition at $J=1$, where the gap closes as $1/n$, and is gapped otherwise. The observable we consider is the energy density $H_\text{SSH}[J]/n$ of the unperturbed Hamiltonian. Figure \ref{fig:ssh_gs_study}(a) shows impact of changing system size on this energy density --- we see that for both gapped ($J = 0.5, 1.5$) and gapless ($J = 1.0$) cases, the errors in the energy density becomes independent of $n$ as $n\to \infty$, verifying the expectation in proposition 2. Furthermore, we show the error in the energy density for large $n$ as a function of $\delta$ in Fig.~\ref{fig:ssh_gs_study}(b) and see that, consistent with proposition 2, this error $\to 0$ as $\delta \to 0$. Finally, Fig.~\ref{fig:ssh_gs_study}(c) shows this error as a function of $J$ --- we see the error peak near $J = 1$ (i.e.~the point where the gap in the Hamiltonian closes), and that it is smaller for values of $J$ where the model is gapped.

The translational invariance of the observables considered here is key to the stability result ---  translationally varying observables need not be stable, even if they are intensive and local. A simple example here is Anderson localization --- consider $H$ to be a 1D translationally invariant tight-binding model i.e.~$H = \sum_{i = 1}^{n - 1} (a_{i + 1}^\dagger a_i + a_i^\dagger a_{i + 1})$, with errors $ \sum_{i = 1}^n \delta_i a_i^\dagger a_i$ where $v_i$ is chosen uniformly at random between $[-\delta, \delta]$. The ground state of $H$ is completely delocalized across the spin-chain. In the presence of errors, no matter how small, this model is known to be localized. Now, for every $\delta_1, \delta_2 \dots \delta_n$, consider the intensive observable $O_{\delta_1, \delta_2 \dots \delta_n}$ given by the average particle numbers on $\Theta(1/\delta) $ sites around the site where the ground state is localized. This observable, when measured in the delocalized ground state of the unperturbed Hamiltonian $H$, yields an expected value of $0$ as $n \to \infty$. On the other hand, in the ground state of the perturbed localized model it will yield an expected value of $\Theta(1)$. Thus, not all translationally varying observables can be stable, even if we restrict ourselves to free-fermion models, with the observables being intensive and spatially local.

Next, we consider the Gibbs state, $ e^{-\beta H}/\text{Tr}(e^{-\beta H})$ where the inverse-temperature $\beta$ is a constant independent of $n$ and again study the stability of translationally invariant Gaussian observables that are generated $k-$locally. We show in appendix \ref{app:Gibbs_ff} that
\begin{proposition}
The quantum simulation task of measuring translationally invariant Gaussian observables generated $k-$locally, in the Gibbs state at inverse-temperature $\beta$ of a spatially local free-fermion is stable to coherent Hamiltonian errors with $f(\delta) = O(\beta\sqrt{\delta})$.
\label{prop:Gibbs_ff}
\end{proposition}
\noindent We point out that, in contrast to the corresponding result for ground states, this stability result corresponding to the Gibbs state does not rely on an assumption on the density of modes of the target Hamiltonian. However, $f(\delta)$ grows with $\beta$, so this result does not directly imply the stability of the ground state since $\beta$ would in general have to be increased with $n$ for the Gibbs state to approximate the ground state.

In the more general setting of a Markovian open quantum system, the fixed point of the master equation would capture its equilibrium properties. Here, the quantum simulator is configured to implement the Hamiltonian and jump operators in Eq.~\ref{eq:quadH_jump} but instead, due to coherent and Markovian incoherent errors, implements the perturbed Hamiltonian and jump operators in Eq.~\ref{eq:quadH_jump_pert}. Similar to the case of the Hamiltonian ground state problem, we make an assumption on the spectral properties of the target Lindbladian --- in particular, similar to assumption \ref{assumption:cont_eigenenergies}, we assume (Hölder) continuity of the fraction of modes with decay rates in the interval $[0, \eta]$,
\begin{assumption}[Informal]\label{assump:lindbladian_holder_cont}
The Gaussian Lindbladian has a unique fixed point and the number of eigenmodes $n_\eta$ with decay rates lying in the interval $(0, \eta]$, for sufficiently small $\eta$, satisfy the upper bound 
\[
n_\eta \leq n f_\ell(\eta) + \kappa(\eta, n),
\]
where $n = DL^d$ is the number of fermionic modes, $f_\ell(\eta) \leq O(\eta^\alpha)$ for some $\alpha > 0$ and $\kappa(\eta, n) \leq o(n)$ for any fixed $\eta$.
\end{assumption}
We provide a precise definition of eigenmode decay rate of a Gaussian Lindbladian in appendix \ref{app:fp_ff}. We remark that this assumption is very mild and is expected to be satisfied for physically relevant models --- in particular, similar to the case of Gaussian Hamiltonians, this assumption is satisfied for translationally invariant models. Furthermore, this assumption is expected to be satisfied for gaussian fermion models that are rapidly mixing in which case all the modes other than the fixed point mode typically have a system size independent decay rate (i.e.~there is a gap in the decay rate spectrum of the Lindbladian) \cite{cubitt2015stability}. Beyond rapidly mixing Lindbladians, assumption \ref{assump:lindbladian_holder_cont} includes systems which have eigenmodes with decay rates scaling as $O(1/n)$ (i.e.~the Lindbladian decay rate spectrum is gapless), and take a much longer time ($\sim \Theta(n)$) to reach its fixed point. For models satisfying assumption \ref{assump:lindbladian_holder_cont}, we show in appendix \ref{app:fp_ff} that 
\begin{proposition}
The quantum simulation task of measuring translationally invariant Gaussian observables generated $k-$locally, in the fixed point of a spatially local free-fermion is stable to coherent and Markovian incoherent errors with $f(\delta) = O(\delta^{1/2}) + O(f_\ell(\delta^{1/4})) \leq O(\delta^\beta)$ for some model-dependent constant $\beta$.
\label{prop:fixed_points_gaussian}
\end{proposition}

\begin{figure}
    \centering
\includegraphics[scale=0.325]{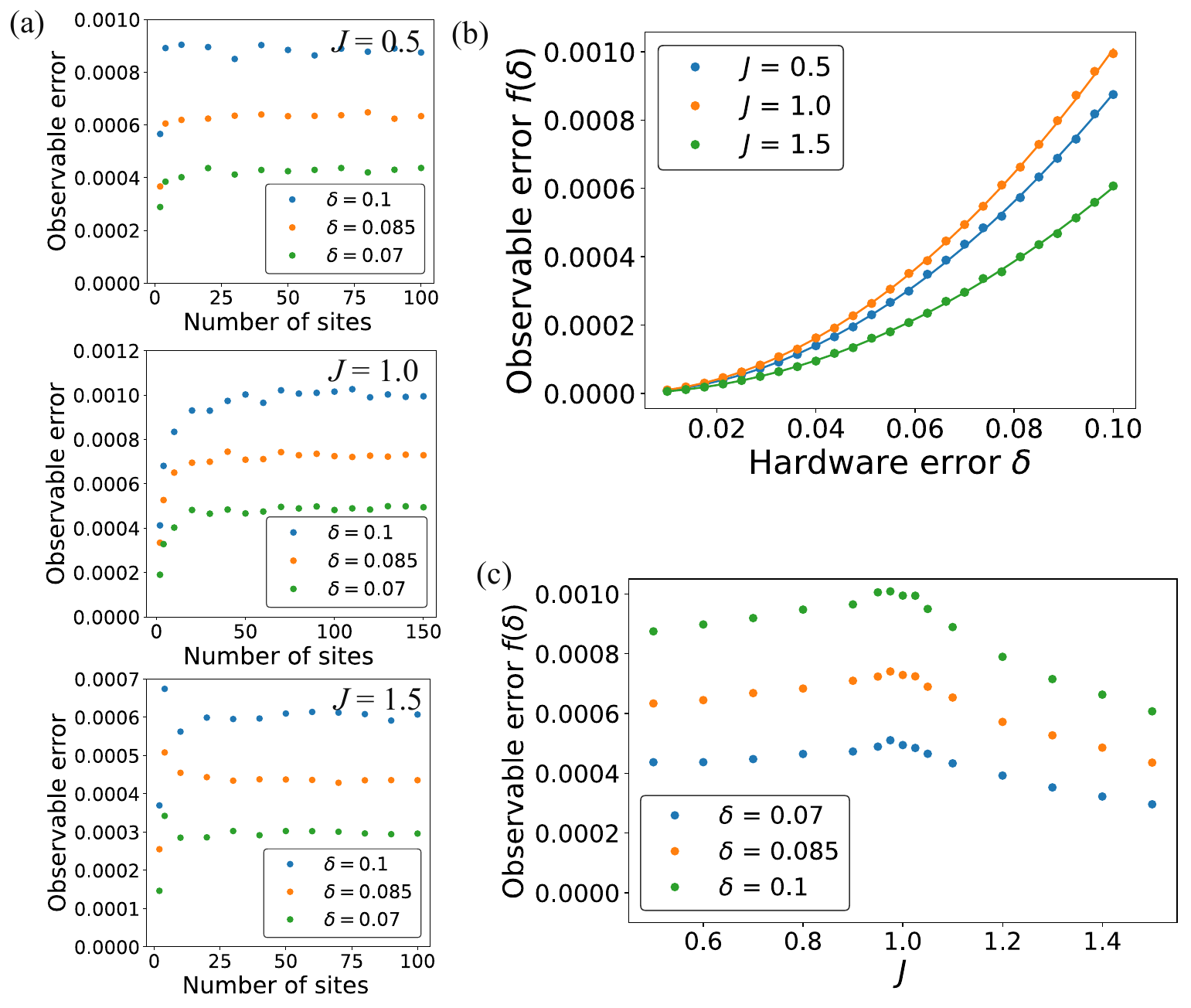}
    \caption{Numerical study of the error for an intensive translationally varying observable in the SSH model. The observable that we study here is $O = H_\text{SSH}[J]/n$, where $H_\text{SSH}[J]$ is the Hamiltonian of the ideal SSH model (Eq.~\ref{eq:SSH}) (a) The error in the expected value of the observable 
 $O$ in the ground state between the perturbed and unperturbed Hamiltonians, as a function of $\delta$, the hardware error, and the number of sites $n$. For both gapped ($J = 0.5, 1.5$) and gapless ($J = 1.0$) cases, we see that the error in $O$ becomes independent of $n$ as $N\to \infty$. (b) Numerically extracted error between the perturbed and unperturbed models for $n\to \infty$ as a function of $\delta$, and its fit with $\delta^2$. (c) The error between the perturbed and unperturbed model as a function of $J$ --- for the same hardware error $\delta$, this error peaks at $J = 1$ which is also the point at which the gap in the unperturbed model closes. All the errors are computed by averaging over 500 random instances of perturbed models.}
    \label{fig:ssh_gs_study}
\end{figure}

%%%%%%%%%%%%%%%%%%%%%%%%%%%%%%%%%%%%%%%%%%%%%%%%%%%%%%%%%%%%%%%%%%%%%%%%%%%%%%

\subsection{Quantum spin systems}

{While for free fermion models, we could prove tight stability results with minimal assumptions on the model, looser stability results hold for quantum spin systems under more restrictive assumptions on their many-body spectrum. In this section, we consider the more general setup described in section \ref{sec:stability_setup} and show the stability of several quantum simulation tasks, in both dynamics and equilibrium, using locality results that have already been established in the many-body literature \cite{lieb1972finite, hastings2004lieb, hastings2005quasiadiabatic, hastings2004locality, poulin2010lieb, nachtergaele2006lieb, bachmann2012automorphic, brandao2019finite}.

%%%%%%%%%%%%%%%%%%%%%%%%%%%%%%%%%%%%%%%%%%%%%%%%%%%%%%%%%
\emph{Finite time dynamics}. Consider first the setting where an initial state $\rho(0) =  (\ket{0}\bra{0})^{\otimes n}$ is evolved under a Lindbladian $\mathcal{L}$ (Eq.~\ref{eq:target_lind_main_text}) for a time $t$ that is independent of $n$. We consider observables $O$ that are either local (i.e.~they only act non-trivially on an $n-$independent subset of spins), or of the form
\begin{align}\label{eq:k_local}
O = O_1 O_2 \dots O_k,
\end{align}
where $O_1, O_2 \dots O_k$ are local and $k$ is independent of $n$, or
\[
O = \sum_{i = 1}^M w_i O_i,
\]
where $\sum_{i = 1}^M \abs{w_i} = 1$, $O_i$ are of the form of Eq.~\ref{eq:k_local} and $M$ can possibly grow with $n$.
For these observables, the stability of this quantum simulation task can be stated:
\begin{proposition}
The quantum simulation task of measuring $k-$local observables, or their weighted averages, for constant-time dynamics under a spatially local Lindbladian is stable under coherent and incoherent errors with $f(\delta) =  O(t^{d + 1}\delta)$.
\label{prop:t_evol}
\end{proposition}
\noindent The proof of this result, provided in appendix \ref{app:t_evol} uses the Lieb-Robinson bounds \cite{lieb1972finite, hastings2004lieb, poulin2010lieb}.  We note that a similar result has been proven for coherent errors and Markovian incoherent errors in Ref.~\cite{cubitt2015stability}. Our contribution is to show that this bound holds in the more general setting of coherent and non-Markovian incoherent errors, and thus is more directly applicable to experimentally realistic quantum simulators.

Note also that for large $t$, the error between the target observable and the observable measured on the quantum simulator grows as $t^{d + 1}$ --- this is looser than the corresponding result in Gaussian fermion models (proposition 1), where the error grows only as $t$. Furthermore, since this error bound becomes loose with $t$, it prevents us from using this result to understand the stability of quantum simulation tasks which are aimed at studying the ground state properties of many-body Hamiltonians, and use the adiabatic algorithms \cite{daley2022practical} that evolve a Hamiltonian for $t \sim \text{poly}(n)$. To address these problems, we separately consider the stability of equilibrium problems (ground states, Gibbs states and fixed points).

%%%%%%%%%%%%%%%%%%%%%%%%%%%%%%%%%%%%%%%%%%%%%%%%%%%%%%%%%
\emph{Equilibrium}. We next study the stability of the task of simulating the ground state and Gibbs states of $H$, and focus only on understanding the impact of coherent Hamiltonian errors. We first consider the problem of measuring $k-$local observables, and their weighted averages, in the ground state. We assume that $H$ is gapped i.e.~the energy difference between the ground state and the first excited state is larger than an constant $\Delta$ independent of $n$. Furthermore, we also assume that the Hamiltonian remains gapped in the presence of errors --- we refer to such a target Hamiltonian $H$ to be \emph{stably gapped}. We point out that the stability of the gap in the presence of errors or perturbations has only been shown for certain frustration free models with local topological order \cite{bravyi2010topological, bravyi2011short, michalakis2013stability, cirac2013robustness}, although we posit it as a reasonable physical assumption. The stability of this quantum simulation task is a direct consequence of the spectral flow method developed by Hastings and co-workers \cite{hastings2005quasiadiabatic, bachmann2012automorphic} which shows that there exists a unitary taking the ground state of $H$ to the ground state of $H'$ that is quasi-local. We thus obtain the following proposition and we include a proof of this in appendix \ref{app:gs}.
\begin{proposition}
The quantum simulation task of measuring $k$-local observables, or their weighted averages, in the ground state of stably gapped spatially local Hamiltonians is stable to coherent Hamiltonian errors with $f(\delta) = O(\delta)$.
\label{prop:gs}
\end{proposition}
\noindent The choice of observables here is crucial to having a stable quantum simulation task --- it is well understood that even for stably gapped Hamiltonians, non-local observables would not be stable. Furthermore, we point out that for the case of Gaussian Fermions and for translationally invariant local observables, the corresponding stability result is less restrictive --- in particular, it does not require even the existence of a gap in the target Hamiltonian. 
%A counter-example here would be the unperturbed Hamiltonian $H = \sum_{i=1}^n Z_i$, and a perturbed Hamiltonian $H' = \sum_{i=1}^n Z_i + \delta X_i$. It can be seen that both these models are gapped for small $\delta$ --- the gap of $H$ is 2, and gap of $H'$ is $2\sqrt{1 + \delta^2}$. However, the non-local observable $O = (\ket{0}\bra{0})^{\otimes n}$, when measured in the ground state of $H$ is 1 and when measured in the ground state of $H'$ evaluates to $(1 + \delta^2)^{-n/2}$ --- thus, at any $\delta > 0$, the error in this observable $\to 1$ as $n \to \infty$, making it unstable in the sense of definition 1.

We next consider the Gibbs state of $H$ at some temperature $\beta$ independent of $n$, and assume that the Gibbs state has an exponential clustering of correlation \cite{brandao2019finite} i.e.~for any two observables $A$, $B$ separated by distance $l$,
\[
\bigabs{\langle A\otimes B \rangle - \langle A \rangle \langle B \rangle }\leq  \norm{A}\norm{B} O(e^{-c_2 l}),
\]
for some model-dependent constant $c_2$. Furthermore, as in the case of ground states, we assume that this exponential clustering of correlations is stable under errors. In this case, we obtain that the problem of measuring $1-$local observables and their weighted averages is stable --- a proof of this is included in appendix \ref{app:Gibbs}.
\begin{proposition}
The quantum simulation task of measuring 1-local observables, or their weighted averages, in the Gibbs state of spatially local Hamiltonians with stable exponential clustering of correlations is stable to coherent Hamiltonian errors with $f(\delta) =  O(\log^{1 - 1/d}(1/\delta)e^{-\Omega(\log^{1/d}(1/\delta))}\big)$.
\label{prop:Gibbs}
\end{proposition}

More generally, we can consider the problem of finding local observables in the fixed points of spatially local Lindbladians. For this, we restrict ourselves to rapidly mixing Lindbladians which were identified in Ref.~\cite{cubitt2015stability} --- a Lindbladian $\mathcal{L}$ on $n$ spins with fixed point $\sigma$ is rapidly mixing if $\norm{e^{\mathcal{L}t} - \text{Tr}(\cdot) \sigma}_\diamond \leq \text{poly}(n)e^{-\Theta(t)}$ i.e.~irrespective of the initial state, the state of the system converges exponentially fast to the fixed point $\sigma$. Under this assumption, we can show the following stability result.
\begin{proposition}
The quantum simulation task of measuring $k$-local observables, or their weighted averages, in the fixed point of spatially local Lindbladians which satisfy rapid mixing is stable to coherent and incoherent errors with $f(\delta) =  O(\delta\big)$.
\label{prop:fixed_points}
\end{proposition}
\noindent The proof of this proposition, presented in appendix \ref{app:fp}, builds on the analysis in Ref.~\cite{cubitt2015stability} --- in particular, we point out that Ref.~\cite{cubitt2015stability} already establishes that local observables in fixed points of spatially local rapidly mixing Lindbladians is stable to coherent errors and incoherent Markovian noise. Our key contribution is to extend this to the more general and experimentally realistic setting of non-Markovian noise.

%%%%%%%%%%%%%%%%%%%%%%%%%%%%%%%%%%%%%%%%%%%%%%%%%%%%%%%%%%%%%%%%%%%%

%%%%%%%%%%%%%%%%%%%%Section: Quantum advantage%%%%%%%%%%%%%%%%%%%%%%%
\section{Quantum advantage with noisy quantum simulators}

\subsection{Ideal quantum simulators}\label{sec:ideal_qsim}
In the absence of any noise, the advantage of a quantum algorithm over a classical algorithm is often formulated in terms of their run-time scaling with respect to the system size. However, in many-body physics problems, the quantities of interest are the value of certain intensive observables in the thermodynamic limit i.e.~when the system size $n\to \infty$. Consequently, it is less meaningful to consider the algorithm's complexity as a function of system size and instead consider it as a function of the target precision 
$\varepsilon$ demanded in the computed thermodynamic limit \cite{aharonov2022hamiltonian, watson2022computational}. More precisely, let us consider a many-body model defined as a family of Lindbladians $\{\mathcal{L}_n\}_{n \in \mathbb{N}}$ and observables $\{O_n\}_{n\in \mathbb{N}}$, where $\mathcal{L}_n, O_n$ act on $n-$spins. We are interested in the expected value of $O_n$ in a many body quantum state, e.g.~in equilibrium or dynamics, associated with the Lindbladian $\mathcal{L}_n$, $\rho_{\mathcal{L}_n}$. We furthermore assume that the models and observables under consideration have a well-defined thermodynamic limit i.e.
\begin{equation}
O^* := \lim_{n\to \infty}\text{Tr}(\rho_{\mathcal{L}_n} O_n)
\end{equation}
exists.
% Examples of physically relevant $\rho_{H_n}$ would include the
% states obtained on evolving an initial state of $n$ spins after some
% time $t$, the Gibbs state corresponding to $H_n$ or the ground
% state of $H_n$. We will mostly focus on  families of local 
% Hamiltonians
% \begin{align}\label{eq:hamiltonian_n_spins}
% H_n = \sum_{x \in \mathbb{Z}_{L_n }^d} h_x,
% \end{align}
% where $h_x$ acts on spins in a cube of unit length with $x$ being
% its lower left corner and satisfies $\norm{h_x} \leq J$ for some $J
% > 0$. Typically, one considers translationally invariant models,
% where $h_x$ is just some $h_0$ translated to the position $x$. We will also
% only consider here observables $O_n$ are then acting on a single or
% few lattice sites.
Now, given a precision $\varepsilon$, we can then choose
$n$ as a function of $\varepsilon$ such that
\begin{equation}
\label{Othermlim}
\abs{O^* - \textnormal{Tr} (O_n\rho_{H_n})} \leq \varepsilon,
\end{equation}
i.e.~approximate the thermodynamic limit by a
finite-size problem. The run-time of a quantum simulation or a
classical simulation for the finite-size problem can thus be
expressed in terms of the precision $\varepsilon$ demanded
in the thermodynamic limit. This allows us to then compare the
scalings of the run-time of these algorithms with the precision
$\varepsilon$, and declare an algorithm \emph{to have an advantage
in precision} compared to others depending on their respective
scaling.  For instance, if a quantum algorithm has a complexity $T_{\rm Q}={\rm poly}(1/\varepsilon)$, then we will have a superpolynomial advantage over a classical algorithm with complexity $T_{\rm cl}={\rm exp}({\log^2(1/\varepsilon})) $ and exponential advantage over a classical algorithm with complexity $T_{\rm cl}={\rm
exp}({O(1/\varepsilon)})$. 

% We note that there are also classical algorithms which directly operate in the thermodynamic limit \cite{white1992density, jordan2008classical}. We will not consider them here, since for equilibrium problems in dimensions higher than two, it is not possible to give rigorous scalings, and for the problem of dynamics, they give the same scaling as the ones considered here.

\begin{figure}
    \centering
\includegraphics[scale=0.45]{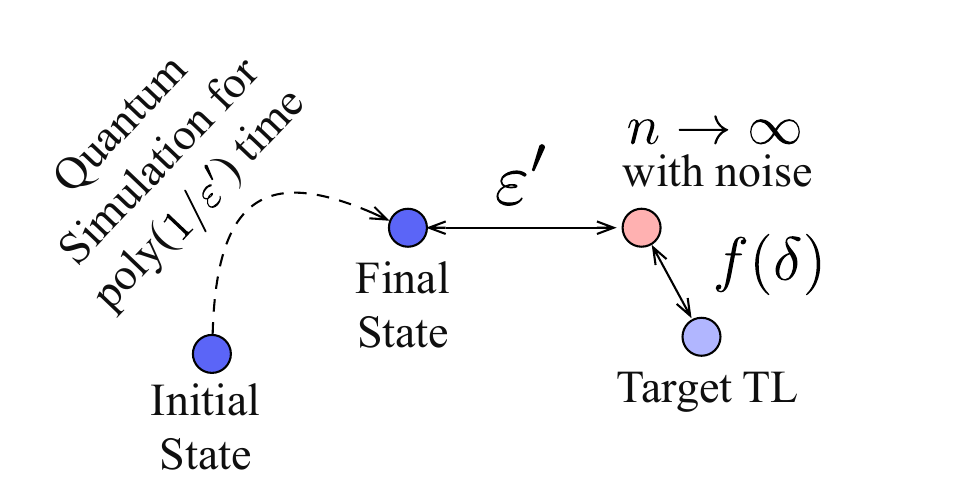}
    \caption{An erroneous quantum simulator can obtain the thermodynamic
    limit of the perturbed model to a precision $\varepsilon'$ in time $\text{poly}(1/\varepsilon')$ --- this thermodynamic limit, however, can have an error $f(\delta)$ from the target thermodynamic limit in the presence of hardware error $\delta$.}
\end{figure}

%%%%%%%%%%%%%%%%%%%%%%%%%%%%%%%%%%%%%%%%%%%%%%%%%%%%%%%%%
\emph{Finite time quantum dynamics}: We consider first an initial product state $\ket{0}^{\otimes n}$, and for $t > 0$, we take $\rho_{\mathcal{L}_n} = e^{\mathcal{L}_n t}\big((\ket{0}\bra{0})^{\otimes n}\big)$. The observable of interest is a fixed \emph{local observable} $O_n = O$. The existence of the thermodynamic limit is obtained directly using
the Lieb-Robinson bounds \cite{lieb1972finite, hastings2004lieb}, which also characterize the error
between the thermodynamic limit and its finite-size approximation -- for the problem of computing a local observable after evolving
$\ket{0}^{\otimes n}$ for finite-time $t$ with respect to a $d-$dimensional spatially-local Hamiltonian, $O^*$ exists and fulfills Eq.~\ref{Othermlim}
for $n = \Omega(\log^d(1 / \varepsilon)+t^d)$.
On an ideal quantum simulator, one would evolve $n = \Theta(\log^d(1 / \varepsilon)+t^d)$ qubits for a time  $t$ and measure the observable. The procedure would be repeated $\Theta(1/\varepsilon^2)$ to reduce the measurement error in the observable to $\varepsilon$, yielding a total run-time of $O(t/\varepsilon^2)$

On a classical computer, in general, the only algorithm with a rigorous guarantee known for the problem of computing local observables for a finite number of spins $n$ are either exact diagonalization or Krylov subspace methods \cite{brenes2019massively}. Using either of these methods has a worst-case run-time that scales at least exponentially with the number of spins $n$. To compute the thermodynamic limit to a precision $\varepsilon$, we would first approximate the thermodynamic limit by a finite-size problem and then use these classical algorithms on the resulting finite-size problem. In appendix \ref{app:lower_bounds_dyn} (proposition \ref{prop:lower_bound_dynamics}), we exhibit a local observable in a simple nearest-neighbour tight-binding model on $\mathbb{Z}^d$ such that the system size needed to approximate its thermodynamic limit to a precision $\varepsilon$, for a fixed $t$, is at least $\Omega(\log^d(\Theta(\varepsilon^{-1})) / \log^d \log(\Theta(\varepsilon^{-1})))$. Thus, exact diagonalization or Krylov subspace methods scale \emph{at least} super polynomially with $\varepsilon^{-1}$ on worst-case instances, yielding a super polynomial advantage of using quantum simulators.

We note that there are several classical heuristic algorithms, which use an efficiently contractable tensor network ansatz, that in many problems are much faster than the worst case  \cite{white1992density, jordan2008classical}. However, these methods do not have rigorous guarantees --- since, when noiseless, the quantum simulation of dynamics is not a heuristic, we only compare it to classical algorithms with rigorous guarantees. We also point out that one could also use other classical methods that work directly in the thermodynamic limit and obtain a better scaling than Krylov subspace methods with the required precision. For instance, in Ref.~\cite{alvaro} a method based on cluster expansion was analyzed for which the computational time is upper bounded by poly$(\varepsilon^{-1})$, although this upper bound scales super-exponentially with time. To the best of our knowledge, providing a lower bound on this classical algorithm is an open problem. However, assuming that it has the same scaling with $t$ as the upper bound provided in Ref.~\cite{alvaro}, a quantum simulator will have an exponential quantum advantage with respect to this method for evolution times $t \sim \textnormal{poly}(\varepsilon^{-1})$.

%%%%%%%%%%%%%%%%%%%%%%%%%%%%%%%%%%%%%%%%%%%%%%%%%%%%%%%%%
\emph{Ground state}: Consider next the problem of estimating local observables in the ground state of many-body Hamiltonians in the thermodynamic limit. The convergence rate of a finite-size approximation of a local observable to its thermodynamic limit for the ground state problem is expected to depend on whether the model is gapped (and hence the ground state has exponentially decaying correlations \cite{hastings2004locality, hastings2006spectral}) or gapless. While it is generally hard to rigorously characterize the rate of convergence of a finite-size approximation to the thermodynamic limit for ground states, it is physically reasonable to assume
\begin{itemize}
\item \textbf{Logarithmic Convergence}, Eq.~\ref{Othermlim} holds for $n = \Omega(
\log^d(1/\varepsilon))$, where $d$ is the lattice dimension, or

\item \textbf{Power-Law Convergence}, Eq.~\ref{Othermlim} holds for $n = \Omega({\rm
poly}(\varepsilon^{-1}))$.
\end{itemize}
\noindent The first case is expected to hold for gapped models, and can be rigorously established for models satisfying local topological quantum order condition \cite{bravyi2010topological, cirac2013robustness, michalakis2013stability}. Additionally, for gapped models, Logarithmic Convergence is expected to be tight since we can easily construct examples of spatially local gapped Hamiltonians (such as the AKLT model \cite{affleck2004rigorous}), for which a system-size of at least $\Omega(\log^d(\varepsilon^{-1}))$ is needed to approximate the thermodynamic limit of a local observable to a precision $\varepsilon$ (see appendix \ref{app:lower_bound_gs}, proposition \ref{prop:lower_bound_gapped}). Power-Law Convergence is expected to hold for critical (gapless) models --- for instance, this is the case for the Gaussian fermionic Hamiltonians analyzed in the previous section, under very general conditions for the Fermi surface. Similar to the situation with dynamics, currently available classical algorithms with rigorous guarantees to compute a ground state observable use either exact diagonalization or a Krylov subspace method on the finite-size Hamiltonian approximating the thermodynamic limit. Thus, for models satisfying the Logarithmic Convergence condition, a classical computer would require time $\exp(\Omega(\log^d(\varepsilon^{-1})))$ in the worst case. Instead, for models satisfying the Power-Law Convergence condition, a classical computer is expected to require time $\exp(\Omega(\text{poly}(\varepsilon^{-1})))$.

%\redit{Similar to the case for gapped models, in appendix \ref{app:lower_bound_gs}, proposition \ref{prop:lower_bound_gapless}, we show a simple example of a local free-fermion Hamiltonian for which a system size of at least $\Omega(\varepsilon^{-1})$ would be needed to approximate the thermodynamic limit of a local observable to a precision $\varepsilon$, thus forbidding a faster convergence to the thermodynamic limit for gapless models.}
Furthermore, to ensure that there is a quantum algorithm that reaches the ground state we will assume that $H_n$ is adiabatically connected to a family of Hamiltonians $H_n^{(0)}$ with efficiently preparable ground states such that the minimal gap, $\Delta_n$, along the adiabatic path fulfills $\Delta_n \ge \Omega(1/{\rm poly}(n))$. This assumption ensures that using the adiabatic algorithm one can reach the ground state within an error $\varepsilon$ in a time $T_{\rm Q}= {\rm
poly}(n,1/\varepsilon)$, or $T_{\rm Q} ={\rm poly}(1/\varepsilon)$ if framed entirely in terms of the precision of the thermodynamic limit\footnote{Note that we could have also
considered a constant gap, in which case, at least under certain
further assumptions on the Hamiltonian, it is provably possible to reach the ground state in $T_{\rm Q}={\rm polylog}(n,1/\varepsilon)$ \cite{ge2016rapid}.}, and is expected to hold for physically relevant gapped or gapless models. Comparing this run-time with those of the classical algorithms discussed above, we then expect a superpolynomial quantum advantage for (gapped) models satisfying Logarithmic Convergence to the thermodynamic limit, and (gapless) models satisfying Power-Law Convergence to the thermodynamic limit.

\emph{Fixed points}: Similar to ground states of many-body Hamiltonians, we can consider the problem of computing local observables in the fixed points of many-body Lindbladians. Depending on the spectral properties of the Lindbladian, a local observable in the fixed point may exhibit Logarithmic or Power-Law Convergence to the thermodynamic limit. In particular, for rapidly mixing Lindbladians, it has been shown in Ref.~\cite{cubitt2015stability} that local observables in the fixed point exhibit a logarithmic convergence to the thermodynamic limit. Additionally, this convergence is tight i.e.~we can exhibit a specific rapidly mixing Lindbladian for which a system-size of at least $\Omega(\log^d(\varepsilon^{-1}))$ is needed to approximate the thermodynamic limit of a local observable to a precision $\varepsilon$ (see appendix \ref{app:lower_bound_fp}, proposition \ref{prop:log_conv_fixed_point}). Furthermore, Lindbladians which are not rapidly mixing but take time polynomial in the system size to reach their fixed points would have local observables satisfying Power-Law Convergence. Examples of such Lindbladians would be those corresponding to Glauber dynamics corresponding to the 2D critical Ising model \cite{lubetzky2012critical}.

As with the ground state problem, Krylov subspace or exact methods, which have a rigorous guarantee, would require a worst-case time $\exp(\Omega(\log^d(\varepsilon^{-1})))$ for models with logarithmic convergence and $\exp(\Omega(\text{poly}(\varepsilon^{-1})))$ for models with power-law convergence. Furthermore, under the physically-motivated assumption that Lindbladian dynamics, for a finite system with system-size $n$, reaches its fixed point in at most $\textnormal{poly}(n)$ time, an ideal quantum simulator that implements this Lindbladian would require a time $O(\textnormal{poly}(1/\varepsilon))$ to approximate the thermodynamic limit of local observable to a precision $\varepsilon$. Thus, we expect to obtain a superpolynomial quantum advantage for problems with logarithmic convergence, and exponential quantum advantage for problems with power-law convergence.

%%%%%%%%%%%%%%%%%%%%%%%%%%%%%%%%%%%%%%%%%%%%%%%%%%%%%%%%%%%%%%%%
\subsection{Noisy quantum simulators}

In the presence of errors, the arguments formulated in the previous subsection no longer hold. For unstable quantum simulation tasks, we expect thermodynamic limits in the presence of errors to be a bad approximation to the target thermodynamic limit. However, as we previously discussed, for stable quantum simulation tasks in particular, the noisy quantum simulator can still produce a faithful approximation of the thermodynamic limit with the hardware error $\delta$ setting a limit on the obtained precision. More precisely, in time
$\textnormal{poly}(1/\varepsilon')$, the quantum simulator is
expected to compute the thermodynamic limit of the noisy model
to a precision $\varepsilon'$ --- the precision of the target
thermodynamic limit obtained is thus upper bounded by
\[
\varepsilon \leq O(\text{max}(\varepsilon', f(\delta))),
\]
where $f(\delta)$ is given in Definition 1. Therefore, in the presence of hardware errors, the quantum simulator need not be run 
beyond a time needed to obtain $\varepsilon' = f(\delta)$, and we can expect to compute the target thermodynamic limit to a
precision of $O(f(\delta))$. As summarized in Table \ref{tab:stability}, we typically obtain $f(\delta) =
\text{poly}(\delta)$ for most stable many-body simulation tasks, and
thus to be able to obtain the thermodynamic limit to a precision of
$O(\text{poly}(\delta))$, determined entirely by the hardware error
$\delta$, in quantum-simulation time $O(\text{poly}(1/\delta))$.

A numerical illustration of this analysis is shown in Fig.~\ref{fig:gapless_adb} --- here, we use the adiabatic quantum algorithm to find the energy density observable in the ground state of the critical SSH model (i.e.~Eq.~\ref{eq:SSH} with $J = 1$). Figure \ref{fig:gapless_adb}(a) shows the convergence of the energy density observable, in the absence of errors, to its thermodynamic limit --- we see that a power-law convergence is obtained, as physically expected for gapless models. In Fig.~\ref{fig:gapless_adb}(b), we use a system-size that yields a precision of $O(f(\delta))$, as determined by the stability bounds on the ground-state of this model, and simulate an adiabatic algorithm to find the ground state in the presence of hardware error. We see that, in the presence of errors, the accuracy in the achieved precision is fundamentally limited by the hardware precision $\delta$ --- Fig.~\ref{fig:gapless_adb} shows the run-time of the adiabatic algorithm as a function of this hardware-limited precision. We see that this run-time scales polynomially with $1/\varepsilon$, where $\varepsilon$ is the hardware-limited precision that is achieved by the adiabatic algorithm.
\begin{figure*}
    \centering
    \includegraphics[scale=0.3]{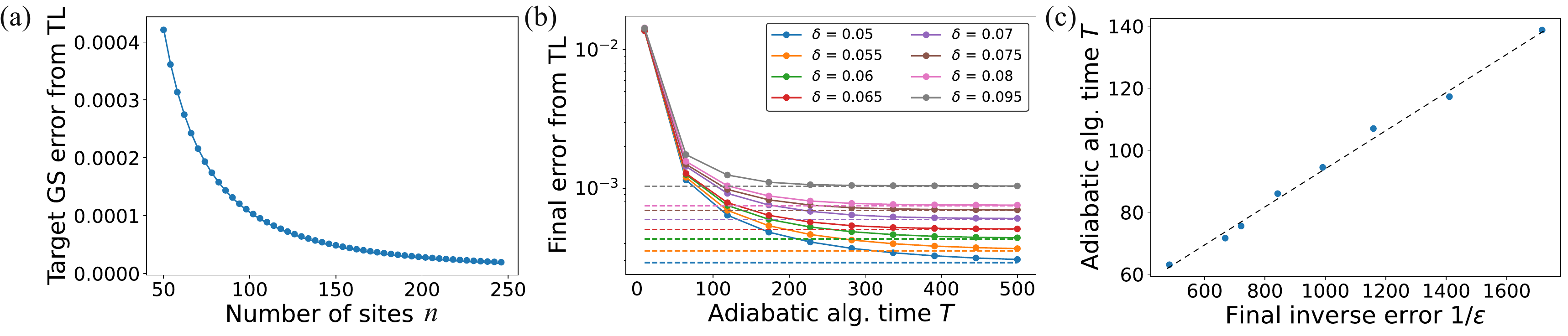}
    \caption{Numerical study of quantum adiabatic algorithm in the presence of error. We consider using the adiabatic algorithm to find the energy density observable for the critical SSH model in the thermodynamic limit (TL). (a) Convergence of the energy density to the thermodynamic limit as $n\to \infty$ --- the scaling of $\varepsilon$ with $n$ reveals a power-law scaling that is expected for gapless models. (b) The adiabatic algorithm in the presence of hardware errors --- the quantity being plotted is the error of the noisy adiabatic algorithm from the thermodynamic limit of the noiseless model. The precision achieved by the adiabatic algorithm is fundamentally limited by hardware errors. (c) The final precision ($\varepsilon$) achieved by an adiabatic algorithm in the presence of errors as a function of the adiabatic algorithm run-time, confirming that $T\sim \text{poly}(1/\varepsilon)$ as expected from our analysis. Thus, on decreasing the hardware error, the error achievable by the noisy quantum simulator decreases and the run-time of the quantum algorithm increases at most polynomially.}
    \label{fig:gapless_adb}
\end{figure*}

To define a notion of advantage in the presence of noise, we can now compare the classical and quantum run-times needed to achieve this hardware-limited precision. Assuming $f(\delta) = \text{poly}(\delta)$, it follows from the discussion in the previous subsection that we would need classical run-times that are either superpolynomial or exponential in $\text{poly}(1/\delta)$ to achieve the precision that can be achieved by quantum simulators in time $\text{poly}(1/\delta)$. That is, if $\delta$ is decreased by a constant factor, then the run-time of the quantum simulator will only increase at-most polynomially with this factor, while the run-time of the classical simulator will increase by a super-polynomial or exponential factor. We summarize our expectation of noisy quantum advantage for the quantum simulation task of finding local observables in dynamics, ground states and fixed points in the table \ref{table:run_times}. This table is based on both the stability results provided in section~\ref{sec:stability}, as well as the classical and quantum-time analysis for many-body models performed in section \ref{sec:ideal_qsim}.

We emphasize the following two points regarding table~\ref{table:run_times}. \emph{First}, the ``provable" quantum advantage for dynamics, stably gapped ground states with logarithmic convergence, and fixed points of rapidly mixing Lindbladians is only with respect to Krylov subspace methods or exact methods, which are the only known classical algorithms with rigorous guarantees known for these models. To the best of our knowledge, it remains an open problem to provide a universal lower bound on \emph{any possible} classical algorithm or even connect a possible quantum advantage for the specific many-body problems that we consider to well-known complexity assumptions. \emph{Second}, we only conjecture the quantum advantage for critical models (both ground states and fixed points) --- our conjecture is based on the novel stability results that we provide for Gaussian fermion models and under the expectation that these models could capture the qualitative physics of more complex non-Gaussian models, even though Gaussian fermion models on their own can be solved efficiently on classical computers.

\begin{table*}
    \begin{tabular}{p{2.5cm} L{2.5cm} L{4cm} L{3.2cm}  L{3.7cm}}
    \hline 
       \textbf{Problem} & 
       \textbf{Error model} &
       \textbf{Assumption} &
       \textbf{Classical run-time } &
       \textbf{Quantum run-time }\\ \hline 
        \multicolumn{3}{c}{\emph{Provable noisy quantum advantage over Krylov Subspace/Exact methods}} \\ \hline \\ 
        Dynamics &
        General errors & 
        None &
        $\exp(\tilde{\Omega}(\log^d(\Theta(\delta^{-1})))) $&
        $O(\text{poly}(\delta^{-1}))$ \\ \ \\
            
        Ground state  &
        Coherent Hamiltonian errors &
        \begin{itemize*}
        \item Stable gap
        \end{itemize*}
        \begin{itemize}[leftmargin=*]
        \item Logarithmic convergence
        \item Adiabatic path with a gap $>\Omega(1/\text{poly}(n))$
        \end{itemize}&
        $\exp({\Omega}(\log^d(\Theta(\delta^{-1}))))$ &
        $O(\text{poly}(\delta^{-1}))$ \\ \ \\

        Fixed points & 
        General errors &
        Rapid mixing &
        $\exp({\Omega}(\log^d(\Theta(\delta^{-1}))))$ &
        $O(\text{poly}(\delta^{-1}))$  \\ \ \\ \hline 
        \multicolumn{2}{c}{\emph{Conjectured noisy quantum advantage}}\\ \hline \\

        Ground states &
       Coherent Hamiltonian errors&
        \begin{itemize*}
        \item Stable $\Omega(1/\text{poly}(n))$ gap
        \end{itemize*}
        \begin{itemize}[leftmargin=*]
        \item Power-law convergence
        \item Adiabatic path with a gap $>\Omega(1/\text{poly}(n))$
        \end{itemize} &
         $\exp({\Omega}(\text{poly}(\delta^{-1})))$ &
        $O(\text{poly}(\delta^{-1}))$\\ \ \\

        Fixed points &
        Coherent and incoherent Markovian errors & 
        \begin{itemize*}
            \item Reaches $\varepsilon-$close to fixed point in $O(\text{poly}(n, 1/\varepsilon))$ time
        \end{itemize*}
        \begin{itemize}[leftmargin=*]
            \item Power-law convergence
        \end{itemize}
        & $\exp({\Omega}(\text{poly}(\delta^{-1})))$ 
        & $O(\text{poly}(\delta^{-1}))$\\ \ \\

    \hline
    \end{tabular}
    \caption{Summary of the classical and quantum run-times for thermodynamic limits needed to obtain the hardware-limited precision with hardware error/noise $\delta$. For classical run-times, we only consider Krylov subspace methods or exact diagonalization as the classical algorithm for the provided lower bounds, and not heuristics which do not have rigorous convergence guarantees. Note that $\tilde{\Omega}$ suppresses $\log \log(\delta^{-1})$ factors.}
    \label{table:run_times}
\end{table*}

\section{Conclusion}
We have considered both the stability and quantum advantage of using near-term analogue quantum simulators for thermodynamic limits of many-body problems in physics. Based on both existing theoretical results in many-body literature, and new technical results for free-fermion models, we argue that many physically relevant many-body problems are stable to a constant rate of error on the quantum hardware being used to solve them and thus are accessible in near-term experiments. We  also hypothesize that these algorithms have an advantage, with respect to the obtained precision, in computing thermodynamic limits of many-body problems. Our formulation and results provides some evidence for near-term analogue quantum simulators being useful for solving many-body problems.

Extending the stability results for gapless/critical models to the case of quantum spins, or non-Gaussian fermionic systems is an immediate open problem suggested by our work. While previous work by Hastings \cite{hastings2005quasiadiabatic} indicates that, under some assumption on the density of states of the many-body model, such a stability result could hold for gapless spin systems, it would be interesting to see if restricting observables to being translationally invariant could help improve these results. Similarly, understanding the stability of Lindbladian dynamics and fixed point problems for quantum spin systems or non-Gaussian fermionic systems beyond the rapid mixing assumption would also be an important extension of our work.

\begin{acknowledgements}
We acknowledge useful discussions from Dorit Aharonov and Álvaro M. Alhambra. This research is funded by the German Federal Ministry of Education and Research (BMBF) through EQUAHUMO (Grant No. 13N16066) within the funding program quantum technologies - from basic research to market and by the Munich Quantum Valley (MQV), which is supported by the Bavarian state government with funds from the Hightech Agenda Bayern Plus. R.T. also acknowledges funding from the Max-Planck Harvard Research Center for Quantum Optics (MPHQ) postdoctoral fellowship. A.F.R. is supported by the Alexander von Humboldt Foundation through a postdoctoral fellowship.

\end{acknowledgements}

\bibliography{references.bib}

\newpage
\onecolumngrid
\appendix

\section{Notational preliminaries}
The following is a list of notations that we will use for the mathematical proofs in the following appendices. Here $v$ denotes a vector, and $A$ a matrix.

\begin{itemize}
    \item $\|v\|_p\equiv \left(\sum_i{v_i^p}\right)^{\frac{1}{p}}$ denotes the $p$-norm of the vector $v$. In the $p\to\infty$ limit, it becomes the max norm,  ${\|v\|_{\infty}\equiv\max_{i}{|v_i|}}$.
    \item $\norm{A}_{\op, p}$ is the $p^\text{th}$ Schatten norm of $A$ i.e.~the $p-$norm of the singular values of $A$. Note that $\norm{A}_{\op} := \norm{A}_{\op, \infty}$ is also the operator norm of $A$, i.e.~the $\norm{A}_{\op, \infty} = \sup_{x, \norm{x}_2 = 1}\norm{Ax}_2$ and $\norm{A}_{\op,1}$ denotes the trace norm, i.e. the 1-Schatten norm of $A$.
    \item For a superoperator $\mathcal{A}$, we define the $\norm{\mathcal{A}}_{p \to p} = \max_{O, \norm{O}_{\op, p} = 1} \norm{\mathcal{A}(O)}_{\op, p}$. We define the completely bounded version of this norm, $\norm{\mathcal{A}}_{p \to p, cb} = \sup_{n \geq 2} \norm{\mathcal{A} \otimes \textnormal{id}_n}_{p \to p}$, which is stable under tensor product. Furthermore, as is standard, we will use $\norm{\mathcal{A}}_\diamond := \norm{\mathcal{A}}_{1 \to 1, cb}$.
    \item $\text{vec}(A)$ denotes the vectorization of $A$, i.e. the vector whose components are the matrix elements of $A$. The precise order in which the matrix elements are arranged in the vector will not be relevant for our proofs, and can be arbitrarily chosen.
    \item Unless otherwise mentioned, $\norm{v}$, where $v$ is a vector, will denote its $\ell^2$ norm and $\norm{O}$, where $O$ is an operator, will be its operator norm.
\end{itemize}

In our analysis of Gaussian fermion models, we will assume that for each site $x \in \mathbb{Z}_L^d$, there are $D$ fermionic modes. Suppose $a_x^n$, for $x \in \mathbb{Z}_L^d$ and $n \in \{1, 2 \dots D\}$ be the fermionic annihilation operator corresponding to the $n^\text{th}$ fermion mode at the site $x$. The Majorana operators, $c_x^{2n - 1}, c_x^{2n}$ associated with this mode will be 
\[
c_x^{2n - 1} = \frac{1}{\sqrt{2}}\big(a_x^{n^\dagger} + a_x^n\big) \text{ and }c_x^{2n} = \frac{1}{\sqrt{2} i}\big(a_x^{n^\dagger} - a_x^n\big).
\]
In most of the proofs, we will not need to distinguish between the two Majorana operators and will represent them as $c_x^\alpha$ with the index $\alpha \in \{1, 2 \dots 2D\}$. For a Hermitian operator $O$ expressed as a quadratic form over the Majorana operators $c_x^\alpha$,
\begin{align}\label{eq:quadH}
O = \sum_{x, y \in \mathbb{Z}^d_L}\sum_{\alpha, \beta = 1}^{2D} o_{x, y}^{\alpha, \beta}c_x^\alpha c_y^\beta,
\end{align}
we will denote by $\tilde{O}$ the matrix of coefficients $o_{x, y}^{\alpha, \beta}$, with the indices $(x, \alpha)$ corresponding to the rows and $(y, \beta)$ corresponding to the columns. We will assume, without loss of generality and unless otherwise mentioned, that $\tilde{O}$ is a Hermitian matrix with purely imaginary matrix elements.

Unless otherwise mentioned, $n$ will be used for the number of spins or fermionic modes in the lattice system under consideration. For Gaussian fermionic models considered in this paper, defined on the lattice $\mathbb{Z}_L^d$ and with $D$ fermions per site, we will use $N = L^d$ to denote the number of lattice sites and therefore $n = D N$.

We will need the following two lemmas:
\begin{lemma}\label{lemma:bound_op_norm}
Let $M \in \mathbb{C}^{n \times m}$ such that $|M_{ij}|\leq \delta, \forall i,j$, the rows of $M$ have at most $m_r$ nonzero elements and the columns of $M$ have at most $m_c$ nonzero elements, then $\norm{M}_{\op}\leq \sqrt{m_c m_r} \delta$.
\end{lemma}
\begin{proof} 
Let $v \in \mathbb{C}^{m}$ be a vector, and denote $\mc R_i\equiv\{j|M_{ij}\neq 0\}$, $\mc C_j\equiv\{i|M_{ij}\neq 0\}$. Note that, by assumption, $\abs{\mathcal{R}_i}\leq m_r \ \forall i, \abs{\mathcal{C}_j} \leq m_c \ \forall j$. Then,
\begin{align*}
    \|Mv\|^2 &=\sum_{i=1}^n{\bigg|\sum_{j\in \mc R_i}{M_{ij}v_j}\bigg|^2}\nonumber\\
    &\leq\sum_{i=1}^{ n}{\left(\sum_{j\in\mc R_i}{|M_{ij}|^2}\sum_{k\in\mc R_i}{|v_k|^2}\right)}\nonumber\\
    &\leq \sum_{i=1}^{n} \abs{\mathcal{R}_i} \delta^2 \sum_{k\in\mc R_i}{|v_k|^2}\leq m_r\delta^2\sum_{i=1}^{ n}\sum_{k\in\mathcal{R}_i}{|v_k|^2} \nonumber\\
    &= m_r\delta^2\sum_{k=1}^{ n}\sum_{i\in\mc C_k}{|v_k|^2} = m_r\delta^2\sum_{k=1}^{ n}\abs{\mathcal{C}_k}{|v_k|^2}\\
    &\leq m_r m_c\delta^2\|v\|^2\implies \|M\|_\op\leq \sqrt{m_r m_c} \delta.
\end{align*}
\end{proof}\begin{lemma}\label{lemma:pert_theory} Given bounded Hermitian operators $H$ and $H'$, for any bounded operator $O$
\[
\bignorm{e^{iH' t} O e^{-iH't} - e^{iHt}Oe^{-iHt}}_{\op} \leq 2\norm{O}_\textnormal{op}\norm{H - H'}_\textnormal{op}t.
\]
\end{lemma}
\begin{proof}
Consider the operator $\tilde{O}(t) \equiv e^{-iHt} e^{iH't} O e^{-iH't} e^{iHt}$. Note that
\[
\frac{d}{dt}\tilde{O}(t) = i\bigg[e^{-iHt} (H' - H) e^{iH't} Oe^{-iH't}e^{iHt} -  e^{-iHt} e^{iH't} Oe^{-iH't}(H' - H)e^{iHt}\bigg],
\]
and consequently,
\[
\bignorm{\frac{d}{dt}\tilde{O}(t)} \leq 2\norm{H - H'}\norm{O}.
\]
We then immediately obtain that
\[
\bignorm{e^{iH' t} O e^{-iH't} - e^{iHt}Oe^{-iHt}}\leq \int_0^t \bignorm{\frac{d}{ds}\tilde{O}(s)} ds \leq 2\norm{O}\norm{H - H'} t.
\]
\end{proof}
\section{Gaussian fermion models}
\subsection{Proof of proposition \ref{prop:t_evol_ff} (Dynamics of free fermion models)}
\label{app:t_evol_ff}
We first derive an equation of motion for the correlation matrix $\Gamma$ of the Gaussian fermionic quantum simulator model, including the non-Markovian decohering errors, described in section~\ref{sec:stability_gaussian_fermion}. It is convenient to consider a system-environment Hamiltonian $\hat{H}(t)$
\[
\hat{H}(t) = H + \sum_{j, x} \bigg(B_{j, x}^\dagger(t) L_{j, x} + L_{j, x}^\dagger B_{j, x}(t)\bigg) + \sum_{j, x} \bigg(A_{j, x}^\dagger(t) Q_{j, x} + Q_{j, x}^\dagger A_{j, x}(t)\bigg),
\]
where $\{B_{j, x}(t), B_{j', x'}^\dagger(t')\} = \delta_{x, x'} \delta_{j, j'} \delta(t - t')$ and, as described in section~\ref{sec:stability_gaussian_fermion}, $\{A_{j, x}(t), A_{j', x'}^\dagger(t')\} = \delta_{j, j'}\delta_{x, x'}K_{j, x}(t - t')$. Furthermore, we also assume that the environments corresponding to the operators $B_{j, x}(t)$ and $A_{j, x}(t)$ are independent by enforcing $\{A_{j, x}(t), B_{j', x'}^\dagger (t')\} = 0$. Since the operators $B_{j, x}(t)$ are delta function correlated, it can easily be verified that tracing out the environment correspnding to these operators effectively yields a master equation with jump operators $L_{j, x}$. Now, we derive a set of dynamical equations for the correlation matrix
\[
(\Gamma(t))_{x, y}^{\alpha, \beta} =\frac{1}{2} \text{Tr}\big([c_x^\alpha, c_y^\beta] \rho(t)\big) =\frac{1}{2}\text{Tr}\big([c_x^\alpha(t), c_y^\beta(t)] \rho_0\big) = \text{Tr}\big(c_x^\alpha(t) c_y^\beta(t) \rho_0\big) - \frac{1}{2}\delta_{\alpha, \beta}\delta_{x, y}.
\]
where $c_x^\alpha(t) = \hat{U}(0, t) c_x^\alpha \hat{U}(t, 0)$, with $\hat{U}(t, s)$ being the propagator corresponding to $\hat{H}(t)$ and $\rho_0$ is the initial system-environment state, which we assume to be vacuum for the environment. It is convenient to introduce the two-point correlation matrix
\[
\big(\Lambda(t, s)\big)_{x, y}^{\alpha, \beta} = \text{Tr}\big(c_x^\alpha(t)c_y^\beta(s)\rho_0)
\]
It can be noted that $\Gamma(t) = \Lambda(t, t) - I_{2n}/2$ and $\Lambda^{T}(t, t) = {I}_{2n} - \Lambda(t, t)$. 
\begin{lemma}\label{lemma:gaussian_open_quantum_dynamics}
Assuming that the environment is initially in the vacuum state, the correlation matrix $\Gamma(t)$ satisfies,
\[
\frac{d}{dt}\Gamma(t) = X\Gamma(t) + \Gamma(t) X^\text{T} + Y + Z(t),
\]
where
\begin{align*}
&X = -i\tilde{H} - \frac{1}{2}(\tilde{Q} + \tilde{Q}^*), \\
&Y = \frac{1}{2}\big(\tilde{Q} - \tilde{Q}^*\big), \\
&Z(t) = \frac{1}{2}\int_0^t \bigg(\tilde{L}^*(s - t)\Lambda^\text{T}(t, s) + \Lambda^{\text{T}}(t, s) \tilde{L}^\dagger(s - t) - \tilde{L}(t -s) \Lambda(s, t) - \Lambda(s, t) \tilde{L}^\text{T}(t- s)\bigg) ds,
\end{align*}
with $\tilde{Q}, \tilde{L}(\tau) \in \mathbb{C}^{2n \times 2n}$ being given by,
\begin{align*}
\big(\tilde{Q}\big)^{\alpha, \beta}_{x, y} = \sum_{j, z}q_{j, z; x}^{\alpha} q^{\beta^*}_{j, z; y} \text{ and }\big(\tilde{L}(\tau)\big)^{\alpha, \beta}_{x, y} = \sum_{j, z}K_{j, z}(\tau) l_{j, z; x}^{\alpha} l^{\beta^*}_{j, z; y}
\end{align*}\end{lemma}
\noindent\emph{Proof}: Our starting point is the Heisenberg equations of motion for $c_x^\alpha(t)$,
\begin{align}\label{eq:cx_derivative}
i\frac{d}{dt}c_x^\alpha(t) = 2\sum_{y, \beta} h_{x, y}^{\alpha, \beta} c_y^\beta(t) + \sum_{j, z} \bigg( q^{\alpha*}_{j, z; x} A_{j, z}(t; t) - q_{j, z; x}^{\alpha} A_{j, z}^\dagger(t; t)\bigg) +  \sum_{j, z} \bigg( l^{\alpha*}_{j, z; x} B_{j, z}(t; t) - l_{j, z; x}^{\alpha} B_{j, z}^\dagger(t; t)\bigg),
\end{align}
where $A_{j, z}(\tau; t) = U(0, t) A_{j, z}(\tau) U(t, 0)$ and $B_{j, z}(\tau; t) = U(0, t) B_{j, z}(\tau) U(t, 0)$. Furthermore, we can also obtain and integrate the Heisenberg equations of motion for $A_{j, z}(\tau; t)$ and $B_{j, z}(\tau; t)$ to obtain
\begin{align*}
   &\frac{d}{dt} A_{j, z}(\tau; t) = -i\sum_{\beta, y}q^{\beta}_{j, z; y} c_y^\beta(t) \delta(t - \tau) \implies A_{j, z}(t; t) = A_{j, z}(t) - \frac{i}{2}\sum_{\beta, y}q^{\beta}_{j, z; y}c^\beta_y(t) \text{ and }\\
   &\frac{d}{dt}B_{j, z}(\tau; t) = -i\sum_{\beta, y} l^{\beta}_{j, z; y}c_y^\beta(t) K_{j, z}(\tau - t) \implies B_{j, z}(t; t) = B_{j, z}(t) - i \sum_{\beta, y} q^{\beta}_{j, z; y}\int_0^t K_{j, z}(t - s) c_y^\beta(s) ds.
\end{align*}
Since $\rho_0$ is assumed to be in the vacuum state in the environments corresponding to annihilation operators $A_{j, x}(t)$ and $B_{j, x}(t)$, it then follows that
\begin{align}\label{eq:input_out_eq_vacuum}
A_{j, z}(t; t)\rho_0 =  - \frac{i}{2}\sum_{\beta, y}q^{\beta}_{j, z; y}c^\beta_y(t)\rho_0 \text{ and } B_{j, z}(t; t) \rho_0 = - i \sum_{\beta, y} q^{\beta}_{j, x; y}\int_0^t K_{j, z}(t - s) c_y^\beta(s) \rho_0 ds.
\end{align}
Now, we can obtain a differential equation for the correlation matrix element $(\Gamma(t))^{\alpha, \beta}_{x, y}$ --- from its definition, it follows that
\begin{align*}
    \frac{d}{dt}\big(\Gamma(t)\big)^{\alpha, \beta}_{x, y} &=\frac{1}{2} \text{Tr}\bigg(\frac{d}{dt}c_x^\alpha(t) c_y^\beta(t) \rho_0\bigg) +\frac{1}{2}\text{Tr}\bigg(c_x^\alpha(t)\frac{d}{dt} c_y^\beta(t) \rho_0\bigg),\\
    &=-i\big([\tilde{H}, \Gamma(t)]\big)_{x, y}^{\alpha, \beta} + \frac{i}{2} \sum_{j , z} \bigg( q^{\alpha}_{j, z; x}\text{Tr}\big( c_y^\beta(t)\rho_0 A_{j, z}^\dagger(t; t)\big) + q^{\alpha^*}_{j, z; x} \text{Tr}\big( c_y^\beta(t)A_{j, z}(t; t) \rho_0\big)\bigg).
\end{align*}
Furthermore, using Eqs.~\ref{eq:cx_derivative} and \ref{eq:input_out_eq_vacuum}, we obtain that
\begin{align}\label{eq:product_rule_term_1}
&\text{Tr}\bigg(\frac{d}{dt}c_x^\alpha(t) c_y^\beta(t) \rho_0\bigg)  \nonumber\\
&\qquad  =-i\big(\tilde{H}\Lambda(t, t)\big)_{x, y}^{\alpha, \beta} + i \sum_{j , z} \bigg( q^{\alpha}_{j, z; x}\text{Tr}\big( c_y^\beta(t)\rho_0 A_{j, z}^\dagger(t; t)\big) + q^{\alpha^*}_{j, z; x} \text{Tr}\big( c_y^\beta(t)A_{j, z}(t; t) \rho_0\big)\bigg) \nonumber\\
&\qquad \qquad \quad + i\sum_{j, z} \bigg(l^{\alpha}_{j, z; x}\text{Tr}\big(c_y^\beta(t) \rho_0 B^{\dagger}_{j, z}(t; t)\big) + l^{\alpha^*}_{j, z; x}\text{Tr}\big(c_y^\beta(t)B_{j, z}(t; t) \rho_0\big)\bigg), \nonumber\\
&\qquad  =-i\big(\tilde{H}\Lambda(t, t)\big)_{x, y}^{\alpha, \beta} - \frac{1}{2}\sum_{{j, z, \alpha',x'}}\bigg(q^{\alpha}_{j,z ; x}q^{\alpha'^*}_{j, z; x'} \text{Tr}\big(c_{x'}^{\alpha'}(t) c_y^\beta(t)\rho_0\big) - q^{\alpha^*}_{j, z; x}q^{\alpha'}_{j, z; x'} \text{Tr}\big(c_y^\beta(t) c^{\alpha'}_{x'}(t)\rho_0\big)\bigg) \nonumber \\
&\qquad \qquad \qquad  -\frac{1}{2} \sum_{j, z, \alpha', x'} \int_0^t K_{j, z}(t - s) \bigg(l^{\alpha}_{j, z; x}l^{\alpha'^*}_{j, z; x'} \text{Tr}\big(c_{x'}^{\alpha'}(s) c_y^\beta(t)\rho_0\big) - l^{\alpha^*}
_{j, z; x}l^{\alpha'}_{j, z; x'} \text{Tr}\big(c_y^\beta(t) c_{x'}^{\alpha'}(s) \big)\bigg)ds, \nonumber \\
&\qquad = \bigg(-i\tilde{H}\Lambda(t, t) - \frac{1}{2}\big( \tilde{Q} \Lambda(t, t) - \tilde{Q}^*\Lambda^\text{T}(t, t)\big) - \frac{1}{2} \int_0^t \big(\tilde{L}(t - s)\Lambda(s, t)- \tilde{L}^*(s - t)\Lambda^{\text{T}}(t, s)\big)ds\bigg)_{x, y}^{\alpha, \beta},
\end{align}
A similar manipulation yields that
\begin{align}\label{eq:product_rule_term_2}
    &\text{Tr}\bigg(c_x^\alpha(t) \frac{d}{dt}c_y^\beta(t) \rho_0\bigg) \nonumber \\
    &\qquad=\bigg(-i\Lambda(t, t) \tilde{H}^\text{T} - \frac{1}{2}\big(\Lambda(t, t) \tilde{Q}^\text{T} - \Lambda^\text{T}(t, t) \tilde{Q}^\dagger\big) - \frac{1}{2}\int_0^t \big(\Lambda(s, t) \tilde{L}^\text{T}(t - s) - \Lambda^\text{T}(t, s) \tilde{L}^\dagger(s - t)\big)ds\bigg)^{\alpha, \beta}_{x, y}
\end{align}
From Eqs.~\ref{eq:product_rule_term_1} and \ref{eq:product_rule_term_2}, together with the facts that $\Lambda(t, t) = \Gamma(t) + I_{2n}/2, \Lambda^\text{T}(t, t) = -\Gamma(t) + I_{2n}/2$, we then obtain the dynamical equation for $\Gamma(t)$. $\hfill \square$

Suppose that in the absence of any errors and noise, our target is to implement a spatially local Hamiltonian $H$ and jump operators $Q_{j, x}$. In the presence of errors and noise, we instead implement a perturbed Hamiltonian $H'$ described by coefficients $h_{x, y}^{\alpha, \beta'}$, perturbed jump operators $Q_{j, x}'$ described by coefficients ${q'}_{j, x; y}^{\alpha}$, as well as interaction with a decohering environment captured by the operators $L_{j,x}$ described by coefficients $l_{j, x; y}^{\alpha}$ which satisfy
\[
\abs{h_{x, y}^{\alpha, \beta} - {h'}_{x, y}^{\alpha, \beta}} \leq \delta, \abs{q_{j, x; y}^{\alpha} - {q'}_{j, x; y}^{\alpha}} \leq \delta \text{ and }\abs{l_{j, x; y}^{\alpha}}\leq \sqrt{\delta}.
\]
\begin{lemma}\label{lemma:bounds_geom_mat}
Let $a_{x, \alpha}^{y, \beta}, {a'}_{x, \alpha}^{y, \beta} \in \mathbb{C}$ for $x, y \in \mathbb{Z}_L^d$, $\alpha, \beta \in \{1, 2 \dots 2D\}$ satisfy $a_{x, \alpha}^{y, \beta}, {a'}_{x, \alpha}^{y, \beta} = 0$ if $ d(x, y) > R$ and $b_{j, x; y}^{\alpha}, {b'}_{j, x; y}^{\alpha}  \in \mathbb{C}$ for $x, y \in \mathbb{Z}_L^d$, $j \in \{1,2 \dots n_L\}, \alpha \in \{1, 2 \dots 2D\}$ be such that $b_{j, x; y}^\alpha, {b'}_{j, x; y}^\alpha = 0$ if $d(x, y) > R$. Furthermore, $\exists a_0 > 0: \abs{a_{x, \alpha}^{y, \beta}} \leq a_0 \ \forall x, y, \alpha, \beta$, $\exists b_0 > 0: \abs{b_{j, x; y}^\alpha} \leq b_0 \ \forall x, y, j, \alpha$ and $\exists \delta > 0 : \abs{a_{x, \alpha}^{y, \beta} - {a'}_{x, \alpha}^{y, \beta}}, \abs{b_{j, x; y}^{\alpha} - {b'}_{j, x; y}^{ \beta}} \leq \delta \ \forall x, y, j, \alpha, \beta$. Denote by $A, A' \in \mathbb{C}^{2n \times 2n}$ the matrices formed by $a_{x,\alpha}^{y, \beta}, {a'}_{x, \alpha}^{y, \beta}$ with $(x, \alpha)$ corresponding to the rows and $(y, \beta)$ corresponding to the columns, and by $B, B' \in \mathbb{C}^{n_L N \times 2n}$ the matrices formed by $b_{j, x; y}^{\alpha}, {b'}_{j, x; y}^{\alpha}$ with $(j, x)$ corresponding to the rows and $(y, \alpha)$ corresponding to the columns. Then,
\begin{align*}
&\norm{A} \leq 2D(2R + 1)^d a_0,   \norm{B} \leq \sqrt{2Dn_L} (2R + 1)^d b_0 , \\
&\norm{A - A'} \leq 2D(2R + 1)^d \delta, \norm{B^\dagger B - B'^\dagger B'} \leq 4Dn_L (2R + 1)^{2d}(2b_0 \delta + \delta^2).
\end{align*}
\end{lemma}
\begin{proof}
The proof of this lemma is a repeated application of lemma \ref{lemma:bound_op_norm}. We note that the matrices $A, A'$ have at most $2D(2R + 1)^d$ non-zero element in any row or column. The matrices $B, B'$ have at most $n_L (2R + 1)^d$ non-zero elements in their columns and at most $2D (2R + 1)^d$ non-zero elements in their rows. Thus, from lemma \ref{lemma:bound_op_norm}, we obtain the bounds
\[
\norm{A} \leq 2D(2R + 1)^d a_0, \norm{A'} \leq 2D(2R + 1)^d a_0', \norm{B} \leq \sqrt{2Dn_L}(2R + 1)^d b_0, \norm{B'} \leq \sqrt{2Dn_L} (2R + 1)^d b_0'.
\]
where $a_0' = a_0 + \delta$ and $b_0' = b_0 + \delta$ are upper bounds on the coefficients $\abs{{a'}_{x, \alpha}^{y, \beta}}$ and $\abs{{b'}_{j, x; y}^{\alpha}}$ respectively. Since $A - A'$ also has at most $2D(2R + 1)^d$ non-zero element in any row or column, it immediately follows that
\[
\norm{A - A'} \leq 2D(2R + 1)^d\delta.
\]
Furthermore, since $B - B'$ has at most $n_L (2R + 1)^d$ non-zero elements in its columns and at most $2D (2R + 1)^d$ non-zero elements in its rowsm
\[
\norm{B - B'} \leq \sqrt{2Dn_L}(2R + 1)^d \delta.
\]
We can now estimate $\norm{B^\dagger B - B'^\dagger B'}$:
\begin{align*}
    \norm{B^\dagger B - B'^\dagger B'} &= \norm{B^\dagger (B - B') + (B^\dagger - B'^\dagger)B'},\nonumber\\
    &\leq (\norm{B} + \norm{B'})\norm{B - B'}, \nonumber \\
    &\leq 4Dn_L(2R + 1)^{2d}(2b_0\delta + \delta^2).
\end{align*}
\end{proof}

Finally, we provide an upper bound on $Z(t)$ defined in lemma \ref{lemma:gaussian_open_quantum_dynamics} --- for this, we assume an upper bound on the kernels are $K_{j, z}(\tau) = \{B_{j, z}(\tau), B_{j, z}^\dagger(0)\}$. We consider kernels which can also contain delta functions i.e.~$K_{j, z}(\tau)$ to be of the form
\[
K_{j, z}(\tau) = K_{j, z}^c(\tau) + \sum_{i = 1}^M k_{j, z}^{i} \delta(\tau - \tau_i),
\]
where $K_{j, z}^c(\tau)$ is a continous function of $\tau$. Then, we define a kernel $K(\tau)$ by
\[
K(\tau) = K^c(\tau) + \sum_{i = 1}^M k^{i}\delta(\tau - \tau_i) \text{ where }K^c(\tau) = \sup_{j, z} \abs{k_{j, z}^c(\tau)} \text{ and }k^i = \sup_{j, z}\abs{k^i_{j, z}}.
\]
The kernel $K(\tau)$ can be considered as a distributional upper bound on $K_{j, z}(\tau)$ i.e.~for any continuous and compactly supported function $f$,
\[
\bigabs{\int_{\mathbb{R}} K_{j, z}^c(\tau) f(\tau) d\tau} \leq \int_{\mathbb{R}} K(\tau) 
\abs{f(\tau)}d\tau.
\]
In the following lemma, under the assumption that $K(\tau)$ has a bounded integral, we provide an upper bound on $Z(t)$.

\begin{lemma}\label{lemma:bound_z}
If $\int_{\mathbb{R}}K(\tau) d\tau \leq 1 $ and $\abs{l_{j, z; x}^{\alpha}} \leq \sqrt{\delta}$, then $\norm{Z(t)} \leq 2\sqrt{2Dn_L}(2R + 1)^d \delta$, where $Z(t)$ is defined in lemma \ref{lemma:gaussian_open_quantum_dynamics}.
\end{lemma}
\begin{proof} First, we will establish that $\norm{\Lambda(t, s)} \leq 1$ for any $t, s$. For this, we consider two vectors $v, u \in \mathbb{C}^{2n}$ and upper bound $\abs{v^\dagger \Lambda(t, s) u}$. From the definition of $\Lambda(t, s)$, we have that
\begin{align}\label{eq:variational_sigma_Lambda}
\abs{v^\dagger \Lambda(t, s) u} = \bigabs{\sum_{\substack{x, y}} \sum_{\alpha, \beta = 1}^{2D} v_{x}^{\alpha^*}\text{Tr}[c_x^\alpha(t) c_y^\beta(s) \rho_0] u_y^\beta} \leq \bigabs{\text{Tr}(c_{v}^\dagger(t) c_u(s) \rho_0)} \leq \norm{c_{v}(t)}\norm{c_u(s)} = \norm{c_{v}}\norm{c_u} ,
\end{align}
where we have defined the operator $c_v = \sum_{x}\sum_{\alpha = 1}^{2D} v_{x}^\alpha c_{x}^\alpha$. Now, to bound $\norm{c_v}$, we note that $\{c_v, c_v^\dagger\} = \norm{v}^2 I$, and therefore
\begin{align}\label{eq:bound_norm_cv}
c_v^\dagger c_v = \norm{v}^2 I - c_v c_v^\dagger \prec \norm{v}^2 I \implies \norm{c_v} \leq \norm{v}.
\end{align}
Therefore, from Eqs.~\ref{eq:variational_sigma_Lambda} and \ref{eq:bound_norm_cv}, we obtain that $\norm{\Lambda(t, s)} = \sup_{v, u \in \mathbb{C}^{2n}} \abs{v^\dagger \Lambda(t, s) u} / \norm{v}\norm{u} \leq 1$.

Let us now consider bounding $Z(t)$ defined in lemma \ref{lemma:gaussian_open_quantum_dynamics}. We will provide a detailed analysis for an upper bound on one of the four terms appearing in $Z(t)$ --- the other terms can be bounded similarly. Consider the first term --- we have that
\[
\bignorm{\int_0^t \tilde{L}^*(s - t) \Lambda^\text{T}(t, s) ds} = \sup_{\substack{v, u \in \mathbb{C}^{2n} \\ \norm{v}, \norm{u} = 1}} \bigabs{\int_0^t v^\dagger \tilde{L}^*(s - t) \Lambda^\text{T}(t, s) u ds} = \bigabs{\sup_{\substack{v, u \in \mathbb{C}^{2n} \\ \norm{v}, \norm{u} = 1}} \int_0^t v^\dagger \tilde{L}^*(s - t) w(s)ds},
\]
where $w(s) = \Lambda^\text{T}(t, s) u$.
Now, by the definition of $K(s)$, we obtain
\begin{align*}
\bigabs{\int_0^t v^\dagger \tilde{L}^*(s - t) w(s) ds} &= \bigabs{\int_0^t \sum_{j, z} \sum_{x, \alpha}\sum_{x', \alpha'} K_{j, z}^*(s - t) v_x^{\alpha^*}l_{j,z; x}^{\alpha^*} l_{j, z; x'}^{\alpha'} w_{x'}^{\alpha'}(s) ds},\nonumber\\
&\leq \int_0^t K(s - t)\sum_{j, z}  \bigabs{ \sum_{x, \alpha}\sum_{x', \alpha'}  v_x^{\alpha^*}l_{j,z; x}^{\alpha^*} l_{j, z; x'}^{\alpha'} w_{x'}^{\alpha'}(s) }ds, \nonumber\\
&\leq \int_0^t K(s - t) \abs{v}^\text{T} \tilde{L}_m \abs{w(s)} ds,
\end{align*}
where $\abs{v}$ and $\abs{w(s)}$ are vectors formed by taking the absolute values of $v$ and $w(s)$ respectively, and $\tilde{L}_m \in \mathbb{R}^{2n \times 2n}$ is a matrix given by
\[
\big(\tilde{L}_m\big)_{x, \alpha}^{x', \alpha'} = \sum_{j, z} \abs{l_{j, z; x}^\alpha} \abs{l_{j, z; x'}^{\alpha'}}.
\]
Now, from lemma \ref{lemma:bounds_geom_mat}, it follows that $\norm{\tilde{L}_m} \leq \sqrt{2Dn_L}(2R + 1)^d \delta$, and consequently $\abs{v}^\text{T}\tilde{L}_m \abs{w(s)} \leq \sqrt{2Dn_L}(2R + 1)^d \delta \norm{v} \norm{w(s)} \leq \sqrt{2Dn_L}(2R + 1)^d \delta \norm{v} \norm{u}$, where we have used the previously shown fact of $\norm{\Lambda(t, s)} \leq 1$. Thus, we have that
\[
\bigabs{\int_0^t v^\dagger \tilde{L}^*(t - s) w(s) ds }\leq \sqrt{2Dn_L}(2R + 1)^d \delta \norm{v} \norm{u} \int_0^t K(s - t)ds \leq  \sqrt{2Dn_L}(2R + 1)^d \delta \norm{v} \norm{u},
\]
and therefore
\[
\bignorm{\int_0^t \tilde{L}^*(t - s) \Lambda^\text{T}(t, s) ds } \leq \sqrt{2Dn_L}(2R + 1)^d \delta.
\]
Performing a similar analysis to the remaining three terms in $Z(t)$, we obtain that $\norm{Z(t)} \leq 2\sqrt{2Dn_L}(2R + 1)^d \delta$.
\end{proof}

% we first note that $c_v = (a_v + b_v^\dagger)/\sqrt{2}$ where $a_v = \sum_{x} \sum_{n = 1}^{D} \big(v_{x}^{2n - 1} + i v_x^{2n}\big) a_x^n $ and $b_v = \sum_{x} \sum_{n = 1}^{D} \big({v_{x}^{2n - 1}}^* + i {v_x^{2n}}^*\big) {a_x^n}$. Now, utilizing the fact that $a_v$ can be considered to be an unnormalized annihilation operator corresponding to another fermionic mode which has a particle number operator with maximum eigenvalue 1, we have that
% \[
% \norm{a_v} = \sqrt{\norm{a_v^\dagger a_v}} = \bigg( \sum_{x}\sum_{n = 1}^D \abs{v_x^{2n - 1} + v_x^n}^2 \bigg)^{1/2} \leq \sqrt{2} \norm{v}.
% \]
% A similar analysis yields $\norm{b_v} \leq \sqrt{2}\norm{v}$ --- we thus obtain $\norm{c_v} \leq (\norm{a_v} + \norm{b_v}) / \sqrt{2} \leq 2 \norm{v}$. We thus have that $\forall v, u \in \mathbb{C}^{2n}$, $\abs{v^\dagger \Lambda(t, s) u} \leq 4\norm{v}\norm{u} \implies \norm{\Lambda(t, s)} \leq 4$.

\begin{proof}[Proof (of proposition \ref{prop:t_evol_ff})]
In this proof, we will follow the notation introduced in lemma \ref{lemma:gaussian_open_quantum_dynamics} and use unprimed matrices for the noiseless (target) problem, and primed matrices for the noisy problem. From lemma \ref{lemma:bounds_geom_mat}, it follows that
\[
\norm{X - X'} \leq \norm{\tilde{H} - \tilde{H}'} + \norm{\tilde{Q} - \tilde{Q}'} \leq \delta_X := 2D (2R + 1)^d \delta + 4D n_L (2R + 1)^{2d}(2\delta + \delta^2),
\]
and
\[
\norm{Y - Y'} \leq \norm{\tilde{Q} - \tilde{Q}'} \leq \delta_Y := 4D n_L (2R + 1)^{2d}(2\delta + \delta^2).
\]
In the absence of noise and errors, the correlation matrix $\Gamma(t)$ is governed by the differential equation
\[
\frac{d}{dt}\Gamma(t) = X\Gamma + \Gamma X^\text{T} + Y \implies \Gamma(t) = e^{Xt}\Gamma(0) e^{X^\text{T}t} + \int_0^t e^{X(t - s)}Y e^{X^\text{T}(t - s)}ds.
\]
Note that we have set $Z(t) = 0$, since in the noiseless problem there is no decohering environment. In the presence of noise and errors, we instead obtain that
\begin{align*}
\frac{d}{dt}\Gamma'(t) &= X'\Gamma'(t) + \Gamma'(t) {X'}^\text{T} + Y' + Z'(t),
\end{align*}
from which it follows that
\begin{align}\label{eq:perturb_gamma}
    \Gamma'(t) - \Gamma(t) = \int_0^t e^{X(t - s)}\bigg((X' - X) \Gamma'(t) + \Gamma'(t)(X' - X)^\text{T}  + (Y' - Y) + Z'(s)\bigg)e^{X^\text{T}(t - s)}ds.
\end{align}
We have already established in lemma \ref{lemma:bound_z} that $\norm{Z'(t)} \leq \delta_Z  =2\sqrt{2Dn_L} (2R + 1)^d \delta$. Furthermore, we note that $\norm{e^{Xt}} = \norm{e^{X^\text{T}t}} \leq \norm{e^{(X + X^\text{T})t/2}} \leq 1$ (theorem IX.3.1  of Ref.~\cite{bhatia2013matrix}), where we have used the fact that $X + X^\text{T} \preceq 0$. Consequently, from Eq.~\ref{eq:perturb_gamma}, we have that
\[
\norm{\Gamma'(t) - \Gamma(t)} \leq t(2\delta_X + \delta_Y + \delta_Z) \leq O(\delta t).
\]
From this bound, we can now conclude a bound on the error in a quadratic observables. Suppose $O$ is a quadratic observable specified by the coefficients $o_{x, y}^{\alpha, \beta}$ (see Eq.~\ref{eq:gaussian_few_sites_obs}), and let $\mathcal{O}$ and $\mathcal{O}'$ be respectively the noisy and noiseless expected values of this observables. Then,
\[
\abs{\mathcal{O} - \mathcal{O}'} = \abs{\text{Tr}\big(\tilde{O}(\Gamma(t) - \Gamma'(t))\big)} \leq \norm{\tilde{O}}_{\op, 1} \norm{\Gamma(t) - \Gamma'(t)},
\]
where $\tilde{O}$ is the matrix of coefficients $o_{x, y}^{\alpha, \beta}$ with $(x, \alpha)$ corresponding to the row index and $(y, \beta)$ corresponding to the column index. Since we have already established $\norm{\Gamma(t) - \Gamma'(t)} \leq O(\delta t)$, it only remains to prove that $\norm{\tilde{O}}_{\op, 1}$ is bounded above by a system-size independent constant. Observe that since $O$ only acts on $k$ sites, $\tilde O$ has at most $2kD$ nonzero eigenvalues, thus $\norm{\tilde{O}}_{\op, 1} \leq {2kD} \norm{\tilde{O}}$. Assuming the observable to be normalized such that $\norm{O} \leq 1 \implies \norm{\tilde{O}}\leq 1$, we thus obtain that $\abs{\mathcal{O} - \mathcal{O}'} \leq O(\delta t)$.
\end{proof}}

\subsection{Proof of proposition \ref{prop:gs_ff} (Ground states of local Gaussian fermionic models)}
\label{app:gs_ff}

In this appendix we will prove the stability of the expectation value of a translationally invariant, k-locally generated Gaussian observable on the ground state of a quadratic Hamiltonian. We first provide a lemma that uses the translation invariance of a local observable to provide an error bound.
\begin{lemma}\label{lemma:translation_operator}Consider a quadratic operator $O$ which is translationally invariant and expressible as
\[
O = \frac{1}{n}\sum_{x \in \mathbb{Z}_L^d} \tau_x (O_0),
\]
where $n = L^d$ is the number of sites in $\mathbb{Z}_L^d$, $O_0$ is a quadratic operator with a support on at most $k$ sites and $\tau_x$ is a super-operator that translates an operator by $x$, then for any quadratic operator $A_0$,
\[
\bigabs{\textnormal{Tr}(O^\dagger A_0)} \leq \frac{4D^2 k}{n} \norm{\tilde{O}_0}\norm{\tilde{A}_0}_{\textnormal{op}, 1}.
\]
\end{lemma}
\begin{proof}
We note that
\[
\textnormal{Tr}(\tilde{O}^\dagger \tilde{A}_0) = \frac{1}{n}\sum_{x\in \mathbb{Z}_L^d} \text{Tr}\big(\tilde{O}_0 \tau_x^\dagger(\tilde{A}_0)) = \frac{1}{n}\sum_{x\in \mathbb{Z}^d} \text{Tr}\big(\tilde{O}_0 \tau_{-x} (\tilde{A}_0)\big).
\]
Define $A \equiv n^{-1}\sum_{x \in \mathbb{Z}_L^d} \tau_{-x} (A_0)$. We note that $A$ is translationally invariant on the underlying lattice --- consequently, if $F$ is the $n \times n$ Fourier transform matrix, then $F_D = F \otimes I_{2D}$ block diagonalizes $\tilde{A}$ i.e. $ F_D \tilde{A} F_D^\dagger$ will be a block diagonal matrix with $n$ blocks of size $2D \times 2D$. Then we use
\[
\bigabs{\tr(\tilde{O}^\dagger \tilde{A}_0)} = \bigabs{\tr(F_D^\dagger \tilde{O}^\dagger F^{\phdg}_{D} F_D^\dagger \tilde{A} F^{\phdg}_D)}\leq \bignorm{\text{vec}\big(F_D^\dagger \tilde{O}_0 F^{\phdg}_{D} \big)}_{ \infty} \bignorm{\text{vec}(F_D^\dagger \tilde{A} F^{\phdg}_D)}_{1}
\]
where we have applied Hölder's inequality to the norms of the vectorized matrices. Now we bound each of the factors in the right hand side. Since the operator $O_0$ has support only $k$ sites, $\tilde{O}_0$ only has $2D k \times 2D k$ non-zero elements. Suppose that $\Pi_{O_0}$ is a diagonal matrix with $1$s on the entries that correspond to non-zero elements of $\tilde{O}_0$ --- it then follows that $\tilde{O}_0 = \Pi_{O_0} \tilde{O}_0 \Pi_{O_0}$. We further note that if $f_i$ is the $i^\text{th}$ column of $F_D$ then
\[
\bignorm{\text{vec}\big(F_D^\dagger \tilde{O}_0 F_{D} \big)}_{ \infty} = \sup_{i, j} {\bigabs{f_i^\dagger \Pi_{O_0} \tilde{O}_0 \Pi_{O_0} f_j  }} \leq \norm{\tilde{O}_0} \sup_{i, j} \norm{\Pi_{O_0} f_i} \norm{\Pi_{O_0} f_j} = \frac{2D k}{n}\norm{\tilde{O}_0},
\]
where we have used that each entry of $F_D$ has magnitude $1 / \sqrt{n}$ since it is the Fourier transform matrix. Next, since $F_D^\dagger \tilde{A} F_D$ is block diagonal with $N$ $2D\times 2D$ blocks, and labelling by $A_1, A_2 \dots A_N$ these blocks, we obtain that
\[
\bignorm{\text{vec}(F_D^\dagger \tilde{A} F_D)}_{1} = \sum_{i = 1}^N \bignorm{\text{vec}(A_i)}_1 \leq 2D \sum_{i = 1}^N \bignorm{A_i}_{\text{op}, 1} = 2D \bignorm{F_D^\dagger \tilde{A} F_D}_{\text{op}, 1} = 2D \bignorm{ \tilde{A}}_{\text{op}, 1}.
\]
where we have used $\norm{\text{vec}(M)}_1\leq n\norm{M}_{\op,1}$ for an $n\times n$ matrix\footnote{To see this, let $\sigma_{ij}\equiv\text{sign}(M_{ji})$. Then ${\norm{\sigma}_\op\leq n\norm{\text{vec}(\sigma)}_\infty = n}$, and $\norm{\text{vec}(M)}_1=\tr(\sigma M)\leq\norm{\sigma}_\op\norm{M}_{\op,1}\leq n\norm{M}_{\op,1}$.}. Finally, since $A = \sum_{x \in \mathbb{Z}_L^d} \tau_{-x}(A_0) / N$, it follows that $\norm{\tilde{A}}_{\text{op}, 1 } \leq \norm{\tilde{A}_0}_{\text{op}, 1}$. Combining the above estimates, the lemma statement follows.
\end{proof}

The correlation matrix $\Gamma$ of the ground state of a quadratic Hamiltonian $H$ with matrix of coefficients $\tilde H$ (see Eq.~\eqref{eq:quadH}) is given by
\[
\Gamma = \text{sign}(\tilde{H}),
\]
where $\text{sign}(x) = x / \abs{x}$ for $x \neq 0$ and $0$ for $x = 0$\footnote{The reader may be familiar with the equivalent formulation in terms of complex fermions, where the function to be applied to the Hamiltonian matrix to obtain the correlation matrix of the ground state is of the Heaviside type, such that it populates negative energy states and depopulates positive energy states. The function of the sign function in the language of Majorana fermions is exactly analogous.}. The sign function applied on a matrix is to be understood as an operator function i.e.~as a function acting on the eigenvalues of the argument while keeping the eigenvectors unchanged. Our proof will rely on a Fourier series approximation to the sign function. Within the interval $(-\pi, \pi)$, we will investigate the approximation of $\text{sign}(x)$ with $\text{sign}_M(x)$, where
\[
\text{sign}_M(x) = \sum_{n = -M}^M c_n e^{inx} \ \text{where} \ c_n = \frac{1}{2\pi}\int_{-\pi}^\pi \text{sign}(x)e^{-inx}dx.
\]
To analyze the error between $\text{sign}_M(x)$ and $\text{sign}(x)$, it is convenient to express $\text{sign}_M(x)$ in terms of the Dirichlet kernel,
\[
\text{sign}_M(x) \equiv \int_{-\pi}^\pi D_M(x - y) \text{sign}(y) dy,
\]
where
\[
D_M(x) \equiv \frac{1}{2\pi} \sum_{n = -M}^M e^{-in x} = \frac{1}{2\pi} \frac{\sin[(M + 1/2) x]}{\sin(x/2)}.
\]
Below, we provide two technical lemmas about the $\text{sign}_M$ function --- one that quantifies the approximation error between it and the exact sign function, and the next that quantifies the maximum value of the $\text{sign}_M$ function. Both of these lemmas will be used for the perturbation theory analysis of the free-fermion ground state problem.
\begin{figure}
    \centering
    \includegraphics[width = .8\linewidth]{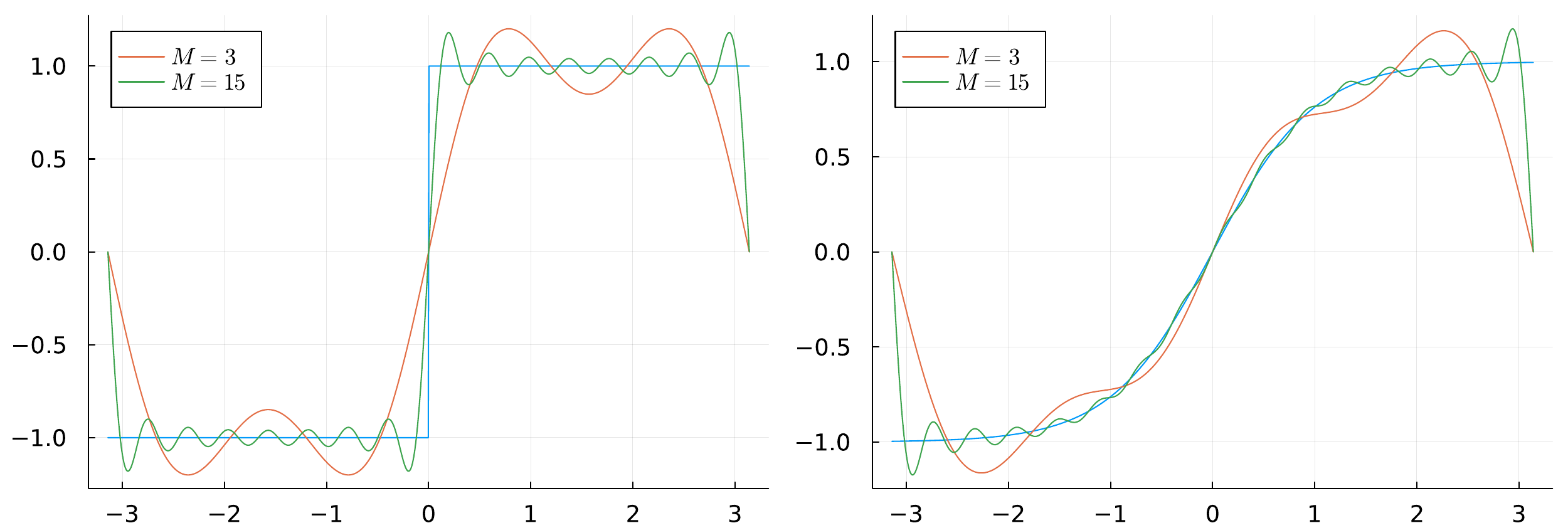}
    \caption{(left) Truncated Fourier series approximation $\sign_M(x)$ to the $\sign(x)$ function, used in the proof of proposition \ref{prop:gs_ff}. (right) Truncated Fourier series approximation $t_M(x)$ to the $\tanh(\beta x)$ function (for $\beta =1$), used in the proof of proposition \ref{prop:Gibbs_ff}.}
    \label{fig:Fourier}
\end{figure}
\begin{lemma}\label{lemma:sign_func_appx}
For all $\eta \leq \abs{x} \leq \pi - \eta$ and $M > 0$,
\[
 \big| \textnormal{sign}(x) - \textnormal{sign}_M(x)  \big| \leq \frac{1}{M} + \frac{1}{M\eta}.
\]
\end{lemma}
\begin{proof}
We first consider $x \in [\eta, \pi]$. We note that
\begin{align*}
    \text{sign}_M(x) = \int_{0}^{\pi} D_M(x - y) dy - \int_0^\pi D_M(x + y) dy.
\end{align*}
Now, since $\int_{-\pi}^\pi D_M(y) dy = 1$, we obtain that
\[
\int_{0}^{\pi} D_M(x - y) dy = 1 - \int_{0}^\pi D_M(x + y) dy,
\]
and thus
\[
\abs{\text{sign}_M(x) - \text{sign}(x)} = 2 \bigabs{ \int_0^\pi D_M(x+ y) dy}.
\]
Next, we apply integration by parts to obtain
\begin{align*}
\int_0^\pi D_M(x + y) dy = \frac{1}{\pi(2M+1)}\bigg(\frac{\cos((M + 1/2)(\pi + x))}{\cos(x/2)} + \frac{\cos((M + 1/2)(\pi + x))}{\sin(x/2)} -\nonumber \\ \frac{1}{2}\int_0^\pi \frac{\cos((M + 1/2)(x + y)) \cos((x + y)/2)}{\sin^2((x + y)/2)} dy\bigg),
\end{align*}
and therefore
\begin{align*}
&\bigabs{ \int_0^\pi D_M(x+ y) dy} \leq \frac{1}{\pi(2M + 1)}\bigg(\frac{1}{\abs{\cos(x/2)}} + \frac{1}{\abs{\sin(x/2)}} + \frac{1}{2} \int_0^\pi \frac{\abs{\cos((x + y) / 2)}dy}{\sin^2((x + y)/2)}\bigg) \nonumber \\
&\qquad \qquad \leq \frac{2}{\pi(2M + 1)}\bigg(\frac{1}{\abs{\cos(x/2)}} + \frac{1}{\abs{\sin(x/2)}} - 1 \bigg),
\end{align*}
where in the last step we have used the integral
\[
\frac{1}{2}\int_x^{\pi + x} \frac{\abs{\cos(y / 2)}}{{\sin^2(y / 2)}} dy  = \frac{1}{\sin (x / 2)} + \frac{1}{\cos(x / 2)} - 2.
\]
Now, for $x \in (\eta, \pi/2)$, $\abs{\cos(x / 2)} \geq 1/\sqrt{2} $ and $\abs{\sin(x/2)} \geq x / \pi \geq \eta / \pi $. Therefore, we obtain that
\[
\bigabs{ \int_0^\pi D_M(x+ y) dy} \leq \frac{2}{\pi(2M + 1)}\bigg(\sqrt{2} + \frac{\pi}{\eta} - 1\bigg).
\]
While this bound is true for $x \in [\eta, \pi / 2]$, we note that both $\text{sign}, \text{sign}_M$ satisfy $f(x) = f(\pi - x)$ for $x \in [0, \pi]$ and consequently this bound also holds for $x \in [\pi /2, \pi - \eta]$. Finally, since for both $\text{sign}, \text{sign}_M$, $f(x) = -f(-x)$, it follows that this bound holds for $[-\pi + \eta, -\eta] \cup [\eta, \pi - \eta]$.
\end{proof}

\begin{lemma}\label{lemma:upper_bound_sign}
For all $x \in [-\pi, \pi]$, $\abs{\textnormal{sign}_M(x)} \leq 5$.
\end{lemma}
\begin{proof}This proof is an adaptation of the standard technique based on Riemann integration that is used to treat Gibbs phenomena in Fourier analysis. We repurpose that technique to provide error bounds as a function of $M$ instead of just concentrating on the asymptotic limit $M \to \infty$. Again, we only consider $x \in [0, \pi/2]$, and extend the bound on $\abs{\text{sign}_M(x)}$ to the remaining interval by symmetry. We divide the interval $[0, \pi / 2]$ into $[0, \alpha_0 / M] \cup [\alpha_0 / M, \pi / 2]$, where $\alpha_0$ is a constant that we pick later.

Consider first $x \in [\alpha_0 / M, \pi / 2]$. An application of lemma \ref{lemma:sign_func_appx} yields
\[
\abs{\text{sign}_M(x)} \leq 1 + \frac{2}{\pi(2M + 1)} \bigg(\sqrt{2} - 1 + \frac{\pi M}{\alpha_0}\bigg).
\]
For large $M$, this bound scales as $\sim 1/ \alpha_0$ and thus does not allow us to provide an upper bound on $\text{sign}_M(x)$ for $x$ close to $0$. For this, we use the representation of $\text{sign}_M(x)$ as a Fourier series which approximates a Riemann integral of $\sin(t) / t$. Consider $x \in [0, \alpha_0 / M)$ and let $\alpha = x M$ ($\alpha \leq \alpha_0$). Note that
\[
\text{sign}_M(x) =\frac{2}{\pi} \sum_{k \in [1, M] | k \text{ is  odd}} \frac{2}{k} \sin\bigg(\frac{k \alpha}{M}\bigg)
\]
To bound the term in the summation, we observe that it is an approximation of the Riemann integral of $\sin(\alpha x) / x$ in the interval $[0, 1]$. In particular, since $\sup_{x \in \mathbb{R}} \abs{(\sin x / x)'} \leq 2$, Taylor's theorem yields that
\[
\bigabs{\sum_{k \in [1, M] | k \text{ is odd}} \frac{2}{k} \sin \bigg(\frac{k \alpha}{M}\bigg) - \int_0^1 \frac{\sin \alpha x}{x}dx} \leq \frac{4 \alpha^2}{M} \leq \frac{4\alpha_0^2}{M}.
\]
Finally, we note that
\[
\int_0^1 \frac{\sin \alpha x}{x} dx \leq \alpha \leq \alpha_0.
\]
Thus, we obtain that for $x \in [0, \alpha_0 / M)$,
\[
\abs{\text{sign}_M(x)} \leq \alpha_0 + \frac{4\alpha_0^2}{M}.
\]
Thus, for the entire interval $[0, \pi / 2]$, we obtain that
\[
\abs{\text{sign}_M(x)} \leq \text{max}\bigg(\alpha_0 + \frac{4\alpha_0^2}{M}, 1 + \frac{2}{\pi(2M + 1)}\bigg(\sqrt{2} - 1 + \frac{\pi M}{\alpha_0}\bigg)\bigg).
\]
Since this holds for any $\alpha_0$, we choose $\alpha_0 = 1$. We then obtain that
\[
\abs{\text{sign}_M(x)} \leq \text{max}\bigg(1 + \frac{4}{M}, 1 + \frac{2}{\pi(2M + 1)}\big(\sqrt{2} - 1 + {\pi M}\big)\bigg) \leq 5 \text{ for }M \geq 1.\]
\end{proof}
%-----------------------------------------------
\begin{proof}[Proof (of proposition \ref{prop:gs_ff})]The expectation value of the observable $O$ in the ground state of the Hamiltonian $H$ is given by
\[
\langle O \rangle_{H} = \text{Tr}\big( \tilde{O}\ \text{sign}(\tilde{H})\big), \ \langle O \rangle_{H'} = \text{Tr}\big( \tilde{O} \ \text{sign}(\tilde{H'})\big).
\]
Without loss of generality, we will assume that $\tilde H, \tilde H'$ are normalized so that $\norm{\tilde H}, \norm{\tilde H'} \leq \frac{\pi}{2}$. This way all the eigenfrequencies lie in the interval $[-\frac{\pi}{2},\frac{\pi}{2}]$. Lemma \ref{lemma:bound_op_norm} guarantees that this can be done with a constant normalization factor, i.e. one that does not depend on the system size, and does not change the ground state (note however that $\delta$ and $f_h$ would have to be rescaled accordingly). Now, from lemma \ref{lemma:translation_operator}, it follows that
\[
\bigabs{\langle O \rangle_{H} -  \langle O \rangle_{H'} } \leq \frac{4D^2 k}{n} \norm{\tilde{O}_0} \norm{\text{sign}(\tilde{H}) - \text{sign}(\tilde{H}')}_{\text{op}, 1}.
\]
Furthermore,
\begin{align*}
&\norm{\text{sign}(\tilde{H}) - \text{sign}(\tilde{H}')}_{\text{op}, 1} \leq \nonumber\\
&\quad\quad\norm{\text{sign}(\tilde{H}) - \text{sign}_M(\tilde{H})}_{\text{op}, 1} + \norm{\text{sign}(\tilde{H}') - \text{sign}_M(\tilde{H}')}_{\text{op}, 1}  + \norm{\text{sign}_M(\tilde{H}) - \text{sign}_M(\tilde{H'})}_{\text{op}, 1}.
\end{align*}
We bound each term on the right hand side separately. Consider $\norm{\text{sign}(\tilde{H}) - \text{sign}_M(\tilde{H})}_{\text{op}, 1}$ --- denoting by $\lambda_i$ the eigenvalues of $\tilde{H}$ and for any $\eta > 0$, we can express it as
\[
\norm{\text{sign}(\tilde{H}) - \text{sign}_M(\tilde{H})}_{\text{op}, 1} = \sum_{i | \lambda_i \in[-\eta, \eta]} \abs{\text{sign}(\lambda_i) - \text{sign}_M(\lambda_i)} + \sum_{i | \lambda_i \notin [-\eta, \eta]} \abs{\text{sign}(\lambda_i) - \text{sign}_M(\lambda_i)}.
\]
The motivation behind splitting the error into these two terms is that, within the interval $[-\frac{\pi}{2},\frac{\pi}{2}]$, the approximation of $\sign{(\lambda)}$ by $\sign_M{(\lambda)}$ is only good outside the neighbourhood of $0$ (see Fig.~\ref{fig:Fourier}) --- consequently, we treat the eigenvalues of $\tilde{H}$ which lie within $\eta$ radius of $0$ separately from the rest. It now follows that from assumption \ref{assum:eigenfreqs} and lemma \ref{lemma:upper_bound_sign} that
\[
\sum_{i | \lambda_i \in [-\eta, \eta]} \abs{\text{sign}(\lambda_i) - \text{sign}_M(\lambda_i)} \leq 6  n f_h(\eta) + 6 \kappa(\eta, n).
\]
Furthermore, from lemma \ref{lemma:sign_func_appx},
\[
\sum_{i | \lambda_i \notin [-\eta, \eta]} \abs{\text{sign}(\lambda_i) - \text{sign}_M(\lambda_i)} \leq \frac{n}{M} \bigg(1 + \frac{1}{\eta}\bigg).
\]
Therefore, we obtain that
\[
\frac{1}{n}\norm{\text{sign}(\tilde{H}) - \text{sign}_M(\tilde{H})}_{\text{op}, 1} \leq 6f_h(\eta) + 6\frac{\kappa(\eta, n)}{n} + \frac{1}{M}\bigg(1 + \frac{1}{\eta}\bigg).
\]
We can similarly analyze $\norm{\text{sign}(\tilde{H}') - \text{sign}(\tilde{H}')}_{\text{op}, 1}$. Denote by $\lambda'_i$ the eigenvalues of $\tilde{H}'$ --- it follows from Weyl's theorem and lemma \ref{lemma:bounds_geom_mat} that $\abs{\lambda_i - \lambda_i'} \leq \norm{\tilde{H}-\tilde{H'}}_{\text{op}} \leq c_0 \delta $ where $c_0 = 2D (2R + 1)^d$. Consequently, for sufficiently small, but $\Theta(1)$, $\delta$, we obtain that
\[
\sum_{i | \lambda_i' \in [-\eta, \eta]} \abs{\text{sign}(\lambda_i') -\text{sign}_M(\lambda_i')} \leq 6n f_h(\eta + c_0 \delta) + 6 \kappa(\eta + c_0 \delta, n),
\]
and
\[
\sum_{i | \lambda_i' \notin [-\eta, \eta]}  \abs{\text{sign}(\lambda_i') -\text{sign}_M(\lambda_i')}  \leq \frac{n}{M}\bigg(1 + \frac{1}{\eta}\bigg).
\]
Therefore,
\[
\frac{1}{n}\norm{\text{sign}(\tilde{H}') - \text{sign}_M(\tilde{H}')}_{\text{op}, 1} \leq 6f_h(\eta + c_0 \delta) + 6\frac{\kappa(\eta + c_0 \delta, n)}{n} + \frac{1}{M}\bigg(1 + \frac{1}{\eta}\bigg).
\]
Finally, we consider $\norm{\text{sign}_M(H) - \text{sign}_M(H')}_{\text{op}, 1} \leq n \norm{\text{sign}_M(H) - \text{sign}_M(H')}$. Now, denoting by $\{c_m\}_{m \in \mathbb{Z}}$ the Fourier series components of $\text{sign}$ function, then
\[
 \norm{\text{sign}_M(H) - \text{sign}_M(H')} \leq \sum_{m = -M}^M \abs{c_m} \norm{e^{im \tilde{H}} - e^{im  \tilde{H}'}} \leq \sum_{m = -M}^M \abs{ m c_m} \norm{\tilde{H}-\tilde H'}.
\]
Using the explicit expression for $c_m$, we can immediately conclude that $\abs{m c_m} = 2 / \pi$ when $m$ is odd, and $0$ when $m$ is even. Therefore, we obtain that
\[
 \norm{\text{sign}_M(H) - \text{sign}_M(H')}  \leq \frac{2(M + 1)}{\pi} \norm{\tilde{H}-\tilde H'} \leq \frac{2(M + 1)}{\pi}c_0 \delta.
 \]
Combining all of these estimates, we obtain that
 \[
 \frac{1}{n} \norm{\text{sign}(\tilde{H}) - \text{sign}(\tilde{H}')}_{\text{op}, 1} \leq \frac{2(M + 1)}{\pi} c_0\delta  + \frac{2}{M}\bigg(1 + \frac{1}{\eta}\bigg) + 6\big(f_h(\eta) + f_h(\eta + c_0 \delta)\big) + 6\bigg(\frac{\kappa(\eta, n)}{n} + \frac{\kappa(\eta + c_0 \delta, n)}{n}\bigg).
 \]
 with $c,c'$ constants. Since this is valid for any $\eta$ and $M$, choosing $M = \delta^{-1/2}$ and $\eta = \delta^{1/4}$, we obtain the proposition.
 \end{proof}
 %%%%%%%%%%%%%%%%%%%%%%%%%%%%%%%%%%%%%%%%%%%%%%%%%%%%%%%%%%%%%%%%%%%%%%%%%

\subsection{Proof of proposition \ref{prop:Gibbs_ff} (Gibbs state of free fermion models)}
\label{app:Gibbs_ff}
The correlation matrix of a thermal state of a quadratic Hamiltonian can be written in terms of the coefficient matrix $H$ of the latter as $\Gamma = \tanh(\beta H)$. Note that the $\beta\to\infty$ limit yields the sign function, which was used in the previous appendix to compute the ground state correlation matrix. Indeed, the reasoning here will be similar to that of appendix \ref{app:gs_ff}, replacing the sign function with the hyperbolic tangent. The next couple of lemmas discuss the Fourier series approximation of $\tanh{\beta x}$, defined as
 \[
 {t}_M(x) \equiv \sum_{n = -M}^M c_n e^{inx}, \ \text{where } c_n = \frac{1}{2\pi} \int_{-\pi}^\pi \tanh{\beta x} e^{-inx} dx.
 \]
\begin{lemma}\label{lemma:tanh_appx}
For $M \geq 1$, and $x \in \left[-\frac{\pi}{2}, \frac{\pi}{2}\right]$,
\[
\abs{t_M(x) - \tanh{\beta x}} \leq \frac{q(\beta)}{M},
\]
where $q(\beta) \equiv 12\pi^2\beta^3+2\pi^2\beta^2+\left(2+\frac{\pi^2}{2}\right)\beta+\left(\frac{4\sqrt{2}}{\pi}+\frac{\pi^2}{2}\right)=O(\beta^3)$.
\end{lemma}
\begin{proof}
We fix the value of $\beta$ and let $t(x)$ be the $2\pi$-periodic extension of $\tanh{\beta x}$
\begin{equation}
    t(x)\equiv\tanh{\beta (x-2n\pi)},\qquad x-2n\pi\in[-\pi,\pi],\quad n\in\mathbb{Z}
\end{equation}
Once again, it will be convenient to represent $t_M(x)$ in terms of the Dirichlet kernel $D_M$. We note,
\[
t_M(x) = \int_{-\pi}^\pi{D_M(x-y)t(y)dy}=\int_{-\pi}^\pi{D_M(y)t(x-y)dy}=\int_{-\pi}^\pi{D_M(y)t(x+y)dy},
\]
% \[
% t_M(x) = \frac{1}{2\pi} \int_{-\pi}^\pi \big(\tanh{\beta x} - \text{tanh}\beta(x - y)\big) \frac{\sin(M + 1/2) y}{\sin y/2} dy.
% \]
and therefore, using that the Dirichlet kernel is normalized, we write
\[
t(x)-t_M(x) = \dfrac{1}{2}\int_{-\pi}^\pi{D_M(y)\left(2t(x)-t(x-y)-t(x+y)\right)dy} = \int_{0}^\pi{D_M(y)f_x(y)dy},
\label{eq:tanherror}
\]
where in the last step we have defined $f_x(y)\equiv 2t(x)-t(x-y)-t(x+y)$. In the integration interval $[0,\pi]$, $f_x(y)$ is piecewise smooth with a single jump discontinuity at $y = \pi - x$. We thus split the integral into the two intervals $[0,\pi-x]$ and $[\pi-x,\pi]$ and apply integration by parts in each of them. For the first one,
\begin{align*}
    \int_{0}^{\pi-x}{D_M(y)f_x(y)dy} = -\dfrac{1}{\pi}\dfrac{\cos{\left(\left(M+\frac{1}{2}\right)y\right)}}{2M+1}\dfrac{f_x(y)}{\sin\frac{y}{2}}\Bigg|_{y=0}^{\pi-x} + \dfrac{1}{(2M+1)\pi}\int_{0}^{\pi-x}{g_x(y)\dfrac{\cos{\left(\left(M+\frac{1}{2}\right)y\right)}}{\sin^2\frac{y}{2}}dy}
\end{align*}
where $g_x(y)\equiv2\sin\frac{y}{2}f_x'(y)-\cos{\frac{y}{2}}f_x(y)$. To bound this expression, we will use the following properties of the functions $f_x(y), g_x(y)$ on the interval $[0,\pi-x]$, where they are smooth:
\begin{align*}
    f_x(0) =  f'_x(0) = 0, |f_x(y)|\leq 4, |f'_x(y)|\leq 2\beta, |f''_x(y)|\leq 4\beta^2, |f'''_x(y)|\leq 12\beta^3,\\
    g_x(0) = g'_x(0) = 0, |g_x''(y)|\leq 24\beta^3+4\beta^2+\beta+1.
\end{align*}
These bounds follow from direct computation, and in the case of $g_x(y)$ they are easiest to see when expressed in terms of $f_x(y)$. They imply (via Taylor's theorem with second order remainder) that
\[|g_x(y)|\leq (24\beta^3+4\beta^2+\beta+1)\dfrac{y^2}{2}
\]
which together with $\sin^2(y)\geq\frac{y^2}{\pi^2}$ will allow us to bound the integral. Putting it all together, we have
\[\bigabs{\int_{0}^{\pi-x}{D_M(y)f_x(y)dy}}\leq \dfrac{4\sqrt{2}}{(2M+1)\pi}+\dfrac{\pi^2}{(2M+1)}(24\beta^3+4\beta^2+\beta+1)
\]
Now we proceed on to the second interval $y\in[\pi-x,\pi]$ and similarly integrate by parts,
\begin{align*}
    \int_{\pi-x}^{\pi}{D_M(y)f_x(y)dy} = -\dfrac{1}{\pi}\dfrac{\cos{\left(\left(M+\frac{1}{2}\right)y\right)}}{2M+1}\dfrac{f_x(y)}{\sin\frac{y}{2}}\Bigg|_{y=\pi-x}^{\pi} + \dfrac{1}{(2M+1)\pi}\int_{\pi-x}^{\pi}{g_x(y)\dfrac{\cos{\left(\left(M+\frac{1}{2}\right)y\right)}}{\sin^2\frac{y}{2}}dy}.
\end{align*}
Now the bound on $g_x(y)$ from Taylor's theorem no longer holds, due to the discontinuity, but since $y=0$ is not in the integration interval, we can just use the constant bound $|g(x)|\leq 4\beta+4$ to obtain
\[\bigabs{\int_{\pi-x}^{\pi}{D_M(y)f_x(y)dy}}\leq \dfrac{4\sqrt{2}}{(2M+1)\pi}+\dfrac{4}{(2M+1)}(\beta+1),
\]
and putting everything together the lemma follows.
\end{proof}
\begin{lemma}\label{lemma:sum_coeffs}
If $\{c_n\}_{n \in \mathbb{Z}}$ are the Fourier series coefficients of $\tanh{\beta x}$ in the interval $[-\pi, \pi]$, then for $M \geq 1$
\[
\sum_{n = -M}^M \abs{n c_n} \leq 2M (\beta + 1).
\]
\end{lemma}
\begin{proof}
This follows by a straightforward manipulation of $c_n$ --- note that $c_0 = 0$, and for $n\neq 0$, we obtain from integration by parts that
\[
c_n = \frac{1}{2\pi} \int_{-\pi}^\pi \tanh{\beta x} \ e^{-inx} dx = \frac{1}{2\pi} \bigg(\frac{2i}{n}\tanh \beta \pi \ e^{-in\pi} + \frac{\beta}{in}\int_{-\pi}^\pi \frac{e^{-inx}}{\text{cosh}^2 \beta x} dx\bigg).
\]
Consequently,
\[
\abs{c_n} \leq \frac{1}{2\pi} \bigg(\frac{2}{n} +\frac{2\pi \beta}{n}\bigg) \leq \frac{\beta + 1}{n}.
\]
From this bound, the lemma follows.
\end{proof}
\begin{proof}[Proof (of proposition \ref{prop:Gibbs_ff})]
We bound the error between $\langle O \rangle_{H, \beta}$ and $\langle O \rangle_{H', \beta}$ using the same procedure as for the ground state (see appendix \ref{app:gs_ff}) --- the proof simplifies significantly because $\tanh{\beta x}$ does not have a discontinuity near $x = 0$ (unlike the sign function). From lemma \ref{lemma:translation_operator} it follows that
\[
\abs{\langle O \rangle_{H, \beta} - \langle O \rangle_{H', \beta}} \leq \frac{4D^2 k}{n}\norm{\tilde{O}_0} \norm{\tanh{\beta \tilde{H}} - \tanh{\beta \tilde{H}'}}_{\text{op}, 1}.
\]
We again split
\begin{align*}
 &\norm{\tanh{\beta \tilde{H}} - \tanh{\beta \tilde{H}'}}_{\text{op}, 1}  \leq \nonumber \\
 &\qquad \norm{\tanh{\beta \tilde{H}} - t_M(\tilde{H})}_{\text{op} , 1} + \norm{\tanh{\beta \tilde{H}'} - t_M(\tilde{H}')}_{\text{op} , 1}  + \norm{t_M(\tilde{H}') - t_M(\tilde{H})}_{\text{op} , 1}
\end{align*}
We will assume once again that $\norm{H}, \norm{H'} \leq \frac{\pi}{2}$, so that from lemma \ref{lemma:tanh_appx}, it follows that
\[
\norm{\tanh{\beta \tilde{H}} - t_M(\tilde{H})}_{\text{op} , 1}, \norm{\tanh{\beta \tilde{H}'} - t_M(\tilde{H}')}_{\text{op} , 1} \leq \frac{n q(\beta)}{M},
\]
and
\[
\norm{t_M(\tilde{H}) - t_M(\tilde{H}')}_{\text{op}, 1} \leq  n\norm{t_M(\tilde{H}) - t_M(\tilde{H}')}_{\text{op}} \leq n\sum_{m = -M}^M \abs{c_m} \norm{e^{im\tilde{H}} - e^{im\tilde{H}'}}_{\text{op}}.
\]
Furthermore, from lemmas \ref{lemma:bound_op_norm}, \ref{lemma:pert_theory} and \ref{lemma:bounds_geom_mat} we have $\norm{e^{im\tilde{H}} -e^{im\tilde{H}'}}_\text{op} \leq m c_0 \delta$, where $c = 2D (2R + 1)^d$. Thus, from lemma \ref{lemma:sum_coeffs}, it follows that
\[
\norm{t_M(\tilde{H}) - t_M(\tilde{H}')}_{\text{op}, 1}  \leq 2nM (\beta + 1) c_0\delta.
\]
Thus, we obtain that for any $M > 1$,
\[
\abs{\langle O \rangle_{H, \beta} - \langle O \rangle_{H', \beta}} \leq 4D^2k \norm{O_0}_\text{op} \bigg(\frac{2q(\beta)}{M} + 2(\beta + 1) c_0 M \delta\bigg).
\]
choosing $M = \sqrt{q(\beta) / c_0 (\beta + 1) \delta}$, we obtain the result.
\end{proof}

\subsection{Stability of fixed points}\label{app:fp_ff}
We will now consider translationally invariant local observables in the fixed point. Recall from section \ref{app:t_evol_ff}, lemma \ref{lemma:gaussian_open_quantum_dynamics} that the dynamics of the correlation matrix $\Gamma(t)$ under evolution by a Gaussian master equation is governed by
\[
\frac{d}{dt}\Gamma(t) = X\Gamma(t) + \Gamma(t) X^\text{T} + Y,
\]
where $X, Y$ are defined in lemma \ref{lemma:gaussian_open_quantum_dynamics}. Assuming $X$ to be invertible, this differential equation has a unique fixed point which can be expressed as
\[
\Gamma_\infty = \int_0^\infty e^{X t} Y e^{X^\text{T}t}dt.
\]
We also note from lemma \ref{lemma:gaussian_open_quantum_dynamics} is $X + X^\text{T}$ is a negative definite matrix if $X$ is invertible, and its eigenvalues can be interpreted as a measure of the decay rates of the eigenmodes of the open system. We now restate assumption \ref{assump:lindbladian_holder_cont} in terms of the eigenvalues of $X + X^\text{T}$.
\begin{repassumption}{assump:lindbladian_holder_cont}[Formal]
 The matrix $X + X^\text{T}$ has no zero eigenvalues. Furthermore, if its eigenvalues are $-\lambda_i$, for $i \in \{1, 2 \dots 2n\}$ where $0 < \lambda_1\leq \lambda_2 \dots \lambda_N$, then the number of eigenvalues in the interval $(0, \eta]$, $n(\eta)$ satisfies
 \[
 n(\eta) \leq n f_\ell(\eta) + \kappa(\eta, n),
 \]
 where $f_\ell$ is a function such that $f_\ell(\eta) \leq O(\eta^\alpha)$ for some $\alpha > 0$ and $\kappa(\eta, n) = o(n)$ for any fixed $\eta$.
\end{repassumption}
\noindent This assumption is expected to be satisfied for translationally invariant systems, as well as rapidly mixing systems where typically there would be a gap in the Lindbladian spectrum. However, it additionally includes systems where the minimum eigenvalue could have a real part that scales as $1/n$, and thus would generically need $\Theta(n)$ time to reach their fixed points.

Similar to the stability result for ground states, the observables we will consider will be translationally invariant local observables - let $O_0$ be a $k$-local observable , then we consider observables $O$ which can be expressed as $O =  \sum_{x\in \mathbb{Z}_L^d} \tau_x(O_0) / n$. We will denote by $\mathcal{O}$ and $\mathcal{O}'$ the expectation value of the observable $O$ in the unperturbed and perturbed fixed point i.e.~$\mathcal{O} = \textnormal{Tr}(O \Gamma_\infty)$ and $\mathcal{O}' = \textnormal{Tr}(O\Gamma_\infty')$. \\

\noindent \begin{proof}[Proof (of proposition \ref{prop:fixed_points_gaussian})] We start by using lemma~\ref{lemma:translation_operator} to obtain
\begin{align}\label{eq:fp_obeservable_error_corr_mat}
\abs{\mathcal{O} - \mathcal{O}'} \leq \frac{4D^2 k}{n} \norm{\tilde{O}_0}_\textnormal{op} \norm{\Gamma_\infty - \Gamma_\infty'}_{\text{op}, 1}.
\end{align}
Furthermore,
\[
\Gamma_\infty = \int_0^\infty e^{X t}Ye^{X^\text{T}t} dt \text{ and }\Gamma_\infty' = \int_0^\infty e^{X't}Y' e^{{X'}^\text{T}t} dt.
\]
For any $t_0 > 0$, it follows that
\begin{align}\label{eq:fp_split_corr_error}
\norm{\Gamma_\infty - \Gamma_\infty'}_{\text{op}, 1} &\leq \frac{1}{n}\bignorm{\int_0^{t_0} \bigg(e^{X t}Ye^{X^\text{T}t} - e^{X't}Y' e^{{X'}^\text{T}t}\bigg) dt}_{\textnormal{op}, 1} + \frac{1}{n} \bignorm{\int_{t_0}^\infty e^{X t}Ye^{X^\text{T}t} dt}_{\text{op}, 1} + \frac{1}{n} \bignorm{\int_{t_0}^\infty e^{X't}Y' e^{{X'}^\text{T}t} dt}_{\text{op}, 1}.
\end{align}
Furthermore, we note that
\[
\int_{t_0}^\infty e^{X t}(Y) e^{X^\text{T}t} dt = e^{X t_0}\Gamma_\infty e^{X^\text{T}t_0}.
\]
Let us now estimate $\norm{e^{Xt_0}\Gamma_\infty e^{X^\text{T}t_0}}_{\text{op}, 1}$. Note that since $\Gamma_\infty$ is a covariance matrix, it is positive semi-definite and satifies $\norm{\Gamma_\infty}_{\text{op}} \leq 1$. Let $\Gamma_\infty = \sum_{\alpha}\sigma_\alpha \ket{v_\alpha}\bra{v_\alpha}$ where $0\leq \sigma_\alpha \leq 1$, then
\[
e^{Xt_0}\Gamma_\infty e^{X^\text{T}t_0} = \sum_{\alpha} \sigma_\alpha e^{Xt_0} \ket{v_\alpha} \bra{v_\alpha}e^{X^\text{T} t_0},
\]
from which it follows that
\begin{align*}
\norm{\Gamma_\infty}_{\text{op}, 1} &\leq \sum_{\alpha} \sigma_\alpha\norm{e^{Xt_0} \ket{v_\alpha}}^2, \nonumber\\
&\leq \sum_{\alpha} \bra{v_\alpha}e^{X^\text{T}t_0} e^{Xt_0}\ket{v_\alpha}, \nonumber\\
&\leq \text{Tr}(e^{X^\text{T}t_0}e^{Xt_0}) = \norm{e^{Xt_0}}_{\text{op}, 2}^2, \nonumber \\
&\leq \norm{e^{(X + X^T)t_0/2}}_{\text{op}, 2}^2.
\end{align*}
Note that the first step follows from the variational definition of the Schatten-1 norm i.e.~for any rank 1 decomposition of a matrix $A = \sum_{\alpha}c_\alpha \ket{v_\alpha}\bra{u_\alpha}$ where $\norm{v_\alpha}, \norm{u_\alpha} = 1$, then $\norm{A}_{\text{op}, 1} \leq \sum_{\alpha} \abs{c_\alpha}$. The last step follows from theorem IX.3.1  of Ref.~\cite{bhatia2013matrix}. Next, $\norm{e^{(X + X^\text{T})t_0/2}}_{\text{op}, 2}$ can be written explicitly in terms of the eigenvalues of $X + X^\text{T}$ 
\begin{align*}
\norm{e^{(X + X^T)t_0/2}}_{\text{op}, 2}^2 &= \sum_{i = 1}^{2n} e^{-\lambda_i t_0} =  \sum_{i | \lambda_i \in (0, \eta]}e^{-\lambda_i t_0} + \sum_{i | \lambda_i \in (\eta, \infty)} e^{-2\lambda_i t_0} \leq n \varphi(\eta) + 2n e^{-\eta t_0} + o(n)
\end{align*}
Choosing $\eta = t_0^{-1 + \beta}$ for any $\beta \in (0, 1)$, we obtain that $\norm{e^{(X + X^\text{T})t_0/2}}_\textnormal{op, 2} \leq n(f_\ell(t_0^{-1 + \beta}) + 2 e^{-t_0^\beta}) + o(n)$, which yields that
\begin{align}\label{eq:fp_long_time_bound_1}
\frac{1}{n}\bignorm{\int_{t_0}^\infty e^{Xt} Y e^{X^\text{T}t}dt}_{\textnormal{op}, 1} \leq f_\ell(t_0^{-1 + \beta}) + 2e^{-t_0^\beta} + o(1).
\end{align}
Following a similar procedure, we obtain that
\begin{align}\label{eq:fp_long_time_bound_2}
\frac{1}{n}\bignorm{\int_{t_0}^\infty e^{X't} Y' e^{{X'}^\text{T}t} dt}_{\textnormal{op}, 1} \leq f_\ell(t_0^{-1 + \beta} + c_0 \delta) + e^{-t_0^\beta} + o(1),
\end{align}
where $c_0 = 4D(2R + 1)^d + 8D n_L (2R + 1)^{2d}(2 + \delta) \leq O(1)$ is a constant that is independent of $n$, but dependent on $R, D, d, n_L$. In arriving at this result, we just need to account for the fact that the eigenvalue $\lambda_i'$ of $X' + X'^{\textnormal{T}}$ could differ by at-most $2\norm{X - X'}_\textnormal{op}$ (which, from lemma \ref{lemma:bounds_geom_mat}, is $ \leq c_0 \delta$) from the corresponding eigenvalue $\lambda_i$ of $X + X^\textnormal{T}$, and consequently the number of eigenvalues of $X' + X'^{\text{T}}$ in the interval $(0, \eta]$ is upper bounded by the number of eigenvalues of $X + X^\text{T}$ in the interval $(0, \eta + c_0 \delta]$. In particular, this implies that $\sum_{i | \lambda_i' \in (0, \eta]} e^{-\lambda_i't_0} \leq \sum_{i | \lambda_i \in (0, \eta + c_0\delta]} 1 \leq nf_\ell(\eta + c_0\delta) + o(n)$.

Let us now estimate the remaining term in $\norm{\Gamma_\infty - \Gamma_\infty'}_{\text{op}, 1}$ --- we bound the Schatten 1 norm by the operator (or Schatten $\infty$ norm) in the trivial way
\begin{align}\label{eq:short_time_bound_start}
\frac{1}{n}\bignorm{\int_0^{t_0} \bigg(e^{Xt} Y e^{X^\text{T}t} - e^{X't} Y' e^{{X'}^\text{T}t}\bigg) dt}_{\text{op}, 1} \leq 2D \int_0^{t_0}\bignorm{e^{Xt} Y e^{X^\text{T}t} - e^{X't} Y' e^{{X'}^\text{T}t}}_{\text{op}}dt.
\end{align}
Now, we can bound the error $\norm{e^{Xt} Y e^{X^\text{T}t} - e^{X't} Y' e^{{X'}^\text{T}t}}_\op$ using standard perturbation theory. We begin by noting from lemma \ref{lemma:bounds_geom_mat} that $\norm{X - X'}_\text{op}, \norm{Y - Y'}_\text{op} \leq O(\delta)$ and $\norm{Y}_\text{op}, \norm{Y'}_\text{op} \leq O(1)$ . Furthermore, from theorem IX.3.1 of Ref.~\cite{bhatia2013matrix} and the fact that $X + X^\text{T}$, $X' + {X'}^\text{T}$ are negative-definite, $\norm{e^{Xt}}_\text{op} \leq  \norm{e^{(X + X^\text{T})t}}_\op \leq 1$ and $\norm{e^{X't}}_\text{op} \leq  \norm{e^{(X' + {X'}^\text{T})t}}_\op \leq 1$. Now, we note that
\begin{align*}
\norm{e^{X t} Y e^{X^\text{T}t} - e^{{X'}t} Y'e^{{X'}^\text{T} t}}_\textnormal{op} \leq \norm{e^{X t} (Y - Y') e^{X^\text{T}t}}_\text{op} + \norm{e^{X t} Y' e^{X^\text{T}t} - e^{X't}Y'e^{{X'}^\text{T}t}}_\text{op}.
\end{align*}
We can bound both of these terms separately. For the first term, we obtain that
\begin{align*}
&\norm{e^{Xt}(Y - Y') e^{X^\text{T}t}}_\text{op}\leq \norm{e^{Xt}}_\op\norm{Y - Y'}_\text{op} \norm{e^{X^\text{T}t}}_\op \leq \norm{Y - Y'}_\op \leq O(\delta).
\end{align*}
For the second term, we obtain that
\begin{align*}
&\norm{e^{X t} Y' e^{X^\text{T}t} - e^{X't}Y'e^{{X'}^\text{T}t}}_\text{op} \nonumber\\
&\leq \int_0^t\bignorm{ e^{X(t - s)} (X - X') e^{{X'}s} Y' e^{{X'}^\text{T}s}e^{X^\text{T}(t - s)}}_\text{op} ds + \int_0^t\bignorm{ e^{X(t - s)} e^{{X'}s} Y' e^{{X'}^\text{T}s} (X - X')^\text{T} e^{X^\text{T}(t - s)}}_\text{op}ds, \nonumber \\
&\leq 2\int_0^t \norm{e^{X(t - s)}}_\op^2 \norm{e^{X' s}}^2_\op \norm{X - X'}_\op \norm{Y}_\op ds \leq O(\delta t).
\end{align*}
Thus, from Eq.~\ref{eq:short_time_bound_start}, we obtain that 
\begin{align}\label{eq:fp_short_time_bound}
\frac{1}{n}\bignorm{\int_0^{t_0} \bigg(e^{X t} Y e^{X^\text{T}t} - e^{{X'}t}Y'e^{{X'}^\text{T} t} \bigg) dt}_{\text{op}, 1} \leq O(t_0^2 \delta)
\end{align}
Using the estimates in Eq.~\ref{eq:fp_obeservable_error_corr_mat}, \ref{eq:fp_split_corr_error}, \ref{eq:fp_long_time_bound_1}, \ref{eq:fp_long_time_bound_2} and \ref{eq:fp_short_time_bound}, we have that for any $\beta \in (0, 1)$.
\[
\abs{\mathcal{O} - \mathcal{O}'} \leq O(\delta t_0^2) + f_\ell(t_0^{-1 + \beta}) + f_\ell(t_0^{-1 + \beta} + c_0\delta) + O(e^{-t_0^\beta}) + o(1).
\]
Clearly, any choice of $t_0 = \delta^{-\alpha}$, where $\alpha < 1/2$ yields an upper bound on $\abs{\mathcal{O} - \mathcal{O}'}$ that is uniform in $n$ and goes to 0 as $\delta\to 0$. A concrete choice could be $t_0 = \delta^{-1/4}$, and choosing $\beta \approx 0$, which yields $\abs{\mathcal{O} - \mathcal{O}'} \leq O(\delta^{1/2}) + O(\varphi(\delta^{1/4}))$. 
\end{proof}

%%%%%%%%%%%%%%%%%%%%%%%%%%%%%%%%%%%%%%%%%%%%%%%%%%%%%%%%%%%%%%%%%%%%%%%%%%

\section{Stability of spin models}
\subsection{Proof of proposition \ref{prop:t_evol} (Dynamics of locally interacting spin systems)}
\label{app:t_evol}

\noindent In this section, we will consider the target problem to be a spatially local Lindbladian
\[
\mathcal{L} = \sum_{\alpha} \mathcal{L}_\alpha,
\]
where $\norm{\mathcal{L}_\alpha}_\diamond \leq 1$, and the support of $\mathcal{L}_\alpha$ denoted by $\Lambda_\alpha$ satisfies $\textnormal{diam}(\Lambda_\alpha) \leq R \ \forall \alpha$. In the presence of coherent errors and incoherent noise, the quantum simulator instead implements a Lindbladian
\begin{align}\label{eq:perturbed_lindbladian}
\mathcal{L}'(t) = \sum_{\alpha} \mathcal{L}_\alpha'(t), \text{ where }\mathcal{L}_\alpha'(t) =  \mathcal{L}_\alpha' - i \sum_{\alpha}[h_{\text{SE}, \alpha}(t), \cdot] \text{ with }h_{\text{SE}, \alpha}(t) = \sum_{j = 1}^{n_L} \big(L_{j, \alpha} A_{j, \alpha}^\dagger(t) + L_{j, \alpha}^\dagger A_{j, \alpha}(t)\big).
\end{align}
Here $\mathcal{L}'_\alpha$ is the Linbladian implemented on qubits in $\Lambda_\alpha$ due to coherent errors, and we assume that $\norm{\mathcal{L}_\alpha - \mathcal{L}_\alpha'} \leq \delta$. The Hamiltonian $h_{\text{SE}, \alpha}(t)$ captures interaction of the qubits contained in $\Lambda_\alpha$ with an external decohering non-Markovian environment --- the operators $A_{j, \alpha}(t)$ are assumed to be bosonic annihilation operators which satisfy $[A_{j, \alpha}(t), A_{j', \alpha'}^\dagger(t)] = \delta_{j, j'}\delta_{\alpha, \alpha'}K_{j, \alpha}(t)$ for a memory kernel $K_{j, \alpha}(t)$ and the operators $L_{j, \alpha}$ are system operators with support in $\Lambda_\alpha$ which also satisfy $\norm{L_\alpha} \leq \sqrt{\delta}$. Similar to the noise model assumed for Gaussian fermion models, we assume that $K_{j, \alpha}(\tau)$ can have delta function contributions i.e.
\[
K_{j, \alpha}(\tau) = K_{j, \alpha}^c(t) + \sum_{i = 1}^M k_{j, \alpha}^i \delta(\tau - \tau_i),
\]
where $K_{j, \alpha}^c(\tau)$ is a continuous function of $\tau$. We also assume that 
\begin{align}\label{eq:tv_upper_bound}
\int_{\mathbb{R}} \abs{K_{j, \alpha}(\tau)}d\tau = \int_{\mathbb{R}}\abs{K_{j, \alpha}^c(\tau)}d\tau + \sum_{i = 1}^M \abs{k_{j, \alpha}^i} \leq 1.
\end{align}
This bound can be interpreted in a distributional sense by viewing $K_{j, \alpha}$ as a map that takes a continuous compact function $f$ and maps it to a complex number given by the integral $\int_{\mathbb{R}} K_{j, \alpha}(\tau) f(\tau) d\tau$. Equation~\ref{eq:tv_upper_bound} can is then equivalent to requiring
\begin{align}\label{eq:upper_bound_kernel}
\bigabs{\int_{\tau_1}^{\tau_2} {K_{j, \alpha}(\tau)} f(\tau) d\tau} \leq \sup_{\tau \in [\tau_1, \tau_2]} \abs{f(\tau)} \ \text{for all continuous compact functions }f.
\end{align}
The main tool that we will use to prove the stability of local observables are the Lieb Robinson bounds for spatially local Lindbladians.
\begin{lemma}[Lieb-Robinson bounds, Ref.~\cite{hastings2004lieb, bravyi2006lieb}]\label{lemma:lieb_robinson_finite} Suppose $\mathcal{L} = \sum_{\alpha} \mathcal{L}_\alpha$ is a spatially local Lindbladian defined on a lattice $\mathbb{Z}^d_{L}$ such that $\norm{\mathcal{L}_\alpha}_\diamond \leq 1$ and $\mathcal{L}_\alpha$ is supported on sites in $\Lambda_\alpha$ with $\text{diam}(\Lambda_\alpha) \leq R$. Suppose $O$ is an observable supported in $S_O \subseteq \mathbb{Z}_L^d$, and $\mathcal{K}_Y$ is a super-operator satisfying $\mathcal{K}_Y(I) = 0$ supported in $Y \subseteq \mathbb{Z}_L^d$, then $\exists \mu, v > 0$, independent of the system size $n = L^d$, such that
\[
\norm{\mathcal{K}_Y(e^{\mathcal{L}^\dagger t}(O))} \leq \norm{O} \norm{\mathcal{K}_Y}_{\infty \to \infty, cb}\min\big(\abs{X}   e^{-\mu d(S_O, Y)} (e^{vt} - 1), 1).
\]
\end{lemma}
We also provide another lemma which can be interpreted as a generalization of the input-output equations used in the theory of open quantum systems.
\begin{lemma}\label{lemma:generalized_input_output_equations}
    Suppose ${\mathcal{E}}'(t, s) = \mathcal{T}\exp(\int_s^t \mathcal{L}'(\tau)d\tau)$ is the channel generated by the Lindbladian $\mathcal{L}'(t)$ in Eq.~\ref{eq:perturbed_lindbladian}. Then, 
    \[
    A_{j, \alpha}(t) {\mathcal{E}}'(t, 0)(\cdot) = \mathcal{E}'(t, 0)( A_{j, \alpha}(t) (\cdot)) - i\int_0^t K_{j, \alpha}(t - t') \mathcal{E}'(t, t')\big( [L_{j, \alpha}, \mathcal{E}'(t',0)(\cdot)]\big) dt'.
    \] 
\end{lemma}
\begin{proof} Consider, the superoperator $\mathcal{A}_{j, \alpha}(\tau; t)(\cdot) = {\mathcal{E}'}^{-1}(t, 0) (A_{j, \alpha}(\tau) \mathcal{E}'(t, 0)(\cdot))$ --- it follows that
\[
\frac{d}{dt}\mathcal{A}_{j, \alpha}(\tau; t) (\cdot) = {\mathcal{E}'}^{-1}(t, 0) \big(A_{j, \alpha}(\tau) \mathcal{L}'(t)(\mathcal{E}'(t, 0)(\cdot))) -   \mathcal{L}'(t)(A_{j, \alpha}(\tau)\mathcal{E}'(t, 0)(\cdot))\big).
\]
Now, from the explicit expression for $\mathcal{L}'(t)$, we can verify that for any operator $X$
\begin{align*}
A_{j, \alpha}(\tau) \mathcal{L}'(t)(X) - \mathcal{L}'(t)(A_{j, \alpha}(\tau) X) &=  i\big([h_{\text{SE}, \alpha}(t), A_{j, \alpha}(\tau) X] -A_{j, \alpha}(\tau)[h_{\text{SE}, \alpha}(t), X]\big) = -iK_{j, \alpha}(\tau - t)[L_{j, \alpha}, X],
\end{align*}
and therefore,
\[
\frac{d}{dt}\mathcal{A}_{j, \alpha}(\tau; t) (\cdot) = -iK_{j, \alpha}(\tau  - t){\mathcal{E}'}^{-1}(t, 0) \big([L_{j, \alpha}, \mathcal{E}'(t, 0)(\cdot)])\big).
\]
Integrating this equation from $0$ to $t$, and setting $\tau = t$,
\[
\mathcal{A}_{j, \alpha}(t; t)(\cdot) = \mathcal{A}_{j, \alpha}(t)(\cdot) - i\int_0^t K_{j, \alpha}(t - t'){\mathcal{E}'}^{-1}(0, t')\big([L_{j, \alpha}, \mathcal{E}'(t, 0)(\cdot)]\big) dt',
\]
from which the lemma statement follows. 
\end{proof}
An immediate useful consequence of lemma \ref{lemma:generalized_input_output_equations} is the following lemma.
\begin{lemma}\label{lemma:norm_anninilation_vacuum}
Suppose $\mathcal{E}'(t, s) = \mathcal{T}\exp(\int_s^t \mathcal{L}'(\tau) d\tau)$ is the channel generated by the Lindbladian $\mathcal{L}'(t)$ in Eq.~\ref{eq:perturbed_lindbladian}, and let $\rho(0)$ be an initial state which is vacuum in the decohering environment, then
\[
\norm{\textnormal{Tr}_E(A_{j, \alpha}(t)\mathcal{E}'(t, 0)(\rho(0)))}_{\op, 1} \leq 2\sqrt{\delta},
\]
where $\textnormal{Tr}_E$ is a partial trace over the decohering environment.
\end{lemma}
\begin{proof}
    This lemma can be proved by using lemma \ref{lemma:generalized_input_output_equations}. Since $\rho(0)$ is in the vacuum state in the environment, from lemma \ref{lemma:generalized_input_output_equations} it follows that
\begin{align}\label{eq:action_annihilation}
A_{j, \alpha}(t) \mathcal{E}'(t, 0) (\rho(0)) = -i\int_0^s K_{j, \alpha}(t - t')\mathcal{E}'(t, t')([L_{j, \alpha}, \mathcal{E}'(t', 0)(\rho(0))]) dt',
\end{align}
where we have used that $A_{j, \alpha}(s) \rho(0) = 0$. Using Eqs.~\ref{eq:upper_bound_kernel} and \ref{eq:action_annihilation} and the fact that quantum channels are contraction with respect to the one norm, we then obtain that
\begin{align}\label{eq:bound_annihilation_vac}
\bignorm{A_{j, \alpha}(s) \mathcal{E}'(s, 0) (\rho(0))}_{\op, 1} \leq \sup_{s' \in[0, s]} \bignorm{\mathcal{E}'(s, s')([L_{j, \alpha}, \mathcal{E}'(s, 0) (\rho(0))])}_{\op, 1}ds' \leq 2\norm{L_{j, \alpha}} \leq 2\sqrt{\delta}.
\end{align}
Noting that $\textnormal{Tr}_E$ is again a contraction with respect to the one norm, we obtain the lemma statement.
\end{proof}

\begin{lemma}\label{lemma:error_bound_individual_term_lr}
For any $t > 0$ and $s \in [0, t]$, initial state $\rho(0)$ which is vacuum in the decohering environment and local operator $O$ supported in $S_O$,
\begin{align*}
&\bigabs{\textnormal{Tr}\bigg(Oe^{\mathcal{L}(t - s)}(\mathcal{L}_\alpha' - \mathcal{L}_\alpha) \mathcal{E}'(s, 0)(\rho(0))\bigg)} \leq \delta \abs{S_O}\norm{O}\min\big(\big(e^{vt} - 1) e^{-\mu d(S_O, \Lambda_\alpha)}, 1\big), \text{ and } \\
&\bigabs{\textnormal{Tr}\bigg(Oe^{\mathcal{L}(t - s)}\big([h_{\textnormal{SE}, \alpha}(s), \mathcal{E}'(s, 0) (\rho(0))]\big)\bigg)} \leq 4n_L \delta \abs{S_O} \norm{O}\max\big(\big(e^{vt} - 1) e^{-\mu d(S_O, \Lambda_\alpha)}, 1\big).
\end{align*}
\end{lemma}
\begin{proof}
The first bound follows directly from the Lieb-Robinson's bound from lemma \ref{lemma:lieb_robinson_finite}. Note that the super-operator $(\mathcal{L}_\alpha - {\mathcal{L}'}_\alpha)^\dagger$ has $I$ in its null space, and $\norm{(\mathcal{L}_\alpha - {\mathcal{L}'}_\alpha)^\dagger}_{\infty \to \infty, cb} = \norm{\mathcal{L}_\alpha - {\mathcal{L}'}_\alpha}_{\diamond} \leq \delta$. Therefore, it follows from that
\begin{align*}
    \bigabs{\text{Tr}\bigg(Oe^{\mathcal{L}(t - s)} (\mathcal{L}_\alpha' - \mathcal{L}_\alpha)\mathcal{E}'(s, 0) (\rho(0)) \bigg)} &= \bigabs{\text{Tr}\bigg((\mathcal{L}_\alpha' - \mathcal{L}_\alpha)^\dagger (e^{\mathcal{L}^\dagger (t - s)}(O)) \mathcal{E}'(s, 0) (\rho(0)) \bigg)}, \nonumber\\
    & \leq \norm{(\mathcal{L}_\alpha' - \mathcal{L}_\alpha)^\dagger (e^{\mathcal{L}^\dagger (t - s)}(O))}, \nonumber\\
    & \leq \delta \abs{S_O} \norm{O} \min(e^{-\mu d(S_O, \Lambda_\alpha)}(e^{vt} - 1), 1),
\end{align*}
where we have used lemma \ref{lemma:lieb_robinson_finite} in the last step.

Next, we provide a similar upper bound on $\text{Tr}(Oe^{\mathcal{L}(t - s)}([h_{\text{SE}, \alpha}, \mathcal{E}'(s, 0)(\rho(0))]))$. Using the explicit expression for $h_{\text{SE}, \alpha}(t)$, we obtain that
\begin{align}\label{eq:bound_nmkv_term}
\bigabs{\text{Tr}(Oe^{\mathcal{L}(t - s)}([h_{\text{SE}, \alpha}, \mathcal{E}'(s, 0)(\rho(0))]))} &= \bigabs{\sum_{j = 1}^{n_L}\bigg(\text{Tr}_S(Oe^{\mathcal{L}(t - s)}([L_{j, \alpha}^\dagger ,\text{Tr}_E \big(A_{j, \alpha}(s)\mathcal{E}'(s, 0)(\rho(0))])) - \text{h.c.}\bigg)}, \nonumber \\
&\leq 2\sum_{j = 1}^{n_L} \bigabs{\text{Tr}_S\big(Oe^{\mathcal{L}(t - s)}([L_{j, \alpha}^\dagger ,\text{Tr}_E \big(A_{j, \alpha}(s)\mathcal{E}'(s, 0)(\rho(0))]))\big)}, \nonumber\\
&=2\sum_{j = 1}^{n_L} \bigabs{\text{Tr}_S\big([e^{\mathcal{L}^\dagger(t-s)}(O), L_{j, \alpha}^\dagger]\text{Tr}_E \big(A_{j, \alpha}(s)\mathcal{E}'(s, 0)(\rho(0)))\big)}, \nonumber \\
&\leq 2\sum_{j = 1}^{n_L} \norm{[e^{\mathcal{L}^\dagger(t- s)}(O), L_{j, \alpha}^\dagger]} \bignorm{\text{Tr}_E \big(A_{j, \alpha}(s)\mathcal{E}'(s, 0)(\rho(0))}_1
\end{align}
where $\text{Tr}_E$ is the partial trace over the decohering environment, and by $\text{Tr}_S$ is the partial trace over the system. Now, if the initial state $\rho(0)$ is vacuum in the environment then from lemma \ref{lemma:generalized_input_output_equations} we obtain
\begin{align}\label{eq:split_bound_nmkv}
\abs{\text{Tr}_S(Oe^{\mathcal{L}(t - s)}([L_{j, \alpha}^\dagger ,\text{Tr}_E \big(A_{j, \alpha}(s)\mathcal{E}'(s, 0)(\rho(0))]))} \leq \bignorm{[e^{\mathcal{L}^\dagger(t- s)}(O), L_{j, \alpha}^\dagger]} \bignorm{\text{Tr}_E \big(A_{j, \alpha}(s)\mathcal{E}'(s, 0)(\rho(0))}_{\op, 1}.
\end{align}
We note that since the super operator $[\cdot, L_{j, \alpha}^\dagger]$ is supported only on $\Lambda_\alpha$ and satisfies $[I, L_{j, \alpha}^\dagger] = 0$ together with $\norm{[\ \cdot\ , L_\alpha^\dagger]}_{\infty \to \infty, cb} \leq 2\norm{L_\alpha} \leq 2\sqrt{\delta}$. Therefore, from the Lieb Robinson bound in lemma \ref{lemma:lieb_robinson_finite}, we obtain the bound
\begin{align}\label{eq:bound_comm}
\norm{[O, L_{j, \alpha}^\dagger]} \leq 2\sqrt{\delta}\abs{S_O}  \norm{O}\max\big(e^{-\mu d(S_O, \Lambda_\alpha)} (e^{vt} - 1), 1\big).
\end{align}
From Eqs.~\ref{eq:bound_nmkv_term}, \ref{eq:split_bound_nmkv}, \ref{eq:bound_comm} and lemma \ref{lemma:norm_anninilation_vacuum}, we then obtain
\[
\bigabs{\text{Tr}(Oe^{\mathcal{L}(t - s)}([h_{\text{SE}, \alpha}(s), \mathcal{E}'(s, 0)(\rho(0)))])} \leq 8n_L \delta \abs{S_O}\norm{O}\max\big(e^{-\mu d({S_O, \Lambda_\alpha})}(e^{vt} - 1), 1\big).
\]
\end{proof}

\begin{proof}[Proof (of proposition \ref{prop:t_evol})]
Suppose $\mathcal{O} = \text{Tr}[O e^{\mathcal{L}t}(\rho_0)]$ and $\mathcal{O}'=\text{Tr}[O \mathcal{E}'(t, 0)(\rho_0)]$ are the observable expectation values in the noiseless and noisy settings respectively. We can express the error in the observable as
\begin{align*}
\abs{\mathcal{O}' - \mathcal{O}} &= \abs{\text{Tr}[O( \mathcal{E}'(t, 0)(\rho_0) - e^{\mathcal{L}t}(\rho_0))]} \leq\int_0^t \bigabs{\text{Tr}\bigg(O e^{\mathcal{L}(t - s)} (\mathcal{L}'(s) - \mathcal{L})\mathcal{E}'(s, 0) (\rho(0))\bigg)}ds, \nonumber \\
&\leq \sum_{\alpha}\bigg(\int_0^t \bigabs{\text{Tr}\bigg(Oe^{\mathcal{L}(t - s)} (\mathcal{L}_\alpha' - \mathcal{L}_\alpha)\mathcal{E}'(s, 0) (\rho(0)) \bigg)}ds + \int_0^t \bigabs{\text{Tr}\bigg(Oe^{\mathcal{L}(t - s)} [h_{\text{SE}, \alpha}(s),\mathcal{E}'(s, 0) (\rho(0))] \bigg)}ds\bigg).
\end{align*}
Using lemma \ref{lemma:error_bound_individual_term_lr}, we obtain that
\[
\abs{\mathcal{O} - \mathcal{O}'} \leq \delta t \abs{S_O}\norm{O} (1 + 8n_L) \sum_{\alpha}\text{max}\big(e^{-\mu_\alpha d(S_O, \Lambda_\alpha)}(e^{vt} - 1), 1\big) \leq O(\delta t^{d + 1}).
\]

\end{proof}

%%%%%%%%%%%%%%%%%%%%%%%%%%%%%%%%%%%%%%%%%%%%%%%%%%%%%%%%%%%%%%%%%%%%%%%

\subsection{Proof of proposition \ref{prop:gs} (Ground states of gapped local Hamiltonians)}
\label{app:gs}

We will apply the formalism developed in Ref.~\cite{bachmann2012automorphic} for spectral flows for families of gapped Hamiltonians. We are interested in a target spatially local Hamiltonian $H$, expressed as
\[
H = \sum_{x \in \mathfrak{L}} h_x,
\]
where $h_x$ acts only on spins with a distance $R$ of $x \in \mathfrak{L}$, and $\norm{h_x} \leq 1$. The implemented Hamiltonian $H'$ is assumed to have a similar form,
\[
H' = \sum_{x \in \mathfrak{L}} \big(h_x + v_x\big),
\]
where $\norm{v_x} \leq \delta$ for all $x \in \mathfrak{L}$. We assume that $H$ is stably gapped with gap $\Delta$ i.e.~any $H'$ of the above form has an energy gap between the ground state and the first excited state that is larger than $\Delta$. We consider the family of Hamiltonians, $H_s$, for $s \in [0, 1]$, defined by
\[
H_s = H + s(H' - H) = \sum_{x \in \mathfrak{L}}h_x + s v_x,
\]
and note that the assumption of being stably gapped is equivalent to $H_s$ being gapped, with the gap being larger than $\Delta$, for all $s \in [0, 1]$. Now, the spectral flow method allows us to construct a unitary $U(s)$ that relates the ground state $\ket{G_{s = 0}}$ of $H_{s = 0} = H$ to the ground state $\ket{G_s}$ of $H_s$ as provided in the following lemma.
\begin{lemma}[From Ref.~\cite{bachmann2012automorphic}]
Consider the unitary $U(s)$ obtained from
\[
\frac{d}{ds}U(s) = iD(s) U(s) \ \text{where } D(s) = \int_{-\infty}^\infty W_{\Delta}(t) e^{-itH_s} (H' - H) e^{itH_s} dt
\]
where $W_\Delta \in L^1(\mathbb{R})$ is a real valued odd function which satisfies
\begin{enumerate}
    \item[(a)] $\abs{W_\Delta(t)}$ is bounded and satisfies
    \begin{align}\label{eq:prop_W_inf_norm}
    \norm{W_\Delta}_\infty = \sup_{t\in\mathbb{R}}\abs{W_\Delta(t)} = \frac{1}{2}.
    \end{align}
    \item[(b)] For $t > 0$, the function $I_\Delta(t) = \int_{t}^\infty \abs{W_\Delta(s)}ds$ satisfies
    \begin{align}\label{eq:prop_W_integral}
    I_\Delta(t) \leq G(\Delta t),
    \end{align}
    where $G(x)$ falls off faster than any polynomial as $x \to \infty$.
\end{enumerate}
Then, $\ket{G_s} = U(s) \ket{G_{s = 0}}$, where $\ket{G_s}$ is the ground state of $H(s)$.
\end{lemma}
\noindent\emph{Proof (of proposition \ref{prop:gs})}. Using this result, we can straightforwardly show the stability of the quantum simulation task of computing a local observable in the ground state of $H$. To see this, we note that
\[
\abs{\bra{G_0} O \ket{G_0} - \bra{G_s} O \ket{G_s}} = \bigabs{\bra{G_0}\bigg(O - U^\dagger(s) O U(s) \bigg) \ket{G_0}} \leq \norm{O - U^\dagger(s) O U(s)}\leq \int_0^s \norm{[O, D(s')]}ds'.
\]
It then remains to bound $\norm{[O, D(s')]}$ --- we can do this by following lemma 4.7 in Ref.~\cite{bachmann2012automorphic}, and we reproduce this below --- we start by noting that
\[
\norm{[O, D(s')]} \leq \sum_{x \in \mathfrak{L}} \bignorm{\int_{-\infty}^\infty W_\Delta(t)[O, e^{itH_s} v_x e^{-itH_s}]dt}.
\]
For each term in this summation, we further split the integral and bound it as
\[
\bignorm{\int_{-\infty}^\infty W_\Delta(t)[O, e^{itH_s} v_x e^{-itH_s}]dt} \leq \bignorm{\int_{\abs{t} \leq T_x} W_\Delta(t)[O, e^{itH_s} v_x e^{-itH_s}]dt} + \bignorm{\int_{\abs{t} > T_x}W_\Delta(t)[O, e^{itH_s} v_x e^{-itH_s}]dt}.
\]
For the first term, which only concerns with $\abs{t}\leq T_x$, we use the Lieb-Robinson bound (lemma \ref{lemma:lieb_robinson_finite}) and Eq.~\ref{eq:prop_W_inf_norm} to obtain
\[
\bignorm{\int_{\abs{t} \leq T_x} W_\Delta(t)[O, e^{itH_s} v_x e^{-itH_s}]dt} \leq \norm{O}\norm{v_x} \abs{S_O}e^{-\mu d(S_O, S_{v_x})} \int_0^{T_x} (e^{vt} - 1)dt \leq \norm{O}\norm{v_x} \abs{S_O}\frac{e^{-\mu d(S_O, S_{v_x})}e^{vT_x}}{v}.
\]
For the second term for $\abs{t} \geq T_x$, we use Eq.~\ref{eq:prop_W_integral} together with the fact that $W_\Delta$ is an odd function and the simple bound $\norm{[O, e^{itH_s} v_x e^{-itH_s}]} \leq 2\norm{O}\norm{v_x}$ to obtain that
\[
\bignorm{\int_{\abs{t} \geq T_x} W_\Delta(t)[O, e^{itH_s} v_x e^{-itH_s}]dt} \leq 2\norm{O}\norm{v_x}\int_{\abs{t}\geq T_x} \abs{W_\Delta(t)} \leq 4\norm{O}\norm{v_x} G(\Delta T_x),
\]
Note that $T_x$ can be arbitrary in the above two estimates --- choosing $T_x = \mu d(S_O, S_{v_x}) /2v$, we obtain that
\[
\bignorm{\int_{-\infty}^\infty W_\Delta(t)[O, e^{itH_s} v_x, e^{-itH_s}]dt} \leq \norm{O}\norm{v_x} \bigg[\frac{\abs{S_O}}{v} e^{-\mu d(S_O, S_{v_x})/2} + 2G\bigg(\frac{\Delta \mu}{v}d(S_O, S_{v_x})\bigg)\bigg],
\]
and therefore, for all $s' \in [0, s]$, we obtain the bound
\[
\norm{[O, D(s')]} \leq \norm{O}\delta \sum_{x \in \mathfrak{L}} \bigg[\frac{\abs{S_O}}{v} e^{-\mu d(S_O, S_{v_x})/2} + 2G\bigg(\frac{\Delta \mu}{v}d(S_O, S_{v_x})\bigg)\bigg].
\]
Noting that the summand in the above expression decreases faster than any polynomial in $d(S_O, S_{v_x})$, we see that it will be upper bounded by a constant independent of the size of the lattice $\mathfrak{L}$, thus independent of $n$. This proves the proposition. $\hfill \square$

%%%%%%%%%%%%%%%%%%%%%%%%%%%%%%%%%%%%%%%%%%%%%%%%%%%%%%%%%%%%%%%%%%%%%%%%%

\subsection{Proof of proposition \ref{prop:Gibbs} (Gibbs state with exponential clustering of correlations)}
\label{app:Gibbs}
We begin by presenting a proof of a standard bound on the perturbation of Gibbs states of a Hamiltonian. This can be found in \cite{alhambra2022quantum}, and we reproduce it here for the convenience of the reader.
\begin{lemma}\label{lemma:error_partition_func}
Suppose $H$ and $V$ are two Hermitian bounded operators, then for any $\beta \geq 0$:
\begin{enumerate}
    \item[(a)] The partition functions satisfy
\[
\textnormal{Tr}(e^{-\beta (H +V)}) \leq \textnormal{Tr}(e^{-\beta H}) e^{\beta \norm{V}}
\]
\item[(b)] For any $O \succeq 0$, it follows that
\[
\textnormal{Tr}\big(O e^{-\beta (H + V)}\big) \leq \textnormal{Tr}\big(Oe^{-\beta H}\big) \exp\big(e^{\beta (\norm{H} + \norm{V})} \norm{V}\big).
\]
\end{enumerate}    
\end{lemma}
\begin{proof}
\noindent(a) We will use Duhamel's formula, which states that for any differentiable bounded operator $F(t)$,
\[
\frac{d}{dt}e^{F(t)} = \int_0^1 e^{(1 - u)F(t)}\frac{dF(t)}{dt}e^{uF(t)}du.
\]
Defining $H(s) = H + sV$, we note from Duhamel's formula that
\[
\bigabs{\frac{d}{ds}\text{Tr}\big(e^{-\beta H(s)}\big)} = \beta\bigabs{\int_0^1 \text{Tr}\big(e^{-(1 - u)\beta H(s)} V e^{-u\beta H(s)}\big)du} =  \beta\bigabs{\int_0^1 \text{Tr}\big(e^{-\beta H(s)} V \big)du} \leq \beta \norm{V}_\text{op}\norm{e^{-\beta H(s)}}_{\text{op}, 1},
\]
where we have used Hölder's inequality in the last step. Noting that $\norm{e^{-\beta H(s)}}_{\text{op}, 1} = \text{Tr}(e^{-\beta H(s)})$, we obtain that
\[
\bigabs{\frac{d}{ds}\text{Tr}\big(e^{-\beta H(s)}\big)} \leq \beta \norm{V}_\text{op} \text{Tr}\big(e^{-\beta H(s)}\big)
\]
Therefore,
\[
\abs{\log\text{Tr}(e^{-\beta(H+V)}) - \log \text{Tr}(e^{-\beta H})} \leq \int_0^1  \bigabs{\frac{d}{ds}\log \text{Tr}(e^{-\beta H(s)})}ds = \int_0^1 \bigabs{\frac{1}{\text{Tr}(e^{-\beta H(s)})} \frac{d}{ds}\text{Tr}(e^{-\beta H(s)})}ds \leq \beta \norm{V}_\text{op}.
\]
Thus, we obtain the lemma statement. \\ \ \\
\noindent(b) We again use the Duhamel's formula to obtain
\begin{align*}
&\bigabs{\frac{d}{ds}\text{Tr}(Oe^{-\beta H(s)})} = \beta \bigabs{\int_0^1 \text{Tr}\big(O e^{-(1 - u) \beta H(s)} V e^{-u \beta H(s)}\big)du} \leq \beta \int_0^1 \bigabs{\text{Tr}\big(O e^{-(1 - u) \beta H(s)} V e^{-u \beta H(s)}\big)}du
\end{align*}
We note that
\begin{align*}
\bigabs{\text{Tr}\big(O e^{-(1 - u) \beta H(s)} V e^{-u \beta H(s)}\big)} &= \bigabs{\text{Tr}\big(e^{-\beta H(s)/2}O e^{-\beta H(s)/2}e^{-(1/2 - u) \beta H(s)} V e^{-(u - 1/2) \beta H(s)}\big)},\\
&\leq\text{Tr}\big(e^{-\beta H(s)/2}O e^{-\beta H(s)/2}\big)\norm{e^{-(1/2 - u) \beta H(s)} V e^{-(u - 1/2) \beta H(s)}}, \\
&\leq \text{Tr}\big(O e^{-\beta H(s)}\big) \norm{V}e^{\beta (\norm{H} + s\norm{V})}.
\end{align*}
where we have used the fact that since $u \in [0, 1]$, $\abs{u - 1/2} \leq 1/2$. Therefore, we obtain that
\[
\bigabs{\frac{d}{ds}\log\textnormal{Tr}\big(Oe^{-\beta H(s)}\big)} \leq \norm{V}e^{\beta \norm{H}}e^{\beta s\norm{V}} \implies \bigabs{\log\text{Tr}(Oe^{-\beta H}) - \log\text{Tr}(Oe^{-\beta (H + V)})} \leq \norm{V}e^{\beta(\norm{H} + \norm{V})}.
\]
This estimate yields the lemma statement.
\end{proof}
\begin{lemma}\label{lemma:pert_theory_gibbs}
    Given bounded Hermitian operators $H$ and $V$, and any bounded Hermitian operator (observable) O,
    \[
    \bigabs{\textnormal{Tr}\bigg(\frac{O e^{-\beta H}}{Z_H(\beta)}\bigg) - \textnormal{Tr}\bigg(\frac{O e^{-\beta (H + V)}}{Z_{H + V}(\beta)}\bigg)} \leq 2\norm{O}_\textnormal{op}\bigg(\bigabs{\exp(e^{\beta (\norm{H} + \norm{V})}\norm{V}) -1} +\bigg|{\exp(\beta\norm{V}) -1}\bigg| \bigg)
    \]
\end{lemma}
\begin{proof}
This is a straightforward application of lemma \ref{lemma:error_partition_func}. We denote by $H' = H + V$. For simplicity, we will analyze the operator $O' = O + \norm{O}I$, since $O'\succeq 0$ and 
\[
\text{Tr}\bigg(\frac{O'e^{-\beta H}}{Z_H(\beta)}\bigg) - \text{Tr}\bigg(\frac{O'e^{-\beta H'}}{Z_{H'}(\beta)}\bigg) = \text{Tr}\bigg(\frac{Oe^{-\beta H}}{Z_H(\beta)}\bigg) - \text{Tr}\bigg(\frac{Oe^{-\beta H'}}{Z_{H'}(\beta)}\bigg).
\]
We begin by noting that
\[
 \bigabs{\textnormal{Tr}\bigg(\frac{O' e^{-\beta H}}{Z_H(\beta)}\bigg) - \textnormal{Tr}\bigg(\frac{O' e^{-\beta H'}}{Z_{H'}(\beta)}\bigg)} \leq \frac{\textnormal{Tr}(O'e^{-\beta H})}{Z_H(\beta)}\bigabs{ 1 - \frac{\textnormal{Tr}(O'e^{-\beta H'})}{\textnormal{Tr}(O'e^{-\beta H})}} + \frac{\text{Tr}(O'e^{-\beta H'})}{Z_{H'}(\beta)}\bigabs{1 - \frac{Z_{H'}(\beta)}{Z_H(\beta)}}
\]
Noting that
\[
\frac{\textnormal{Tr}(O'e^{-\beta H})}{Z_H(\beta)}, \frac{\textnormal{Tr}(O'e^{-\beta H'})}{Z_{H'}(\beta)} \leq \norm{O'}_\text{op} \leq 2\norm{O}_\text{op},
\]
we obtain that
\[
\bigabs{\textnormal{Tr}\bigg(\frac{O' e^{-\beta H}}{Z_H(\beta)}\bigg) - \textnormal{Tr}\bigg(\frac{O' e^{-\beta H'}}{Z_{H'}(\beta)}\bigg)} \leq 2\norm{O}_\text{op}\bigg(\bigabs{\exp(e^{\beta (\norm{H} + \norm{V})}\norm{V}) -1} +\bigg|{\exp(\beta\norm{V}) -1}\bigg| \bigg)
\]
\end{proof}
%%%%%%%%%%%%%%%%%%%%%%%%%%%%%%%%%%%%%%%%%%%%%%%%%%%%%%%%%%%%%%%%%%%%%%%%%
\noindent We next need the notion of exponentially-clustered correlations in a Gibbs state --- which we reproduce below from Ref.~\cite{brandao2019finite}. We will consider Hamiltonians on $\mathfrak{L}\subset\mathbb{Z}^d$ expressed as
\[
H = \sum_{x \in \mathfrak{L}}h_x,
\]
where $h_x$ acts only on spins within a distance $R$ of $x\in \mathfrak{L}$ and $\norm{h_x}\leq 1$. We will denote by $\text{supp}(h_x) \subseteq \mathfrak{L}$ the support of $h_x$. Given $X \subseteq \mathfrak{L}$, we denote by $H_X$ the operator
\[
H_X  = \sum_{x | \text{supp}(h_x) \subseteq X} h_x,
\]
i.e.~$H_X$ is the Hamiltonian $H$ obtained on restricting $H$ to the set $X$.
\begin{definition}
A local Hamiltonian $H$ is said to have exponential clustering of correlations at inverse temperature $\beta$ if $\exists c_1, c_2 > 0$ such that for all $X \subset \mathfrak{L}$ and operators $A, B$ with $\textnormal{supp}(A), \textnormal{supp}(B) \subset X$ with $d(\textnormal{supp}(A), \textnormal{supp}(B)) \geq l$, 
\[
\bigabs{\textnormal{Tr}\big((A\otimes B) \sigma_X(\beta) \big) -\textnormal{Tr}\big(A \sigma_X(\beta) \big)\textnormal{Tr}\big(B \sigma_X(\beta) \big) } \leq c_2 \norm{A}\norm{B} e^{-c_1 l},
\]
where $\sigma_X(\beta) = e^{-\beta H_X} / \textnormal{Tr}[e^{-\beta H_X}]$ is the Gibbs state corresponding to $H_X$ at inverse temperature $\beta$.
\end{definition}
\noindent An important property of Hamiltonians with exponential clustering of correlations, which relies on quantum belief propagation \cite{hastings2007quantum} and is proved in Ref.~\cite{brandao2019finite}, is that local observables can be estimated locally.
\begin{lemma}[From Ref.~\cite{brandao2019finite}]\label{lemma:exp_clust_corr}
Suppose $H$ is a local Hamiltonian on a finite lattice $\mathfrak{L}\subset \mathbb{Z}^d$ with exponential clustering of correlations at inverse temperature $\beta$. If $\mathfrak{L} = \mathcal{A} \cup \mathcal{B} \cup \mathcal{C}$ such that $\text{dist}(\mathcal{A}, \mathcal{C})\geq l$, then $\exists c_1', c_2'$ such that
\[
\bignorm{\textnormal{Tr}_{\mathcal{B}, \mathcal{C}}(\sigma_\mathfrak{L}(\beta)) - \textnormal{Tr}_\mathcal{B}(\sigma_{\mathcal{A}\cup \mathcal{B}}(\beta))}_\textnormal{tr} \leq \abs{\partial \mathcal{C}}c_2' e^{-c_1'l},
\]  
where $\sigma_X(\beta)$ is the Gibbs state corresponding to $H_X$ and $\partial \mathcal{C}$ is the boundary between $B, C$.
\end{lemma}
\noindent\emph{Proof (of proposition \ref{prop:Gibbs})}. We assume that both $H$ and $H'$ have exponential clustering of correlations and satisfy lemma \ref{lemma:exp_clust_corr}. Suppose $O$ is a local observable with support $S_O$ and consider $\mathcal{B}$ to be a region around $S_O$ and $\mathcal{C}$ be the remainder of the lattice. We also assume that $d(\mathcal{C}, S_O) \geq l$, for some $l$ to be chosen later. We denote by $\sigma_l(\beta)$ and $\sigma_l'(\beta)$ the Gibbs state, at inverse temperature $\beta$, corresponding to $H_{S_O \cup \mathcal{B}}$ and $H'_{S_O \cup \mathcal{B}}$ respectively, and by $\sigma(\beta), \sigma'(\beta)$ the Gibbs state corresponding to $H$ and $H'$. Now, from lemma \ref{lemma:exp_clust_corr} it follows that
\[
\abs{\text{Tr}(O\sigma(\beta)) - \text{Tr}(O\sigma_l(\beta))}, \abs{\text{Tr}(O\sigma'(\beta)) - \text{Tr}(O\sigma_l'(\beta))} \leq \norm{O}d(2l + R_O)^{d - 1} c_2'e^{-c_1' l},
\]
where $R_O = \text{diam}(S_O)$ and we have used that $\abs{\partial \mathcal{C}} \leq d\times \textnormal{diam}(S_O \cup \mathcal{B})^{d - 1} \leq d(2l + R_O)^{d-1}$. Furthermore, lemma \ref{lemma:pert_theory_gibbs} can be used to bound $\abs{\text{Tr}(O\sigma_l(\beta)) - \text{Tr}(O\sigma_l'(\beta))}$. We note that
\[
\norm{H_{S_O\cup \mathcal{B}} - H_{S_O \cup \mathcal{B}}'} \leq \delta (2l + R_O)^d \text{ and }\norm{H_{S_O\cup \mathcal{B}}} \leq (2l + R_O)^d.
\]
Therefore,
\[
\abs{\text{Tr}(O\sigma_l(\beta)) - \text{Tr}(O\sigma_l'(\beta))} \leq 2\norm{O}O(e^{\beta (2l + R_O)^d} (2l + R_O)^d \delta).
\]
Thus, from the triangle inequality we obtain the bound that
\[
\abs{\text{Tr}(O\sigma(\beta)) - \text{Tr}(O\sigma'(\beta))} \leq \norm{O}\bigg[O\bigg((2l + R_O)^{d - 1} e^{-c_1' l}\bigg) + O\bigg(e^{\beta(2l + R_O)^d}(2l + R_O)^d \delta\bigg)\bigg].
\]
For $\beta = \Theta(1)$, choosing $2l + R_O = \Theta(\log^{1/d}(1/\sqrt{\delta}))$, we obtain that
\[
\bigabs{\text{Tr}(O\sigma(\beta)) - \text{Tr}(O\sigma'(\beta))} \leq \norm{O} O(\log^{1 - 1/d}(1/\delta)e^{-\Omega(\log^{1/d}(1/\delta))}\big),
\]
which proves the lemma statement. $\hfill\square$
\subsection{Fixed point of rapidly mixing Lindbladians}\label{app:fp}
In this section, we consider spatially local Lindbladians which are also rapidly mixing as defined in Ref.~\cite{cubitt2015stability}. We establish that local observables meaured in the fixed point is robust to both coherent and incoherent errors. We point out that Ref.~\cite{cubitt2015stability} already proved this statement if the incoherent errors are considered to be Markovian --- we show that this stability of local observables is retained in the more general setting of incoherent non-Markovian errors. We first reproduce the definition of rapidly mixing Lindbladians from Ref.~\cite{cubitt2015stability}.
\begin{assumption}\label{assump:rapid_mixing}
        Suppose $\mathcal{L} = \sum_{\alpha}\mathcal{L}_\alpha$ is a spatially local Lindbladian on $\mathbb{Z}_L^d$ where $\mathcal{L}_\alpha$ are supported in $S_\alpha$, and let $\mathcal{L}_{\Lambda_l} = \sum_{\alpha : S_\alpha \subseteq \Lambda_l}\mathcal{L}_\alpha$ be the Lindbladian obtained from $\mathcal{L}$ by only retaining $\mathcal{L}_\alpha$ which are supported entirely in sublattice ${\Lambda_l}\cong \mathbb{Z}_l^d$. Then, for any choice of $\Lambda_l$, $\mathcal{L}_{\Lambda_l}$ has a unique fixed point $\sigma_{\Lambda_l}$ and
        \[
        \norm{e^{\mathcal{L}_{\Lambda_l}t}(\cdot) - \sigma_{\Lambda_l}\textnormal{Tr}(\cdot)}_{\diamond} \leq \kappa(l)e^{-\gamma t},
        \]
        where $\gamma > 0$ is some constant and $\kappa(l) \leq O(\textnormal{poly}(l))$.
\end{assumption}
An important consequence of this assumption, established in Ref.~\cite{cubitt2015stability}, is the local rapid mixing property which we state below.
\begin{lemma}[Ref.~\cite{cubitt2015stability}]\label{lemma:local_rapid_mixing}
Suppose $\mathcal{L} = \sum_{\alpha}\mathcal{L}_\alpha$ is a spatially local Lindbladian which satisfies assumption \ref{assump:rapid_mixing} and $O$ is a local observable with support $S_O$. Then, 
\[
\norm{e^{\mathcal{L}^\dagger t}(O) - I \ \textnormal{Tr}(\sigma O)} \leq \norm{O} k(\abs{S_O}) e^{-\gamma t},
\]
for some $k(l) \leq O(\textnormal{poly}(l))$.
\end{lemma}
\begin{proof}
    We will show that for any time $t > 0$, the observable obtained by the noisy quantum simulator is closs to the true observable if assumption \ref{assump:rapid_mixing} is satisfied. Denoting by $\mathcal{O}$ the true observable, and by $\mathcal{O}'$ the observable in the presence of errors and noise, we start with
    \begin{subequations}
    \begin{align}\label{eq:perturbation_starting_fixed_point}
    &\abs{\mathcal{O} - \mathcal{O}'} = \bigabs{\int_0^t \text{Tr}\bigg(e^{\mathcal{L}^\dagger(t - s)}(O)(\mathcal{L}'(s) - \mathcal{L}) (\mathcal{E}'(s, 0)(\rho(0)))\bigg) ds} \nonumber,\\
    &\qquad \quad \ \ \leq \sum_{\alpha}\int_0^t \bigabs{\text{Tr}\bigg(e^{\mathcal{L}^\dagger(t - s)}(O)\mathcal{L}_\alpha(s) (\mathcal{E}'(s, 0)(\rho(0)))\bigg)}ds, \nonumber\\
    &\qquad \quad \ \ \leq \sum_{\alpha} \int_0^t \bigabs{\text{Tr}\bigg(e^{\mathcal{L}^\dagger s}(O)\mathcal{L}_\alpha(t - s) (\mathcal{E}'(t - s, 0)(\rho(0)))\bigg)}ds, \nonumber\\
    &\qquad \quad \ \ \leq \sum_{\alpha} \big(e_{\alpha, <}(t_\alpha, t) + e_{\alpha, >}(t_\alpha, t)\big).
    \end{align}
    where $t_\alpha$ remains to be chosen and
    \begin{align}
    &e_{\alpha, <}(t_\alpha, t) =  \bigabs{\int_0^{\text{min}(t, t_\alpha)} \text{Tr}\big(e^{\mathcal{L}^\dagger s}(O)\mathcal{L}_\alpha'(t - s) (\mathcal{E}'(t - s, 0)(\rho(0)))\big) ds}
    \text{ and } \\
    &e_{\alpha, >}(t_\alpha, t) = 
    \bigabs{\int_{\min(t, t_\alpha)}^{t} \text{Tr}\bigg(e^{\mathcal{L}^\dagger s}(O)\mathcal{L}_\alpha'(t - s)(\mathcal{E}'(t - s, 0)(\rho(0)))\bigg) ds}
    \end{align}
    \end{subequations}
    We now upper bound both $e_{<}(t_\alpha, t)$ and $e_{>}(t_\alpha, t)$ separately. We can use lemma \ref{lemma:error_bound_individual_term_lr} to obtain 
    \begin{align}\label{eq:error_short_time}
    e_{\alpha, <}(t_\alpha, t) \leq \delta t_\alpha \abs{S_O}\norm{O}(1 + 8n_L) \max(e^{-\mu d(S_O, \Lambda_\alpha)} (e^{vt_\alpha} - 1), 1).
    \end{align}
     Next, consider $e_{\alpha, >}(t_\alpha, t)$ --- this is $0$ unless $t_\alpha < t$. Furthermore, we use lemma \ref{lemma:local_rapid_mixing}, from which it follows that
    \begin{align*}
    e_{\alpha, >}(t_\alpha, t) &\leq \bigabs{\int_{\min(t, t_\alpha)}^t \text{Tr}(O\sigma)\text{Tr}\bigg(\mathcal{L}_\alpha'(t- 
 s)\big(\mathcal{E}'(t - s, 0)(\rho(0)))\bigg)ds} + \nonumber \\ &\qquad \qquad \norm{O}k(\abs{S_O})  \int_{\min(t, t_\alpha)}^t e^{-\gamma s}{\bignorm{\text{Tr}_E(\mathcal{L}_{\alpha}'(t - s)(\mathcal{E}'(t - s, 0)(\rho(0)))}_{\op, 1} }ds, \nonumber \\
 &= \norm{O}k(\abs{S_O})  \int_{\min(t, t_\alpha)}^t e^{-\gamma s}{\bignorm{\text{Tr}_E(\mathcal{L}_{\alpha}'(t - s)(\mathcal{E}'(t - s, 0)(\rho(0)))}_{\op, 1} }ds, \nonumber\\
 &\leq \norm{O}k(\abs{S_O})\bigg( \int_{\min(t, t_\alpha)}^t e^{-\gamma s} \bignorm{\text{Tr}_E((\mathcal{L}_\alpha - \mathcal{L}_\alpha')(\mathcal{E}'(t - s, 0)(\rho(0)))}_{\op, 1} ds + \nonumber\\
 &\qquad \qquad 2\sum_{j = 1}^{n_L} \int_{\min(t, t_\alpha)}^t e^{-\gamma s} \norm{L_\alpha^\dagger\text{Tr}_E( A_{j, \alpha} \mathcal{E}'(t - s, 0)(\rho(0)))}_{\op, 1}\bigg).
    \end{align*}
Using $\norm{\mathcal{L}_\alpha - \mathcal{L}_\alpha'}_\diamond \leq \delta$ and lemma \ref{lemma:norm_anninilation_vacuum}, we obtain that
\begin{align}\label{eq:error_long_time}
e_{\alpha, >}(t_\alpha, t) \leq \norm{O}k(\abs{S_O})\delta \big(1 + 4n_L) \int_{\min(t, t_\alpha)}^t e^{-\gamma s} ds \leq 
    \norm{O}k(\abs{S_O})\delta \big(1 + 4n_L) e^{-\gamma t_\alpha} .
\end{align}
Finally, from Eqs.~\ref{eq:perturbation_starting_fixed_point}, \ref{eq:error_short_time} and \ref{eq:error_long_time}, we obtain that
\[
\abs{\mathcal{O} - \mathcal{O}'} \leq \norm{O}\delta (1 + 8n_L) \sum_{\alpha} \bigg(k(\abs{S_O}) e^{-\gamma t_\alpha} + \abs{S_O} t_\alpha \max\big(e^{-\mu d(S_O, \Lambda_\alpha)} \big(e^{vt_\alpha} - 1), 1\big)\bigg).
\]
Now, we make a choice of $t_\alpha = \mu d(S_O, \Lambda_\alpha)/2v$, we obtain that
\[
\abs{\mathcal{O} - \mathcal{O}'} \leq \norm{O}\delta (1 + 8n_L) \sum_{\alpha} \bigg(k(\abs{S_O}) e^{-\gamma \mu d(S_O, \Lambda_\alpha)/2v} + \abs{S_O} \frac{\mu d(S_O, \Lambda_\alpha)}{2v} e^{-\mu d(S_O, \Lambda_\alpha)/2}\bigg).
\]
The summation in the above bound converges to a constant independent of the system size, and we thus obtain that $\abs{\mathcal{O} - \mathcal{O}'} \leq O(\delta)$.
\end{proof}
%}
    
    .

%{
\section{Lower bounds on convergence to thermodynamic limits}
\subsection{Dynamics}\label{app:lower_bounds_dyn}
\noindent We show that the simple example of a translationally invariant free-fermionic tight-binding model on a $d-$dimensional lattice saturates the lower bound on the convergence to thermodynamic limit provided in lemma \ref{lemma:tl_dynamics}. The model that we consider is simple fermionic tight-binding model on a $d-$dimensional lattice $\mathfrak{L}^d$, where $\mathfrak{L} = \{-L, -L + 1 \dots L-1\}$, with imaginary hopping amplitudes
\[
H_n = i\sum_{x\in \mathfrak{L}^d} \sum_{m = 1}^d \bigg(a_{x}^\dagger a^{\phantom{\dagger}}_{x + e_m} - \text{h.c.}\bigg), 
\]
where $n = (2L)^d$, $e_m$ is the unit vector along the $m^\text{th}$ lattice direction and a periodic boundary condition is assumed. For simplicity, we will assume that $L$ is even. This Hamiltonian can be analytically diagonalized to obtain
\[
H_n = \sum_{k \in \mathfrak{L}^d} \omega_k \tilde{a}_k^\dagger \tilde{a}^{\phantom{\dagger}}_k,
\]
where
\[
\tilde{a}_k = \frac{1}{(2L)^{d/2}} \sum_{x\in\mathfrak{L}^d} a_x e^{i\pi k\cdot x / L} \text{ and } \omega_k = 2 \sum_{m = 1}^d \sin \bigg(\frac{\pi k_m} {L}\bigg).
\]
We now consider the local observable $O = a_0^\dagger a^{\phantom{\dagger}}_0$, and an initial state $\ket{\psi_0} = a_0^\dagger \ket{\text{vac}}$, and provide a lower bound on the error 
\begin{align}\label{eq:app_H_dynamics_error_def}
\mathcal{E}_n = \bigabs{\bra{\psi_0} e^{iH_n t} O e^{-iH_n t} \ket{\psi_0} - \lim_{n \to \infty}\bra{\psi_0} e^{iH_n t} O e^{-iH_n t} \ket{\psi_0}}
\end{align}
We first establish some useful technical lemmas.
\begin{lemma}\label{lemma:trig_sums_app_h}
For $p \in \mathbb{N}$, it follows that
\item[(a)] If $L \in \mathbb{N}$ is even then,
    \[
   \sum_{m = 0}^{L - 1} \sin^{2p}\bigg(\frac{m\pi}{L}\bigg) = \frac{L}{2^{2p}} {2p \choose p} + \frac{L}{2^{2p - 1}} \sum_{1 \leq k\leq p/L} {2p \choose p + kL}.
    \]
\item[(b)] The integral
\[
\int_0^{\pi} \sin^{2p}(\theta) d\theta = \frac{\pi}{2^{2p}} {2p \choose p}.
\]
\end{lemma}
\noindent\emph{Proof}: (a) We begin by noting that
\begin{align}\label{ref:app_H1_trig_id_1}
\sum_{m = 0}^{L-1}\sin^{2p}\bigg(\frac{\pi m}{L}\bigg) =  2\sum_{m = 0}^{L/2 - 1}\cos^{2p}\bigg(\frac{m \pi}{L}\bigg) - 1.
\end{align}
Next, we have that
\[
\sum_{m = 0}^{L/2-1}\cos^{2p}\bigg(\frac{\pi m}{L}\bigg) = \frac{1}{2^{2p}} \sum_{m = 0}^{L/2 - 1} \sum_{k = 0}^{2p} {2p \choose k}\textnormal{Re}\bigg(e^{i2\pi m (p - k) / L}\bigg) =  \frac{1}{2^{2p}}  \sum_{k = 0}^{2p} {2p \choose k} \sum_{m = 0}^{L/2 - 1}\textnormal{Re}\bigg(e^{i2\pi m (p - k) / L}\bigg).
\]
Noting that
\[
\sum_{n = 0}^{L/2 - 1}\textnormal{Re}\bigg(e^{i2\pi n k / L}\bigg) = \begin{cases}
N/2 & \text{ if } k \in\{ 0, \pm L, \pm 2L \dots\}, \\
0 & \text{ if }k \text{ is odd}, \\
1 & \text{ if }k \text{ is even and }k\notin \{ 0, \pm L, \pm 2L \dots\},
\end{cases}
\]
we obtain
\begin{align}\label{ref:app_H1_trig_id_2}
\sum_{n = 0}^{L/2 - 1}\cos^{2p}\bigg(\frac{\pi n}{L}\bigg) = \frac{L}{2^{2p + 1}} {2p \choose p} + \frac{L}{2^{2p}} \sum_{1 \leq k \leq p / L} {2p \choose p + k} + \frac{1}{2}.
\end{align}
Part (a) of the lemma then follows from Eqs.~\ref{ref:app_H1_trig_id_1} and \ref{ref:app_H1_trig_id_2}. \\ \ \\
\noindent (b) We note that
\[
\int_0^\pi \sin^{2p}(\theta) d\theta = 2\int_0^{\pi / 2} \cos^{2p}(\theta) d\theta = \frac{2}{2^{2p}} \sum_{m = 0}^{2p}{2p \choose m} \int_0^{\pi / 2} \textnormal{Re}\bigg(e^{i2(p - m) \theta}\bigg)d\theta.
\]
Noting that for $q \in \mathbb{Z}$,
\[
\int_0^{\pi / 2} \textnormal{Re}\big( e^{2i q \theta}\big) d\theta = \begin{cases}
\pi/2 & \text{ if } q = 0, \\
0 & \text{ otherwise},
\end{cases}
\]
and thus
\[
\int_0^{\pi} \sin^{2p}(\theta) d\theta = \frac{\pi}{2^{2p}} {2p \choose p}.
\]
$\hfill \square$
\begin{lemma}\label{lemma:lower_bound_value}
For $0\leq \alpha \leq \pi / 4$, 
\[
\frac{1}{L} \sum_{m = -L}^{L - 1} \cos\bigg(\alpha\sin \bigg(\frac{\pi m}{L}\bigg)\bigg) \geq \sqrt{2}.
\]
\end{lemma}
\noindent\emph{Proof}: This follows from the simple observation that since $\abs{\sin(\pi m / L)}\leq 1 $, and the cosine function is decreasing in the interval $[-\pi / 2, \pi/2]$, for $0 \leq \alpha \leq \pi/4$, $\cos(\alpha \sin(\pi m / n)) \geq 1 / \sqrt{2}$ for all $m \in \{-L, -L + 1 \dots L-1\}$. $\hfill \square$
\begin{lemma}\label{lemma:lower_bound_main} If $L \in \mathbb{N}$ is even and for $0 \leq \alpha < 2\sqrt{L}$,
\[
    \frac{\pi}{L} \sum_{m = -L}^{L - 1}\cos\bigg(\alpha \sin\bigg(\frac{\pi m}{L}\bigg)\bigg) - \int_{-\pi}^{\pi} \cos(\alpha \sin \theta) d\theta \geq  4\pi \frac{(\alpha / 2)^{2L}}{(2L)!} \bigg(1 - \frac{\alpha^2}{4(2L + 1)}\bigg)
\]
\end{lemma}
\noindent\emph{Proof}: We begin by using the Taylor expansion of the cosine function to obtain that
\[
\frac{\pi}{L} \sum_{m = -L}^{L - 1}\cos\bigg(\alpha \sin\bigg(\frac{\pi m}{L}\bigg)\bigg) - \int_{-\pi}^{\pi} \cos(\alpha \sin \theta) d\theta = \sum_{p = 1}^\infty (-1)^p \frac{\alpha^{2p}}{(2p)!} \bigg(\frac{\pi}{L}\sum_{m = -L}^{L - 1} \sin^{2p}\bigg(\frac{\pi m}{L}\bigg) - \int_{-\pi}^{\pi} \sin^{2p}(\theta) d\theta \bigg).
\]
Using lemma \ref{lemma:trig_sums_app_h}, we then obtain that
\[
\frac{\pi}{L} \sum_{m = -L}^{L - 1}\cos\bigg(\alpha \sin\bigg(\frac{\pi m}{L}\bigg)\bigg) - \int_{-\pi}^{\pi} \cos(\alpha \sin \theta) d\theta = \sum_{p = 1}^\infty  T_p(\alpha),
\]
where
\[
T_p(\alpha) =4\pi(-1)^p \frac{(\alpha/2)^{2p}} {(2p)!} \sum_{1\leq k \leq p /L} {2p \choose p + kL}.
\]
Furthermore, assuming that $L$ is even, we note that for even $p$,
\begin{align*}
    &T_p(\alpha) +  T_{p + 1}(\alpha) =4\pi \frac{(\alpha / 2)^{2p}}{(2p)!} \sum_{1 \leq k \leq p / L}{2p \choose p + kL} - 4\pi \frac{(\alpha / 2)^{2p + 2}}{(2p + 2)!}\sum_{1 \leq k \leq (p + 1) / L}{2p + 2 \choose p + 1 + kL},\\
    &\qquad =4\pi \frac{(\alpha / 2)^{2p}}{(2p)!} \bigg[\sum_{1 \leq k \leq p / L}{2p \choose p + kL} - \frac{(\alpha / 2)^2}{(2p + 2)(2p + 1)}{2p + 2 \choose p + 1 + kL} \bigg] ,\\
    &\qquad=4\pi \frac{(\alpha / 2)^{2p}}{(2p)!} \bigg[\sum_{1 \leq k \leq p / L}{2p \choose p + kL}\bigg(1 - \frac{\alpha^2}{4((p + 1)^2 - k^2 L^2)} \bigg)\bigg]
\end{align*}
Noting that if $1 \leq k \leq p/L$, $(p + 1)^2 -k^2 L^2 \geq L $, and for $\alpha \leq 2 \sqrt{L}$, $1 - \alpha^2 / 4((p + 1)^2 - k^2 L^2) \geq 0$, and we conclude that $T_p(\alpha) + T_{p + 1}(\alpha) \geq 0$ for all even $p$.

Next, we note that $T_p(\alpha) = 0$ for $1 \leq p < L$. Since $L$ is even and, as established above, $T_{p}(\alpha) + T_{p + 1}(\alpha) \geq 0$ for even $p$, we obtain the lower bound
\[
\frac{\pi}{n} \sum_{m = -L}^{L - 1}\cos\bigg(\alpha \sin\bigg(\frac{\pi m}{L}\bigg)\bigg) - \int_{-\pi}^{\pi} \cos(\alpha \sin \theta) d\theta \geq T_L(\alpha) + T_{L + 1}(\alpha) = 4\pi \frac{(\alpha / 2)^{2L}}{(2L)!} \bigg(1 - \frac{\alpha^2}{4(2L + 1)}\bigg) \geq 2\pi \frac{(\alpha / 2)^{2L}}{(2L)!},
\]
which establishes the lemma. $\hfill \square$
\begin{proposition}\label{prop:lower_bound_dynamics}
There exists a nearest-neighbour Hamiltonian $H_n$ on $n=(2L)^d$ fermions (or spins) arranged on a $d-$dimensional lattice, a single-site observable $O_n$ and an initial state $\rho_n$ and a time $t$ independent of $n$ such that $\lim_{n\to \infty}\textnormal{Tr}(O_n e^{-iH_n t} \rho_n e^{iH_n t})$ exists and
\[
\bigabs{\textnormal{Tr}(O_n e^{-iH_n t} \rho_n e^{iH_n t}) - \lim_{n \to \infty}\textnormal{Tr}(O_n e^{-iH_n t} \rho_n e^{iH_n t})} \leq \varepsilon \implies n \geq \Omega\bigg(\frac{\log^d(\Theta(\varepsilon^{-1}))}{\log^d \log(\Theta(\varepsilon^{-1}))} \bigg).
\]
\end{proposition}
\noindent\emph{Proof}: We can obtain an expression for $\bra{\psi_0} e^{iH_n t} O e^{-iH_n t} \ket{\psi_0} = \abs{\bra{\psi_0} e^{-iH_n t} \ket{\psi_0}}^2$ by noting that $\ket{\psi_0} = \frac{1}{(2L)^{d/2}} \sum_{k \in \mathfrak{L}_n^d}\tilde{a}_k^\dagger \ket{\text{vac}}$, and thus,
\[
\bra{\psi_0} e^{-iH_n t}\ket{\psi_0} = \frac{1}{(2L)^d} \sum_{k \in \mathfrak{L}_n^d} e^{-i\omega_k t } = \frac{1}{(2L)^d}\bigg(\sum_{k = -L}^{L - 1} e^{-2it \sin(\pi k /L)}\bigg)^{d} = \frac{1}{(2L)^d} \bigg(\sum_{k = -L}^{L - 1} \cos\big({2t \sin(\pi k /L)}\big)\bigg)^{d}
\]
It also follows, from this expression, that
\[
\lim_{n \to \infty} \bra{\psi_0} e^{-iH_n t}\ket{\psi_0} = \frac{1}{(2\pi)^d}\bigg(\int_{-\pi}^\pi \cos(2t \sin \theta) d\theta\bigg)^d
\]
Consider now the error $\mathcal{E}$ defined in Eq.~\ref{eq:app_H_dynamics_error_def} --- from the analytical expression for $\bra{\psi_0} e^{-iH_n t} \ket{\psi_0}$ given above, we obtain that
\[
\mathcal{E} = \frac{1}{(2\pi)^{2d}}\bigabs{ \bigg(\int_{-\pi}^\pi \cos(2t \sin \theta) d\theta\bigg)^{2d} - \bigg(\frac{\pi}{n}\sum_{k = -n}^{n - 1} \cos\big({2t \sin(\pi k /n)}\big)\bigg)^{2d} }.
\]
Note that for $x, y \geq a > 0$ and $k \in \mathbb{N}$, $\abs{x^{k} - y^{k}} = \abs{x - y}(x^{k - 1} + x^{k - 2}y + x^{k - 3}y^2 +\dots y^{k -1}) \geq k a^{k - 1}\abs{x - y}$. From this fact and from lemma \ref{lemma:lower_bound_value}, it follows that
\[
\mathcal{E}_n \geq \frac{(2d - 1)2^d}{(2\pi)^{2d}}\bigabs{\int_{-\pi}^\pi \cos(2t \sin \theta) d\theta - \frac{\pi}{n}\sum_{k = -n}^{n - 1} \cos(2t \sin(\pi k / n))}.
\]
Finally, using lemma \ref{lemma:lower_bound_main}, we obtain that
\[
\mathcal{E}_n \geq \frac{(2d - 1)2^{d}}{(2\pi)^{2d - 1}} \frac{(\alpha / 2)^{2L}}{(2L)!} \geq \frac{(2d - 1)2^d}{(2\pi)^{d-1}(2L + 1)^{2L + 1}}  \bigg(\frac{e\alpha}{2}\bigg)^{2L}.
\]
Now, if $\mathcal{E}_n \leq \varepsilon$ for small $\varepsilon$, then
\[
\bigg(\frac{2L + 1}{e\alpha/2} \bigg)^{\frac{2L + 1}{e\alpha/2}} \geq \bigg(\frac{2^{d+1}(2d - 1)}{e\alpha (2\pi)^{d-1} \varepsilon}\bigg)^{2/e\alpha} \implies (2L + 1) \geq \frac{e\alpha}{2} e^{W_0\big( \frac{2}{e\alpha}\log \big(\frac{2^{d+1}(2d - 1)}{e\alpha (2\pi)^{d-1} \varepsilon}\big) \big)},
\]
where $W_0(\cdot)$ is the Lambert W function. Noting that as $x \to \infty$, $W_0(x) \to \log(x / \log (x))$, we obtain that for $\varepsilon \to 0$ and for a constant $\alpha$, $L \geq \Omega(\log(\Theta(\varepsilon^{-1})) / \log \log(\Theta(\varepsilon^{-1})))$. Since $n = (2L)^d$ we obtain the proposition.$\hfill \square$
%%%%%%%%%%%%%%%%%%%%%%%%%%%%%%%%%%%%%%%%%%%%%%%%%%%%%%%%%%%%%%%%%%%%%%%%%%%
\subsection{Ground state}\label{app:lower_bound_gs}
The AKLT Hamiltonian is an example of a Hamiltonian where the convergence of a local observable to the thermodynamic limit is logarithmic and tight \cite{affleck2004rigorous}. We can also easily construct a higher dimensional model using the AKLT model that satisfies a logarithmic convergence to the thermodynamic limit. Consider a $d$-dimensional lattice $\mathfrak{L}^d = \{0, 1, 2 \dots L - 1\}^d$ with $n = L^d$ sites. At each site, we have $d$ spin 1 systems (i.e.~the local Hilbert space at each site is $(\mathbb{C}^{3})^{\otimes d}$). We consider the following translationally invariant Hamiltonian, with periodic boundary conditions, on this lattice of qudits ---
\begin{align}\label{eq:aklt_d_dimensions}
H = \sum_{m = 1}^d H_m \text{ where } H_m = \sum_{x \in \mathfrak{L}^d } \bigg(\vec{S}_{x}^m\cdot \vec{S}_{x + e_m}^m + \frac{1}{3}\big(\vec{S}_{x}^m\cdot \vec{S}_{x + e_m}^m \big)^2\bigg),
\end{align}
where $\vec{S}_x^m$ is the vector of spin $1$ operators at site $x$ and on the $m^\text{th}$ spin 1 system at this site, $e_m$ is the unit translation vector on the lattice along the $m^\text{th}$ direction. Note that $H_m$ is the sum of 1D AKLT models on the $m^\text{th}$ spin in every row of qudits along the $m^\text{th}$ direction.
\begin{lemma}\label{lemma:AKLT}
Consider $\ket{G_{L,\textnormal{AKLT}}} = \sum_{i \in \{-1, 0, 1\}^L} \textnormal{Tr}[A^{i_0} A^{i_1} \dots A^{i_{L-1}}]\ket{i_0, i_1 \dots i_{L - 1}} $ with $A^{\pm 1} = \pm\sqrt{2/3}\sigma_{\pm}$, $A^0 = -\sqrt{1/3} \ \sigma_z$, to be the ground state of the 1D AKLT model with periodic boundary conditions on $L$ spins and the single site observable $O_L = \ket{0}\bra{0} \otimes I^{\otimes(L - 1)}$ then
\[
\bra{G_{L, \textnormal{AKLT}}} O_L \ket{G_{L, \textnormal{AKLT}}} = \frac{1}{3}\bigg(1 - \frac{(-1)^{L-1}}{3^{L-1}}\bigg).
\]
\end{lemma}
\noindent\emph{Proof}: We can explicitly compute $\bra{G_{L, \textnormal{AKLT}}} O_L \ket{G_{L, \textnormal{AKLT}}}$ to obtain
\[
\bra{G_{L, \textnormal{AKLT}}} O_L \ket{G_{L, \textnormal{AKLT}}} =  \frac{1}{3} \sum_{i, j \in \{0, 1\}} (-1)^{i + j} \bra{i} \text{T}^{L - 1}(\ket{i}\bra{j}) \ket{j},
\]
where $\text{T}$ is the transfer tensor corresponding to the MPS $\ket{G_{L, \text{AKLT}}}$, which is a single qubit channel given by
\[
\text{T}(X) = \frac{2}{3} \sigma_+ X \sigma_- + \frac{2}{3} \sigma_- X \sigma_+ + \frac{1}{3} \sigma_z X \sigma_z.
\]
We can note that $\text{T}^k(I) = I$, $\text{T}^k(\sigma_z) = (-1/3)^k \sigma_z, \text{T}^{k}(\ket{0}\bra{1}) = (-1/3)^k \ket{0}\bra{1}$, $\text{T}^{k}(\ket{1}\bra{0}) = (-1/3)^k \ket{1}\bra{0}$. We then obtain that
\begin{align*}
&\bra{0}\text{T}^{k}(\ket{0}\bra{1}) \ket{1} = \bra{1}\text{T}^{k}(\ket{1}\bra{0}) \ket{0} = \frac{(-1)^k}{3^k}, \\
&\bra{0}\text{T}^k(\ket{0}\bra{0})\ket{0} = \frac{1}{2}\bra{0}\text{T}^k(I)\ket{0} + \frac{1}{2}\bra{0}\text{T}^k(\sigma_z)\ket{0} = \frac{1}{2} + \frac{1}{2}\frac{(-1)^k}{3^k}, \\
&\bra{1}\text{T}^k(\ket{1}\bra{1})\ket{1} =  \frac{1}{2}\bra{1}\text{T}^k(I)\ket{1} - \frac{1}{2}\bra{1}\text{T}^k(\sigma_z)\ket{1} =  \frac{1}{2} -  \frac{1}{2}\frac{(-1)^k}{3^k}.
\end{align*}
Thus,
\[
\bra{G_{L, \textnormal{AKLT}}} O_L \ket{G_{L, \textnormal{AKLT}}} = \frac{1}{3}\bigg( 1 -\frac{(-1)^{L-1}}{3^{L - 1}}\bigg),
\]
which establishes the lemma. \hfill $\square$
\begin{proposition}\label{prop:lower_bound_gapped}
    There exists a nearest neighbour gapped Hamiltonian $H_n$ on $n$ qudits arranged on a d-dimensional lattice with a unique ground state $\ket{G_n}$, a single site observable $O_n$ such that 
    \[
    \bigabs{\bra{G_n} O_n \ket{G_n} - \lim_{n \to \infty} \bra{G_n} O_n \ket{G_n}} \leq \varepsilon \implies n \geq \Omega\big(\log^d(\Theta(\varepsilon^{-1}))\big).
    \]
\end{proposition}
\noindent\emph{Proof}: The Hamiltonian is provided in Eq.~\ref{eq:aklt_d_dimensions}. Since the AKLT Hamiltonian is known to be gapped, this Hamiltonian is gapped as well. We choose the observable $O_n = (\ket{0}\bra{0})^{\otimes d} \otimes I^{\otimes (n - 1)}$ i.e.~the projector $\ket{0}\bra{0}$ on all the spin 1 systems at the first site of the lattice, and identity on the remaining sites. Using lemma $\ref{lemma:AKLT}$, we obtain that
\[
\bra{G_n} O_n \ket{G_n} = \frac{1}{3}\bigg(1 - \frac{1}{3^{n^{1/d}-1}}\bigg)^d,
\]
where, for simplicity, we assume that $n^{1/d}$ is odd.  We then obtain that
\[
 \bigabs{\bra{G_n} O_n \ket{G_n} - \lim_{n \to \infty} \bra{G_n} O_n \ket{G_n}} = \frac{1}{3}\bigg(1 - \bigg(1 - \frac{1}{3^{n^{1/d}-1}}\bigg)^d \bigg) \leq \varepsilon \implies n \geq \Omega \big(\log^d(\Theta(\varepsilon^{-1}))\big),
\]
which establishes the proposition. $\hfill \square$

\subsection{Fixed points}\label{app:lower_bound_fp}
Glauber dynamics for classical Ising models are examples of Lindblad evolutions which satisfy rapid mixing, and the convergence of local observables to the thermodynamic limit is logarithmic and tight. To construct an explicit example, we begin with the following result, which follows directly from the transfer matrix method.
\begin{lemma}\label{lemma:gibbs_tmatrix}
    Consider $L$ spins on 1D lattice in the state $\rho_L = e^{-\beta H_L}/ \textnormal{Tr}[e^{-\beta H}]$ where $H_L = -\sum_{i = 1}^{L - 1}Z_i Z_{i + 1}$. Let $\langle Z_1 Z_2\rangle_L : = \text{Tr}(Z_1 Z_2 \rho_L)$, then 
    \[
     \frac{\tanh^{L - 2}\beta}{\sqrt{2}\sinh \beta \cosh \beta} \leq \bigabs{\langle Z_1 Z_2\rangle_L  - \lim_{L \to \infty} \langle Z_1 Z_2\rangle_L } \leq \frac{\sqrt{2} \tanh^{L - 2}\beta}{\sinh \beta \cosh \beta}.
    \]
\end{lemma}
\begin{proof}
    The proof of this proposition follows by explicit computation using the transfer matrix method \cite{schultz1964two}. We note that
    \[
    \text{Tr}(e^{-\beta H}) = \sum_{\sigma \in \{-1, 1\}^L} e^{\beta \sigma_1 \sigma_2}e^{\beta \sigma_2 \sigma_3}\dots e^{\beta \sigma_{L - 1} \sigma_L} = \textnormal{Tr}\big(\text{T}^{L - 1}\big),
    \]
    and
        \[
    \text{Tr}(Z_1 Z_2 e^{-\beta H}) = \sum_{\sigma \in \{-1, 1\}^L} \sigma_1 \sigma_2 e^{\beta \sigma_1 \sigma_2}e^{\beta \sigma_2 \sigma_3}\dots e^{\beta \sigma_{L - 1} \sigma_L} \textnormal{Tr} = \big(\text{R}\text{T}^{L - 2}\big),
    \]
    where
    \[\text{T} = \begin{bmatrix}
        e^{\beta} & e^{-\beta} \\
        e^{-\beta} & e^{\beta}
    \end{bmatrix} \text{ and }\text{R} = \begin{bmatrix}
        e^{\beta} & -e^{-\beta} \\
        -e^{-\beta} & e^{\beta}
    \end{bmatrix}.
    \]
    Therefore, we obtain an explicit expression for $\langle Z_1 Z_2 \rangle_L$:
    \[
    \langle Z_1 Z_2 \rangle_L = \frac{\text{Tr}\big(\text{R}\text{T}^{L - 2}\big)}{\text{Tr}\big(\text{T}^{L - 1}\big)} = \frac{\sqrt{2}\big(\sinh \beta \cosh^{L - 2} \beta + \cosh \beta \sinh^{L - 2}\beta\big)}{ \cosh^{L - 1} \beta +  \sinh^{L - 1}\beta}.
    \]
    We note that for any $\beta > 0$, $\cosh \beta > \sinh \beta$ and therefore $\lim_{L\to \infty} \langle Z_1 Z_2 \rangle_L = \sqrt{2}\tanh \beta$. Furthermore, from the explicit expression for $\langle Z_1 Z_2 \rangle_L$, we obtain that
    \[
    \bigabs{\langle Z_1 Z_2 \rangle_L - \lim_{L \to \infty}\langle Z_1 Z_2 \rangle_L} = \frac{\sqrt{2}}{\sinh \beta \cosh \beta}\bigabs{\frac{\sinh^{L - 2}\beta}{\sinh^{L-2}\beta + \cosh^{L - 2}\beta}}.
    \]
    From this explicit expression and observing that $\cosh^{L - 2} \beta \leq \cosh^{L - 2} \beta + \sinh^{L - 2} \beta \leq 2\cosh^{L - 2} \beta$, we obtain the lemma statement.
\end{proof}

\begin{proposition} \label{prop:log_conv_fixed_point}
    There exists a nearest neighbour Lindbladian $\mathcal{L}_n$ on $n$ qudits arranged on a d-dimensional lattice with a unique fixed point ${\sigma_n}$ and which satisfies rapid mixing, and a local observable $O_n$ such that 
    \[
    \bigabs{\textnormal{Tr}{\big(\sigma_n O_n\big)} - \lim_{n \to \infty} \textnormal{Tr}{\big(\sigma_n O_n\big)}} \leq \varepsilon \implies n \geq \Omega\big(\log^d(\Theta(\varepsilon^{-1}))\big).
    \]
\end{proposition}
\begin{proof}
    Let us first consider the 1D setting --- a spatially local rapidly mixing Lindbladian $n$ spins, $\mathcal{L}^{\text{(1)}}_n = \sum_{\alpha = 0}^{n - 2}\mathcal{L}_{\alpha, \alpha + 1; n}^{(\text{1})}$ where $\mathcal{L}_{\alpha, \alpha + 1; n}^{(1)}$ acts on $\alpha^\text{th}, (\alpha + 1)^\text{th}$ spins, and a local observable that proves the proposition can be constructed by a quantum encoding of classical Glauber dynamics preparing the Gibbs state (at any non-zero inverse temperature $\beta$) of the classical Ising model $H = -\sum_{i = 1}^{n} Z_i Z_{i + 1}$. This construction is explicitly outlined in Ref.~\cite{cubitt2015stability} and shown to be rapidly mixing \emph{as a quantum evolution} as a consequence of log Sobolev inequalities that have been established for the classical Glauber dynamics \cite{martinelli1999lectures}. Furthermore, from lemma \ref{lemma:gibbs_tmatrix}, it follows that the local observable $Z_0 Z_1$ satisfies the convergence condition in the proposition statement.
    
    Similar to the example of the gapped ground state, the 1D example can then be used to construct the $d$-dimensional example. We consider spins arranged on the $d-$dimensional lattice $\mathbb{Z}_L^d$ where $L = n^{1/d}$ --- at every lattice site, we will have $d$ spins (i.e.~the local Hilbert space is $\big(\mathbb{C}^2\big)^{\otimes d}$). We will consider the nearest-neighbour Lindbladian formed by implementing the Lindbladian
    \[
    \mathcal{L}^{(d)}_n =\sum_{m = 0}^{d - 1} \mathcal{L}_{m; n}^{(d)} \text{ where } \mathcal{L}_{m; n}^{(d)} = \sum_{\alpha \in \{0, 2 \dots L - 2\}^d} \mathcal{L}_{\alpha_m, \alpha_m + 1}^{(1)}\otimes \textnormal{id},
    \]
    where $\mathcal{L}^{(1)}_{\alpha_m, \alpha_m + 1}$ acts on the $m^\text{th}$ spin at site $\alpha$ and the $m^\text{th}$ spin at the site displaced by one unit from $\alpha$ along the $m^\text{th}$ direction. We point out that since $\mathcal{L}^{(1)}_n$ satisfies rapid mixing, by construction, so does $\mathcal{L}_n^{(d)}$. Consider now the observable $O_n$
    \[
    O_n = \prod_{m = 0}^{d- 1}Z_{0, m} Z_{e_m, m},
    \]
    where $Z_{\alpha, m}$ is the $Z$ operator acting on the $m^\text{th}$ spin at site $\alpha$, amd $e_m$ is the unit vector along the $m^\text{th}$ lattice direction. Using lemma \ref{lemma:gibbs_tmatrix}, it can be seen that 
    \[
    \textnormal{Tr}(\sigma_n O_n) = \big(\langle Z_1 Z_2 \rangle_{n^{1/d}}\big)^d,
    \]
    where $\langle Z_1 Z_2 \rangle_L$ is defined in lemma \ref{lemma:gibbs_tmatrix}.
    Note that $\lim_{n \to \infty} \text{Tr}(\sigma_n O_n) = (\sqrt{2}\tanh \beta)^d$. We also obtain from lemma \ref{lemma:gibbs_tmatrix} that for positive $\beta$ which is $\Theta(1)$ ,
    \[
    \langle Z_1 Z_2 \rangle_{n^{1/d}} \geq \sqrt{2}\tanh(\beta) - O(e^{-\Theta(n^{1/d})}) \text{ and } \bigabs{\langle Z_1 Z_2 \rangle_{n^{1/d}} - \sqrt{2}\tanh \beta} \geq \Omega(e^{-\Theta(n^{1/d})}).
    \]
    We now obtain that
    \begin{align*}
    &\bigabs{\textnormal{Tr}(\sigma_n O_n) - \lim_{n\to \infty}\text{Tr}(\sigma_n O_n)} = \bigabs{\langle Z_1 Z_2 \rangle_{n^{1/d}} - \sqrt{2}\tanh \beta}\times \sum_{j = 0}^{d - 1}\langle Z_1 Z_2 \rangle_{n^{1/d}}^{j} \big(\sqrt{2}\tanh \beta\big)^{d - 1 - j}, \\
    &\qquad \geq \Omega(e^{-\Theta(n^{1/d})}) \times \big(\sqrt{2}\tanh(\beta) - O(e^{-\Theta(n^{1/d})})\big) \geq \Omega(e^{-\Theta(n^{1/d})}),
    \end{align*}
    from which the proposition follows.
\end{proof}
\end{document}